\documentclass[final,1p]{elsarticle}

\usepackage{amssymb}
\usepackage{graphicx}
\usepackage{url}
\usepackage{stmaryrd}
\usepackage{bussproofs}
\usepackage{amsmath,amssymb,amsfonts}
\usepackage{bussproofs}
\usepackage{diagrams}
\usepackage{hyperref}
\usepackage{wasysym}
\usepackage{enumerate}
\usepackage{amssymb}
\usepackage{bbm}

\newtheorem{theorem}{Theorem}[section]
\newtheorem{lemma}[theorem]{Lemma}
\newtheorem{proposition}[theorem]{Proposition}

\newcommand{\C}{{\mathcal{C}}}
\newcommand{\W}{{\mathcal{W}}}

\newenvironment{proof}[1][Proof]{\begin{trivlist}
\item[\hskip \labelsep {\bfseries #1}]}{\end{trivlist}}
\newenvironment{definition}[1][Definition]{\begin{trivlist}
\item[\hskip \labelsep {\bfseries #1}]}{\end{trivlist}}

\newcommand{\prom}{\mathsf{prom}}

\newcommand{\hd}{\mathsf{hd}}
\newcommand{\tl}{\mathsf{tl}}
\newcommand{\ana}{\mathsf{ana}}
\newcommand{\inn}{\mathsf{in}}

\newcommand{\parr}{\bindnasrepma}
\newcommand{\mul}{\mathsf{mul}}

\newcommand{\p}{\mathsf{P}}
\newcommand{\dst}{\mathsf{dst}}
\mathchardef\mhyphen="2D

\newcommand{\op}{\mathsf{op}}

\newcommand{\set}{\mathsf{set}}
\newcommand{\reify}{\mathsf{reify}}
\newcommand{\dist}{\mathsf{dist}}
\newcommand{\wk}{\mathsf{wk}}
\newcommand{\seqo}{\mathsf{Seq}}
\newcommand{\prf}{\mathsf{Prf}}

\newcommand{\tot}{\mathsf{Tot}}

\newcommand{\fprod}{\mathsf{prod}}
\newcommand{\id}{\mathsf{id}}
\newcommand{\sym}{\mathsf{sym}}
\newcommand{\ri}{\mathsf{runit}_\otimes}
\newcommand{\passoc}{\mathsf{pasc}}
\newcommand{\assoc}{\mathsf{assoc}}
\newcommand{\cut}{\mathsf{cut}}
\newcommand{\af}{\mathsf{af}}
\newcommand{\psym}{\mathsf{psym}}
\newcommand{\app}{\mathsf{app}}

\newcommand{\lfe}{\mathsf{lfe}}

\newcommand{\match}{\mathsf{ma}}
\newcommand{\FamInj}{\mathsf{FamInj}}
\newcommand{\li}{\mathsf{lunit}_\otimes}
\newcommand{\unit}{\mathsf{unit}}

\newcommand{\dec}{\mathsf{dec}}

\newcommand{\ini}{\mathsf{in}_i}

\newcommand{\com}{\mathbf{com}}
\newcommand{\con}{\mathsf{con}}

\newcommand{\ar}{\mathsf{ar}}

\newcommand{\der}{\mathsf{der}}
\newcommand{\term}{\mathsf{t}}

\newcommand{\var}{\mathbf{var}}

\newcommand{\tru}{\mathtt{tt}}
\newcommand{\ff}{\mathtt{ff}}

\journal{Annals of Pure and Applied Logic}

\begin{document}

\begin{frontmatter}

\title{Imperative Programs as Proofs via Game Semantics
\footnote{NOTICE: this is the author’s version of a work that was accepted for publication in Annals of Pure and Applied Logic. Changes resulting from the publishing process, such as peer review, editing, corrections, structural formatting, and other quality control mechanisms may not be reflected in this document. Changes may have been made to this work since it was submitted for publication. A definitive version was subsequently published in Annals of Pure and Applied Logic, volume 164, issue 11 at \url{http://dx.doi.org/10.1016/j.apal.2013.05.005}}
}

\author[label2]{Martin Churchill}
\author[label1]{Jim Laird}
\author[label1]{Guy McCusker}

\address[label1]{University of Bath}
\address[label2]{Swansea University}

\begin{abstract}
  Game semantics extends the Curry-Howard isomorphism to a three-way
  correspondence: proofs, programs, strategies. But the universe of
  strategies goes beyond intuitionistic logics and lambda calculus, to
  capture stateful programs. In this paper we describe a logical
  counterpart to this extension, in which proofs denote such
  strategies. The system is expressive: it contains all of the
  connectives of Intuitionistic Linear Logic, and first-order
  quantification. Use of Laird's \emph{sequoid} operator allows proofs
  with imperative behaviour to be expressed. Thus, we can embed
  first-order Intuitionistic Linear Logic into this system, Polarized
  Linear Logic, and an imperative total programming language.

  The proof system has a tight connection with a simple game model,
  where games are forests of plays. Formulas are modelled as games,
  and proofs as history-sensitive winning strategies. We provide a
  strong \emph{full completeness} result with respect to
  this model: each finitary strategy is the denotation of a unique
  analytic (cut-free) proof. Infinite strategies correspond to
  analytic proofs that are infinitely deep. Thus, we can normalise
  proofs, via the semantics.

\end{abstract}


\end{frontmatter}

\section{Introduction}

The Curry-Howard isomorphism between proofs in intuitionistic logics
and functional programs is a powerful theoretical and practical
principle for specifying and reasoning about programs.  Game semantics
provides a third axis to this correspondence: each proof/program at a
given type denotes  a strategy for the associated game, and typically
a \emph{full completeness} result establishes that this correspondence
is also an isomorphism \cite{AJ_MLL}. However, in languages with
side-effects such as mutable state it is evident that there are many
programs which do not correspond to intuitionistic proofs. Game
semantics has achieved notable success in providing models of such
programs  \cite{AMc_LSS,AHM_GR,Lai_FLC}, in which they typically denote
``history-sensitive'' strategies --- strategies which may break the
constraints of innocence \cite{HO_PCF} or history-freeness
\cite{AJ_MLL} imposed in fully complete models of intuitionistic or
linear logic. The full completeness of these models means there is a
precise correspondence between programs and history-sensitive
strategies, which raises the question: is there a logic to flesh out
the proofs/imperative programs/history-sensitive strategies
correspondence? 

In this paper we present a first-order logic, \textsf{WS1}, and a
games model for it in which proofs denote history-sensitive
strategies. Thus total imperative programs correspond, via the game
semantics, to proofs in \textsf{WS1}. Moreover, because \textsf{WS1} is
more expressive than the typing system for a typical programming
language, it can express finer  behavioural properties of strategies. In
particular, we can embed  first-order intuitionistic logic with
equality, Polarized Linear Logic, and a finitary imperative language
with ground store, coroutines and some infinite data structures. We
also take first steps towards answering some of the questions posed by
the logic and its semantics: Are there any formulas which only have
`imperative proofs', but no proofs in a traditional `functional' proof
system?  Can we use the expressivity of \textsf{WS1}
to 
specify imperative programs?


\subsection{Related Work}

The games interpretation of linear logic upon which \textsf{WS1} is
based was introduced by Blass in a seminal paper
\cite{Bla_LL}. Blass also gives instances of history sensitive
strategies which are not denotations of linear logic proofs; these do,
however, correspond to proofs in \textsf{WS1}. The particular
symmetric monoidal closed category of games underlying our semantics
has been studied extensively from both logical and programming
perspectives \cite{Cur_SS,Lam_PLL,Hy_GS}. Longley's project
to develop a programming language based on it \cite{Long_PLGM} may be
seen as complementary to our aim of understanding it from a logical
perspective.

Several logical systems have taken games or interaction as a semantic basis yielding a richer notion of meaning  than classical or intuitionistic truth, including Ludics \cite{Gir_LS} and Computability Logic \cite{JCL}. The latter also provides an analysis of Blass's examples, suggesting further connections with our logic, although there is a difference of emphasis: the research described here is focused on investigating the structural properties of the games model on which it is based.


Perhaps closest in spirit to our work is tensorial logic, introduced
in~\cite{MT_RM}. Like \textsf{WS1}, tensorial logic is directly
inspired by the structure of strategies in game semantics, and
in~\cite{MelliesPA:gamssd}, Melli\`es demonstrates a tight correspondence between
the logic and categories of innocent strategies on dialogue games. Our
focus in this paper is somewhat different, because we are primarily
concerned with the history-sensitive behaviour characteristic of (game
semantics of) imperative programs, rather than the purely functional
programs that denote innocent strategies.

In \cite{CC_CGC} a proof theory for Conway games is presented, where formulas are the game trees themselves. In \cite{G_LLB}, the $\lambda \overline{\lambda}$-calculus is presented, where individual moves of game semantics are represented by variables and binders. Both settings deal with history-sensitive strategies, and have dynamics corresponding to composition of strategies.

A quite different formalisation of game semantics for first order logic is given in \cite{Lau_FOL}, also with a full completeness result.

\subsection{Contribution} 

The main contribution of this paper is to present an expressive
logical system and its semantics, in which proofs correspond to
history sensitive strategies. Illustrating the expressive power of
this system, we show how proofs of intuitionistic first-order logic,
Polarized Linear Logic and imperative programming constructs 
may be embedded in it. We also demonstrate how formulas in the logic
can be used to represent some properties of imperative programs: for
example, we describe a formula for which any proof corresponds to a
well-behaved (single write) Boolean storage cell.
 
The interpretation of \textsf{WS1} includes some interesting developments of game semantics. In particular,  the exponentials are treated in a novel way: we use the fact that the  semantic exponential introduced in \cite{Hy_GS} is a final coalgebra, and reflect this explicitly in the logic in the style of \cite{Cla_FIX}. This formulation allows us to express the usual exponential introduction rules (promotion and dereliction) but also proofs that correspond to strategies on $!A$ that act differently on each interrogation, such as the reusable Boolean reference cell. Another development is the interpretation of first-order logic with equality. A proof corresponds to a family of  winning strategies --- one for each possible interpretation of the atoms determined by a standard notion of ${\mathcal{L}}$-structure --- which must be \emph{uniform} across ${\mathcal{L}}$-structures. This notion of uniformity is precisely captured by the requirement that strategies are \emph{lax natural transformations}  between the relevant functors. 

The main technical results of this paper concern the sharp
correspondence between proofs and strategies: \emph{full completeness}
results. 
We show that any bounded uniform winning strategy is the denotation of a unique (cut-free) \emph{analytic proof}. In the exponential-free fragment, where all strategies are bounded, it follows that many rules such as cut are admissible; and it allows us to normalise proofs to analytic proofs via the semantics. 
For the full logic, since the exponentials correspond to final coalgebras, proofs can be unfolded to infinitary form. Extending semantics-based normalisation to the full \textsf{WS1}, the resulting normal forms are \emph{infinitary} analytic proofs.

\section{Games and Strategies}


Our notion of game is essentially that introduced by \cite{Bla_LL},
and similar to that of \cite{AJ_MLL,Lam_SGL}, augmented with winning
conditions introduced as in \cite{Hy_GS}. We make use of the
categorical structure on games and strategies first introduced in
\cite{JoyalA:remslt}.

Informally, a game is a tree where Player and Opponent own alternate
nodes, together with a polarity specifying which protagonist owns the
starting node. A play proceeds down a particular branch, with
Opponent/Player choosing the subtree for nodes they control. A
strategy for Player specifies which choice Player should make in
response to Opponent's moves so far. The winner of a finite play is
the last protagonist to play a move. The winner of an infinite play is
specified by a winning condition for each game.


If $A$ is a set, let $A^\ast$ denote the free monoid (set of
sequences) over $A$, $A^\omega$ the set of infinite sequences over
$A$, and $\epsilon$ the empty sequence. We write $s \sqsubseteq t$ if
$s$ is a prefix of $s$, and $s \sqsubset t$ if $s$ is a strict
(finite) prefix of (possibly infinite)~$t$. If $X \subseteq A^\ast$,
write $\overline{X} = \{ s \in A^\omega : \forall t \sqsubset s, t \in
X \}$.

\begin{definition}
A \emph{game} is a tuple $(M_A, \lambda_A, b_A, P_A, W_A)$ where
\begin{itemize}
 \item $M_A$ is a set of moves
 \item $\lambda_A : M_A \rightarrow \{ O, P \}$
  \begin{itemize} \item We call $m$ an
   \emph{\textsf{O}-move} if $\lambda_A(m) = O$ and a
   \emph{\textsf{P}-move} if $\lambda_A(m) = P$. \end{itemize}
 \item $b_A \in \{ O, P \}$ specifies a starting player
  \begin{itemize}
  \item We call $s \in M_A^\ast$ \emph{alternating} if $s$ starts with
    a $b_A$-move and alternates between \textsf{O}-moves and
    \textsf{P}-moves. Write $M_A^\varoast$ for the set of such
    sequences.
  \end{itemize}
 \item $P_A \subseteq M_A^\varoast$ is a nonempty prefix-closed set of
   valid plays.
 \item $W_A \subseteq \overline{P_A}$ represents the set of infinite
   plays that are P-winning; we say an infinite play is O-winning if
   it is not P-winning.
\end{itemize}
\end{definition}
\label{natgame}

For finite plays, the last player to play a move wins: let $W_A^\ast =
W_A \cup E_A$ where $E_A$ is the set of plays that end in a P-move. We
will call a game $A$ \emph{negative} if $b_A = O$ and \emph{positive}
if $b_A = P$. We write $A, B, C, \ldots$ for arbitrary games; $L, M,
N, \ldots$ for arbitrary negative games and $P, Q, R, \ldots$ for
arbitrary positive games.


\label{seqnot}

\begin{definition} If $A$ is a game, we define its \emph{negation} by changing its polarity, and swapping its Player/Opponent labelling. Define $\neg : \{O , P\}
\rightarrow \{O , P\}$ by $\neg(O) = P$ and $\neg(P) = O$.
 $$A^\perp = (M_A, \neg
\circ \lambda_A, \neg b_A , P_A, \overline{P_A} - W_A).$$
Negation is evidently an involutive bijection between negative and positive games.
\end{definition}


\begin{definition}
A \emph{strategy} $\sigma$ for a game $(M_A, \lambda_A, b_A, P_A, W_A)$ is a subset of $P_A$ (a set of traces) satisfying:
\begin{itemize}
 \item If $sa \in \sigma$, then $\lambda_A(a) = P$
 \item If $sab \in \sigma$, then $s \in \sigma$
 \item If $sa,sb \in \sigma$, then $a = b$
 \item If $\sigma = \varnothing$ then $b_A = P$, and if $\epsilon \in
  \sigma$ then $b_A = O$.
\end{itemize}
\end{definition}

\noindent We say a strategy $\sigma$ is \emph{bounded} if $\exists k \in \mathbb{N}
. \forall s \in \sigma . |s| \leq k$; in which case we write
$\mathsf{depth}(\sigma)$ for the smallest such $k$ (the length of the
longest play in $\sigma$).

\begin{definition}
  A strategy on a game $A$ is \emph{total} if it is nonempty and
  whenever $s \in \sigma$ and $sa \in P_A$, there is some $b \in M_A$
  such that $sab \in \sigma$.  A total strategy $\sigma$ is
  \emph{winning} if whenever $s \in \overline{P_A}$ and all prefixes
  of $s$ ending in a \textsf{P}-move are in $\sigma$, then $s \in
  W_A$.
\end{definition}

\subsection{Connectives}

\label{gameconnectives}

We next describe  operations on games, which will correspond to
connectives in our logic. These come in dual pairs, determined by
involutive negation. 

First, some notation. If $X$ and $Y$ are sets, let $X + Y = \{
\mathsf{in}_1(x) : x \in X \} \cup \{ \mathsf{in}_2(y) : y \in Y
\}$. We use standard notation $[f,g]$ for copairing. If $s \in (X +
Y)^\ast$ or $s \in (X + Y)^\omega$ then $s|_i$ is the subsequence of
$s$ consisting of elements of the form $\mathsf{in}_i(z)$. If $X_1
\subseteq X^\ast$ and $Y_1 \subseteq Y^\ast$ let $X_1 \| Y_1 = \{ s
\in (X+Y)^\ast : s|_1 \in X_1 \wedge s|_2 \in Y_1 \}$. If $X_1
\subseteq X^\omega$ and $Y_1 \subseteq Y^\omega$ let $X_1 \| Y_1 = \{
s \in (X+Y)^\omega : s|_1 \in
X_1 \wedge s|_2 \in Y_1 \}$.


\paragraph{Empty Game} We define a negative game with no
moves: $$\mathbf{1} = (\emptyset, \emptyset, O, \{\epsilon\},
\emptyset).$$ There is one strategy on $\mathbf{1}$ given by $\{
\epsilon \}$, and this strategy is total (and winning, as
$\overline{P_\mathbf{1}}$ is empty).

There is one strategy,  $\emptyset$,  on the empty \emph{positive} game $\mathbf{0} = \mathbf{1}^\perp$. This strategy is not total
(intuitively, it is Player's turn to play first but he has no moves to
play).


\paragraph{One-move Game} We write $\bot$ for the negative game with a
 single move $q$ and maximal play consisting of $q$ : $$\bot = (\{q\}, \{ q \mapsto O\} , O,
\{ \epsilon, q \} , \emptyset).$$ There is a single strategy $\{
\epsilon \}$ on $\bot$;  this is not
total. 

We write $\top$ for the positive game with a single move, $\bot^\perp$. 
 There are two strategies on $\top$: $\varnothing$ (which is evidently not total) and $\{q\}$ which
is  total (and thus, trivially winning).

\paragraph{Disjoint Union} 
The negative game $L \& N$ is played over the disjoint union of the moves of $L$ and $N$: a play in this game is either a (tagged) play in $L$ or a (tagged) play in $N$. A play is $P$-winning if it is a $P$-winning play from $L$ or a $P$-winning play from $N$. Thus, on
Opponent's first move he chooses to play either in $L$ or $N$, and
thereafter play remains in that component. Formally, define  
 $$L \& N = (M_L + M_N,[\lambda_L , \lambda_N], O, P_L +^\ast P_N, \{ \inn_1^\omega(s) : s§
\in W_L \} \cup \{ \inn_2^\omega(s) : s \in W_N \})$$ 
where $X_1 +^\ast Y_1 = \{ s \in X_1 \| Y_1 : s|_1 = \epsilon \vee s|_2 = \epsilon \}$ if $X_1 \subseteq X^\ast$ and $Y_1 \subseteq Y^\ast$, and  if $s \in X_i^\ast$
(resp. $X_i^\omega$) we write $\inn_i^\ast(s)$ (resp. $\inn_i^\omega$)
for the corresponding sequence in $(X_1 + X_2)^\ast$ (resp. $(X_1 + X_2)^\omega$).
 A (winning) strategy on $L \& N$ corresponds to a
pairing of a (winning) strategy on $L$ with a (winning)
strategy on $N$ --- hence the identification of this connective with
the ``with'' of linear logic. 

Similarly, the positive game  $Q \oplus R = (Q^\perp \& R^\perp)^\perp$ corresponds to a disjoint union of plays from $Q$ and $R$ where  Player's first move constitutes   a choice to play either in $Q$ or $R$. An infinite play in $Q
\oplus R$ is P-winning if it is P-winning in the relevant component. Thus a
 winning strategy on $Q \oplus R$ corresponds to either a winning strategy on $Q$ or a  winning
strategy on $R$. 

We may  form any set-indexed conjunctions and  disjunctions in this way. Let $X$ be a
set and $\{ N_x : x \in X \}$ a family of negative games indexed
by $X$. We define the game $\prod_{x \in X} N_x$ by \begin{small}$$(\sum_{x \in X}
M_{N_x}, \inn_x(m) \mapsto \lambda_{N_x}(m),O,\{ \inn_x^\ast(s) : x
\in X , s \in P_{N_x} \},\{ \inn_x^\omega(s) : x \in X, s \in W_{N_x}
\}).$$\end{small}
If $\{ Q_x : x \in X \}$ is a family of positive games then $\bigoplus_{x \in X} Q_x
 = (\prod_{x \in X} N_x^\perp)^\perp$.

\paragraph{Symmetric Merge} If $L$ and $N$ are negative games, 
a play in  the negative game $L \otimes N$ is an
\emph{interleaving} of a play in $L$ with a play in $N$.
 Define \begin{small}$$L \otimes N = (M_L + M_N, [\lambda_L , \lambda_N], O,
(P_L \| P_N) \cap M_{L \otimes N}^\varoast, \{ s \in \overline{P_{L \otimes
  N}} : s|_1 \in W_L^\ast \wedge s|_2 \in W_N^\ast \}).$$\end{small}
The
fact that the play restricted to each component must be alternating,
and that the play overall must be alternating, ensures that only
Opponent may switch between components. This operation may be used to interpret the ``times'' of linear logic \cite{Bla_LL}.
 An infinite play in $L \otimes
N$ is P-winning if its restriction to $L$ is P-winning and its
restriction to $N$ is P-winning. 


Similarly, if $Q$ and $R$ are positive games, plays in the positive game $Q \parr R = (Q^\perp \otimes R^\perp)^\perp$
consist of interleavings of plays in $Q$ and $R$ in which Player may switch between the two
components. An infinite play in $Q \parr R$ is P-winning if its
restriction to $Q$ is P-winning or its restriction to $R$ is
P-winning.

\paragraph{Left Merge} Let $A$ be a game of polarity $a$ (positive or negative), and $N$ a negative
game. The game $A \oslash N$ has polarity $a$: a play in this game  is an interleaving of a play in $A$ with a play in $N$ \emph{such that the first move, if any, is in $A$}.  An infinite play in $A \oslash N$ is P-winning if both of its
restrictions are P-winning.    Formally, define \begin{small}$$A \oslash N = (M_A + M_N, [ \lambda_A , \lambda_N ] ,
b_A, (P_A \|_L P_N) \cap M_{A \oslash N}^\varoast, \{ s \in P_{A
  \oslash N}^\omega : s|_1 \in W_A^\ast \wedge s|_2 \in W_N^\ast
\}).$$\end{small} 
where 
$X_1 \|_L Y_1
= \{ s \in X_1 \| Y_1 : s|_1 = \epsilon \Rightarrow s|_2 = \epsilon
\}$. 
Observe that it is Opponent who switches between components: if $A$ is
negative then $A \oslash N$ consists of the plays in $A \otimes N$
which start in $A$ (or are empty). This connective on games, the
\emph{sequoid}, was introduced in 
\cite{Lai_HOS} and its properties can be used to model stateful
effects \cite{Lai_HOS,Lai_FPC}.

If $Q$ is a positive game, the game $A \lhd Q = (A^\perp \oslash Q^\perp)^\perp$  has the same polarity as $A$, and consists of interleavings of a play in $A$ and a play in $Q$,  starting in $A$ and with Player switching between components and winning an infinite play if he wins in either $A$ or $Q$.




\paragraph{Exponentials} Let $N$ be a negative game. The negative game
$\mathop{!}N$ consists of countably many 
copies of $N$, tagged with natural numbers. A play over $\mathop{!}N$ is an interleaving of plays in each copy, such that any move in $N_{i+1}$ is preceded by a move in $N_i$.
An
infinite play is winning just if it is winning in each
component. Define $$\mathop{!}N = (M_N \times \mathbb{N} ,
\lambda_N \circ \pi_1 , \{ s : \forall i . s|_i \in P_N \wedge s|_i =
\epsilon \Rightarrow s|_{i+1} = \epsilon \} , \{ s : \forall i . s|_i
\in W_N^\ast \}).$$ As with the tensor, there is an implicit switching
condition: only
Opponent can  open new copies and switch between copies. This operation may be used to interpret the ``of course'' of linear logic \cite{Hy_GS}.


Dually, if $Q$ is a positive game, $\mathop{?}Q = (!Q^\perp)^\perp$ is 
the game consisting of an infinite number of copies of $Q$, where
Player can spawn new copies and switch between them. An infinite play in  $\mathop{?}Q$
is winning if it is winning in at least one component. 


\subsubsection{Derived Connectives}
We shall also make use of the following derived operations:

\paragraph{Lifts} We can use left merge  to add a single move at the beginning of a game. If $N$ is a negative game, a play in the positive
game $$\downarrow N = \top \oslash N$$ consists of a play in $N$
prefixed by an extra P-move.  A strategy on $\downarrow N$ is either $\emptyset$ or corresponds
to a strategy on $N$. 
A winning strategy on $\downarrow N$ corresponds to a  winning strategy on
$N$. If $P$ is a positive game, a play in the negative game $$\uparrow P = \bot \lhd P$$ consists of a play in $P$ prefixed by an extra O-move. A (winning) strategy on  $\uparrow P$ corresponds to a (winning) strategy on $P$.

\paragraph{Affine Implication} If $M$ and $N$ are negative games, we may define 
 $$M \multimap N = N \lhd M^\perp.$$ A play in $M \multimap
N$ consists of a play in $N$ interleaved with a play in
$M^\perp$ (an `input
version' of $M$), starting in $N$. It is winning if its
restriction to $N$ is P-winning or its restriction to $M^\perp$ is
P-winning (i.e. its
restriction to $M$ is O-winning), agreeing with \cite{Hy_GS}.





\subsubsection{Isomorphisms of Games}

\label{isotrees}

Given two games $A$ and $B$, we say that $A$ and $B$ are \emph{forest
  isomorphic} if $b_A = b_B$ and there is a bijection from $P_A$ to $P_B$ which is monotone with respect to the prefix order, and restricts to a bijection on the $P$-winning plays. 
Some forest isomorphisms between games are given in Figure
\ref{game-isos}. Each isomorphism $M \cong N$ gives rise to winning
strategies $M \multimap N$ and $N \multimap M$, which are mutually
inverse. Thus, winning strategies on $M$ are in bijective
correspondence with winning strategies on $N$.


\begin{figure*}[ht]
\caption{Some Characteristic Isomorphisms of Games}
\centering
\label{game-isos}
\vspace{1ex}

\hrule
\hrule


\begin{tabular}{ccc}
$M \otimes N \cong N \otimes M$ & $P \parr Q \cong Q \parr P$ \\
$M \otimes (N \otimes L) \cong (M \otimes N) \otimes L$ & $P \parr (Q \parr R) \cong (P \parr Q) \parr L$ \\
$M \otimes \mathbf{1} \cong M \cong M \& \mathbf{1}$ & $P \parr \mathbf{0} \cong P \cong P \oplus \mathbf{0}$ \\
$M \& N \cong N \& M$ & $P \oplus Q \cong Q \oplus P$ \\
$M \& (N \& L) \cong (M \& N) \& L$ & $P \oplus (Q \oplus R) \cong (P \oplus Q) \oplus R$ \\
$(M \otimes N) \multimap L \cong M \multimap (N \multimap L)$ & $P \oslash (M \otimes N) \cong (P \oslash M) \oslash N$ \\
$M \multimap (N \& L) \cong (M \multimap N) \& (M \multimap L)$ & $(P \oplus Q) \oslash N \cong (P \oslash N) \oplus (Q \oslash N)$ \\
$M \multimap \mathbf{1} \cong \mathbf{1} \oslash M \cong \mathbf{1}$ & $\mathbf{0} \oslash M \cong \mathbf{0} \lhd M^\perp \cong \mathbf{0}$ \\
$M \otimes N \cong (M \oslash N) \& (N \oslash M)$ & $P \parr Q \cong (P \lhd Q) \oplus (Q \lhd P)$ \\
$(M \& N) \oslash L \cong (M \oslash L) \& (N \oslash L)$ & $(P \oplus Q) \lhd R \cong (P \lhd R) \oplus (Q \lhd R)$ \\
$\prod_{x \in X} (M_x \oslash L) \cong \prod_{x \in X} M_x \oslash L$ & $\bigoplus_{x \in X} (P_x \lhd R) \cong \bigoplus_{x \in X} P_x \lhd R$ \\
$M \oslash (N \otimes L) \cong (M \oslash N) \oslash L$ & $P \lhd (Q \parr R) \cong (P \lhd Q) \lhd R$ \\
$M \oslash \mathbf{1} \cong M$ & $P \lhd \mathbf{0} \cong P$ \\
$(M \multimap N) \multimap \bot \cong (N \multimap \bot) \oslash M$ & $\top \oslash (M \lhd Q) \cong (\top \oslash M) \lhd Q$ \\
$\bot \oslash M \cong \bot$ & $\top \lhd P \cong \top$ \\
$!N \cong N \oslash !N$ & $?P \cong P \lhd ?P$ \\
$!(N \& M) \cong !N \otimes !M$ & $?(P \oplus Q) \cong ?P \parr ?Q$ \\
\end{tabular}

\hrule
\hrule

\end{figure*}

\subsection{Imperative Objects as Strategies}

\label{bangsigma}
\label{impobjstrat}
We may model higher-order programming languages with imperative features by interpreting \emph{types} as games and \emph{programs} as strategies. (Such a semantics of a full object-oriented language, using essentially the notion of game described here, is described in  \cite{Wol_OO}.)
Here, we illustrate the capacity of our games and strategies to represent imperative objects by describing a strategy with the behaviour of a Boolean reference cell, on a  game corresponding to the type of imperative Boolean variables --- essentially the \emph{cell} strategy first described, for a different notion of game, in \cite{AMc_LSS}. (We will later
see how this strategy can be represented as a proof in our logic.)

Let  $\mathbf{B} = \bot
\lhd \top \oplus \top$ be the (negative) game of ``Boolean output'' --- this
has one initial Opponent-move \texttt{q} and two possible Player
responses, representing \texttt{True} or \texttt{False}. Let
$\mathbf{Bi}  = (\bot \& \bot) \lhd \top$ be the (negative) game of
``Boolean input'' which has two starting Opponent-moves
\texttt{in(tt)} and \texttt{in(ff)} and one possible response to this,
\texttt{ok}. The game $!(\mathbf{Bi} \& \mathbf{B})$ represents the
type of a Boolean variable --- it is a product of a \texttt{write}
method which accepts a Boolean input and a \texttt{read} method which
on interrogation produces  a Boolean output, under an exponential
which allows these methods to be used arbitrarily many times. 

The strategy \emph{cell} on this game represents a reference cell
which accepts Boolean input on the left, and returns the last value
written to it as output on the right (we assume it is initialised with
$\mathtt{ff}$).  For readability, we will 
omit the tags on the product and the exponential (since they can be
inferred). 

\[
  \begin{array}{cccccl}
   !(\mathbf{Bi} & \& & \mathbf{B}) \\ 
             &         &  {\mathtt{q}}          & \mathsf{O} \\
             &         &     {\mathtt{ff}}       & \mathsf{P} \\
     { \mathtt{in(tt)}}       &         &            & \mathsf{O} \\
      {\mathtt{ok}}       &         &            & \mathsf{P} \\
              &         &     {\mathtt{q}}      & \mathsf{O} \\
              &         &     {\mathtt{tt}}      & \mathsf{P} \\
  \end{array}
\]

In contrast with the \emph{history-free} strategies which denote
proofs of linear logic in the model of \cite{AJ_MLL}, this strategy is
\emph{history-sensitive} --- the move prescribed by the strategy
depends on the entire play so far.  It is this property which allows
the state of the object to be described implicitly, as in \cite{AMc_LSS}.

\section{The Logic \textsf{WS1}}

\subsection{Formulas of \textsf{WS1}}


The formulas of \textsf{WS1} are based on first-order linear logic,
with some additional connectives, and subject to a notion of
polarity. A \emph{first-order language} consists of:
\begin{itemize}
\item A collection of complementary pairs of predicate symbols $\phi$
  (negative) and $\overline{\phi}$ (positive), each with an arity in
  $\mathbb{N}$ such that $\ar(\phi) = \ar(\overline{\phi})$. This must
  include the binary symbol $=$ (negative), and we write $\neq$ for
  its complement
\item A collection of function symbols, each with an arity.
\end{itemize}

The negative and positive formulas of \textsf{WS1} over
${\mathcal{L}}$ are defined by the following grammar. $M,N$ range
over negative formulas and $P, Q$ over positive formulas; variables
range over some global set $\mathcal{V}$.
\begin{center}
\begin{tabular}{rlllllllllllllllll}
$M$, $N$ := & $\mathbf{1}$ &  $|$ & $\bot$ &$|$ & $\phi(\overrightarrow{s})$ & $|$\\
 & $M \otimes N$ & $|$ & $M \varoslash N$ & $|$ & $N \lhd P$ & $|$ \\
 & $\forall x . N$  & $|$ & $M \& N$ & $|$ & $!N$ \\
$P$, $Q$ := & $\mathbf{0}$ & $|$ & $\top$ & $|$ & $\overline{\phi}(\overrightarrow{s})$ & $|$ \\
 & $P \parr Q$ & $|$ & $P \lhd Q$ & $|$ & $P \varoslash N$  & $|$\\
 & $\exists x . P$ & $|$ & $P \oplus Q$ & $|$ & $?P$ 
\end{tabular}
\end{center}

\noindent Here, $s$ ranges over $\mathcal{L}$-terms, $x$ over
variables, and $\phi(\overrightarrow{s})$ over $n$-ary predicates
$\phi$ applied to a tuple of terms $\overrightarrow{s} = (s_1 , \ldots
, s_n)$.

The involutive negation operation $(\_)^\perp$ sends negative formulas
to positive ones and \emph{vice versa} by exchanging each atom, unit
or connective for its dual --- i.e. ${\mathbf{1}}$ for $\mathbf{0}$,
$\bot$ for $\top$, $\phi(\overrightarrow{x})$ for
$\overline{\phi}(\overrightarrow{x})$, $\otimes$ for $\parr$,
$\varoslash$ for $\lhd$, $\forall$ for $\exists$, $\&$ for $\oplus$
and $!$ for $?$.

\subsubsection{Interpreting Formulas as Games}
We may interpret each positive formula as a positive game, and each negative formula as a negative game, by fixing a truth assignment for the atomic formulas via a standard notion of first-order structure.
\begin{definition}
  An $\mathcal{L}$-structure $L$ is a set $|L|$ together with an
  interpretation function $I_L$ sending:
\begin{itemize}
\item each predicate symbol (with
  arity $n$) to a function $|L|^n \rightarrow \{\tru,\ff\}$ such that
  $I_L(\phi)(\overrightarrow{a}) \not =
  I_L\left(\overline{\phi}\right)(\overrightarrow{a})$ for all
  $\vec{a}$ and $I_L(=)(a,b) = \tru$ iff $a = b$;
\item each function
  symbol $f$ (with arity $n$) to a function $I_L(f) : |L|^n
  \rightarrow |L|$.
\end{itemize}
For any $X \subseteq \mathcal{V}$, an
  \emph{$\mathcal{L}$-{model} over $X$} is a pair $(L,v)$ where $L$ is
  an $\mathcal{L}$-structure and $v : X \rightarrow |L|$ a valuation
  function, yielding an assignment of truth values to all atomic
  formulas with variables in $X$.
\end{definition}

Given a $\mathcal{L}$-model $(L,v)$ over $X$, we may interpret each formula $A$ with free variables in $X$ as a game $\llbracket A\rrbracket(L,v)$  in as follows:
\begin{itemize}
\item Each of the units and connectives
  $\otimes$,$\parr$,$\oslash$,$\lhd$,$\mathbf{1}$,$\mathbf{0}$,$\top$,$\bot$,$!$,$?$,$\&$,$\oplus$
  is interpreted as the corresponding operation on games  from Section
  \ref{gameconnectives}, lifted to an action on families of games.
\item  Positive atoms which are assigned \emph{true} in $(L,v)$ are interpreted as the game with a single (Player) move ($\top$); positive atoms which are assigned \emph{false} are interpreted as the game with no moves ($\mathbf{0}$). Conversely, negative atoms which are assigned \emph{true} in $(L,v)$ are interpreted as the empty game ($\mathbf{1}$), whilst negative atoms which are assigned \emph{false} are interpreted as the game with a single Opponent move ($\bot$).
\item Quantifiers are interpreted as additive conjunctions and
  disjunctions over the domain of $L$ --- i.e. $\llbracket\forall x
  . N\rrbracket(L,v) = \prod_{l \in |L|} \llbracket N \rrbracket(L,v[x
  \mapsto l])$ and $\llbracket \exists x . P \rrbracket(L,v) =
  \bigoplus_{l \in |L|} \llbracket P \rrbracket(L,v[x \mapsto l])$. In
  the case of $\forall x . N$, this is equivalent to Opponent choosing
  an $x \in |L|$ and play proceeding in $N(x)$. In the case of
  $\exists x. P$, this is equivalent to Player choosing an $x \in |L|$
  and play proceeding in $P(x)$.
\end{itemize}

\noindent Note that $\llbracket A^\perp \rrbracket = \llbracket A
\rrbracket^\perp$.

\subsection{Proofs}

\label{proofinterp}

A proof of a formula $\vdash A$ will be interpreted as a \emph{uniform
  family} of \emph{winning strategies} on $\llbracket A
\rrbracket(L,v)$ for each $(L,v)$. We will formalise this
interpretation (and, importantly, the meaning of ``uniformity'') in
Section 6, but with this in mind, we can define proof rules for
\textsf{WS1}. A \emph{sequent} of \textsf{WS1} is of the form $X ;
\Theta \vdash \Gamma$ where $X \subseteq \mathcal{V}$, $\Theta$ is a
set of positive atomic formulas and $\Gamma$ is a nonempty list of
formulas such that $FV(\Theta,\Gamma) \subseteq X$. The explicit free
variable set $X$ is required for the tight correspondence between the
syntax and semantics. For brevity, let $\Phi$ range over $X ; \Theta$
contexts.

We shall interpret such a sequent as a (family of) dialogue games by
interpreting the comma operator in $\Gamma$ as left-associative
left-merge (i.e. either $\oslash$ or $\lhd$ depending on the polarity of
the right-hand operand), so that the first move must occur in the
first element (or head formula) of $\Gamma$. For example, if $M,N$ are negative
formulas and $P,Q$ positive formulas, the sequent $$\vdash M, P, Q, N$$
is semantically equivalent to $$\vdash ((M \lhd P) \lhd Q) \oslash
N.$$ Thus, in the game interpretation of a sequent $\Gamma$ the first
move must occur in the first (or \emph{head}) formula of $\Gamma$.

The derivation rules for proofs are
partitioned into \emph{core rules} and \emph{other rules}. Here $M,N$
range over negative formulas, $P,Q$ over positive formulas,
$\Gamma,\Delta$ over lists of formulas, $\Gamma^\ast$ over non-empty
lists of formulas and $\Gamma^+,\Delta^+$ over lists of positive
formulas.

\subsubsection{Core Rules}

Each $n$-ary connective $\varodot$ of \textsf{WS1} is associated with
\emph{core introduction rules} which introduce that connective in the head position of a sequent: they conclude $\Phi \vdash
\varodot(A_1, \ldots, A_n), \Gamma$ from some premises. These rules
are given in Figure \ref{coreintros}. These core introduction rules
are all \emph{additive}  (by contrast to linear logic: note in particular the difference with respect to the $\otimes$ introduction rule).

\begin{figure*}[ht]
\caption{Core Introduction Rules for \textsf{WS1}}
\begin{small}
\centering
\label{coreintros}

\vspace{1ex}
\hrule
\hrule
\begin{tabular}{ccc}
\\[0.2ex]

\AxiomC{}
\LeftLabel{$\p_\mathbf{1}$}
\UnaryInfC{$\Phi \vdash \mathbf{1}, \Gamma$}
\DisplayProof

&

\AxiomC{$\Phi\vdash A,N,\Gamma$}
\LeftLabel{$\p_\oslash$}
\UnaryInfC{$\Phi\vdash A \varoslash N, \Gamma$}
\DisplayProof

&

\AxiomC{$\Phi\vdash A,P,\Gamma$}
\LeftLabel{$\p_\lhd$}
\UnaryInfC{$\Phi\vdash A \lhd P, \Gamma$}
\DisplayProof

\\[2.5ex]

 \AxiomC{$\Phi \vdash M,N,\Gamma$}
 \AxiomC{$\Phi \vdash N,M,\Gamma$}
\LeftLabel{$\p_\otimes$}
 \BinaryInfC{$\Phi \vdash M \otimes N, \Gamma$}
\DisplayProof

&

\AxiomC{$\Phi \vdash P,Q,\Gamma$}
\LeftLabel{${\p_\parr}_1$}
\UnaryInfC{$\Phi \vdash P \parr Q, \Gamma$}
\DisplayProof

&

\AxiomC{$\Phi \vdash Q,P,\Gamma$}
\LeftLabel{${{\p_\parr}}_2$}
\UnaryInfC{$\Phi \vdash P \parr Q, \Gamma$}
\DisplayProof

\\[2.5ex] 

\AxiomC{$\Phi \vdash M, \Gamma$}
\AxiomC{$\Phi \vdash N, \Gamma$}
\LeftLabel{$\p_\&$}
\BinaryInfC{$\Phi \vdash M \& N, \Gamma$}
\DisplayProof

&
 \AxiomC{$\Phi \vdash P, \Gamma$}
\LeftLabel{${{\p_\oplus}}_1$}
 \UnaryInfC{$\Phi \vdash P \oplus Q, \Gamma$}
\DisplayProof

& 

 \AxiomC{$\Phi \vdash Q, \Gamma$}
\LeftLabel{${{\p_\oplus}}_2$}
 \UnaryInfC{$\Phi \vdash P \oplus Q, \Gamma$}
\DisplayProof

\\[2.5ex]

\AxiomC{}
\LeftLabel{$\p_\top$}
\UnaryInfC{$\Phi \vdash \top$}
\DisplayProof

&

\AxiomC{$\Phi \vdash N $}
\LeftLabel{$\p_\top^-$}
\UnaryInfC{$\Phi \vdash \top , N $}
\DisplayProof

&

\AxiomC{$\Phi \vdash P $}
\LeftLabel{$\p_\bot^+$}
\UnaryInfC{$\Phi \vdash \bot , P $}
\DisplayProof

\\[2.5ex]

\AxiomC{$X ; \Theta , \overline{\phi}(\overrightarrow{s}) \vdash \bot, \Gamma$}
\LeftLabel{$\p_\mathsf{at-}$}
\UnaryInfC{$X ; \Theta \vdash \phi(\overrightarrow{s}) , \Gamma$}
\DisplayProof

&

\AxiomC{$X ; \Theta , \overline{\phi}(\overrightarrow{s}) \vdash \top , \Gamma$}
\LeftLabel{$\p_\mathsf{at+}$}
\UnaryInfC{$X ; \Theta , \overline{\phi}(\overrightarrow{s}) \vdash \overline{\phi}(\overrightarrow{s}) , \Gamma$}
\DisplayProof

&

\AxiomC{$\Phi \vdash N , !N , \Gamma$}
\LeftLabel{$\p_!$}
\UnaryInfC{$\Phi \vdash !N , \Gamma$}
\DisplayProof

\\[4ex]

\AxiomC{$X ; \Theta \vdash P[s/x] , \Gamma$}
\LeftLabel{$\p_\exists^s$}
\RightLabel{$FV(s) \subseteq X$}
\UnaryInfC{$X ; \Theta \vdash \exists x . P, \Gamma$}
\DisplayProof

&

\AxiomC{$X \uplus \{ x \} ; \Theta \vdash N , \Gamma$}
\LeftLabel{$\p_\forall$}
\UnaryInfC{$X ; \Theta \vdash \forall x . N , \Gamma$}
\DisplayProof

&

\AxiomC{$\Phi \vdash P , ?P , \Gamma$}
\LeftLabel{$\p_?$}
\UnaryInfC{$\Phi \vdash ?P , \Gamma$}
\DisplayProof

\\[2.5ex]
 
\end{tabular}

\hrule
\hrule
\end{small}
\end{figure*}


We may interpret each of the core introduction rules with respect to
$(L,v)$ as follows:

\begin{itemize}
\item The interpretation of $\p_\mathbf{1}$ is the unique total
  strategy on the game $\mathbf{1},\Gamma$ (where it is Opponent's
  turn to start, but there are no moves for him to play since the
  first move must take place in the empty game $\mathbf{1}$).
\item The interpretation of $\p_\top$ is the unique total strategy on
  the game $\top$, where Player plays a move and the game is over.
\item The interpretation of the unary rule $\p_\oslash$ is the
  identity function, as the game denoted by the conclusion is the same
  game as that denoted by the premise. The interpretation of $\p_\lhd$
  is similar.
\item For $\p_\&$ we note that given strategies $\sigma : M , \Gamma$
  and $\tau : N , \Gamma$ we can construct a strategy on $M \& N ,
  \Gamma$ which plays as $\sigma$ if Opponent's first move is in $M$,
  and as $\tau$ if Opponent's first move is in $N$.
\item Similarly, for $\p_\otimes$ we note that given strategies
  $\sigma : M , N , \Gamma$ and $\tau : N , M , \Gamma$ we can
  construct a strategy on $M \otimes N$ which plays as $\sigma$ if
  Opponent's first move is in $M$, and as $\tau$ if Opponent's first
  move is in $N$. Here we are making use of the isomorphism $M \otimes
  N \cong (M \oslash N) \& (N \oslash M)$ --- each play in $M \otimes
  N$ must either start in $M$ (and thus be a play in $M \oslash N$) or
  in $N$ (and thus be a play in $N \oslash M$). Thus, \textsf{WS1}
  commits to a particular interpretation of $\otimes$, rather than an
  arbitrary monoidal structure.
\item For ${\p_\oplus}_1$ we note that given a strategy $\sigma : P ,
  \Gamma$ we can construct a strategy on $P \oplus Q , \Gamma$ with
  Player choosing to play his first move in $P$ and thereafter playing
  as $\sigma$. For ${\p_\oplus}_2$ Player can play his first move in $Q$
  and then play as the given strategy.
\item Similarly, for the $\p_\parr$ rules, we note that in a strategy
  on $P \parr Q , \Gamma$ Player may choose to either play his first
  move in $P$ (requiring a strategy on $P , Q, \Gamma$) or in $Q$
  (requiring a strategy on $Q,P,\Gamma$).
\item The interpretation of $\p_\bot^+$ uses the observation that
  total strategies on $\bot , P = \mathop{\uparrow} P$ are in
  correspondence with total strategies on $P$. Similarly, the
  interpretation of $\p_\top^-$ uses the observation that total
  strategies on $\top, N = \mathop{\downarrow} N$ are in
  correspondence with total strategies on $N$.
\item For ${\p_\mathsf{at}}_-$, we know that
  $\phi(\overrightarrow{s}), \Gamma$
  is interpreted by $\mathbf{1},\Gamma$ if $(L,v) \models
  \phi(\overrightarrow{s})$ and by  $\bot,\Gamma$ if $(L,v) \not \models
  \phi(\overrightarrow{s})$. In the former case, there are no moves to
  respond to, so we only need to consider the case when $(L,v) \models
  \overline{\phi}(\overrightarrow{s})$.
\item For ${\p_\mathsf{at}}_+$, we can only provide a family of
  strategies on a game whose first move is in
  $\overline{\phi}(\overrightarrow{s})$ if we know that $(L,v) \models
  \overline{\phi}(\overrightarrow{s})$ since otherwise our family has
  to contain a winning strategy on the empty positive game
  $\mathbf{0}$, of which there are none.
\item For $\p_\forall$, to give a family of strategies on $\forall x
  . N , \Gamma$ we must give a strategy on $N , \Gamma$ for each
  choice of $x$ --- that is, a family of strategies on the set of
  $\Theta$-satisfying $\mathcal{L}$-models over $X \uplus \{ x \}$.
\item For $\p_\exists$, to give a family of strategies on $\exists x
  . P , \Gamma$ we must choose a value $s$ for $x$ and give a family
  of strategies on $P[s/x] , \Gamma$. 
\end{itemize}

As well as the core introduction rules, there is a small set of
\emph{core elimination rules}, found in Figure \ref{coreelims}. 
These
permit decomposition of the second and third formula in a sequent, if
the first formula is $\bot$ or $\top$. They correspond to isomorphisms
between the premise and conclusion in the semantics, which induces a
bijection between the winning strategies on each.
For example,
$\p_\bot^-$ uses the isomorphism $\bot \oslash N \cong \bot$, and
$\p_\bot^\parr$ the isomorphism $\bot \lhd (P \parr Q) \cong (\bot
\lhd P) \lhd Q$ and $\p_\bot^\oslash$ the isomorphism $\bot \lhd (P
\oslash N) \cong (\bot \lhd P) \oslash N$.

\begin{figure*}
\caption{Core Elimination Rules for \textsf{WS1}}
\centering
\label{coreelims}

\vspace{1ex}
\hrule
\hrule
\begin{tabular}{ccc}
\\[0.2ex]

\AxiomC{$\Phi \vdash \bot , \Gamma$}
\LeftLabel{$\p_\bot^-$}
\UnaryInfC{$\Phi \vdash \bot , N , \Gamma$}
\DisplayProof

&

\AxiomC{$\Phi \vdash \bot , P \parr Q, \Gamma$}
\LeftLabel{$\p_\bot^\parr$}
\UnaryInfC{$\Phi \vdash \bot , P,  Q, \Gamma$}
\DisplayProof

& 

\AxiomC{$\Phi \vdash \bot , P \varoslash N, \Gamma$}
\LeftLabel{$\p_\bot^\varoslash$}
\UnaryInfC{$\Phi \vdash \bot , P, N, \Gamma$}
\DisplayProof

\\[2.5ex]

\AxiomC{$\Phi \vdash \top , \Gamma$}
\LeftLabel{$\p_\top^+$}
\UnaryInfC{$\Phi \vdash \top , P , \Gamma$}
\DisplayProof

&

\AxiomC{$\Phi \vdash \top , M \otimes N , \Gamma$}
\LeftLabel{$\p_\top^\otimes$}
\UnaryInfC{$\Phi \vdash \top , M , N, \Gamma$}
\DisplayProof

& 

\AxiomC{$\Phi \vdash \top , N \lhd P , \Gamma $}
\LeftLabel{$\p_\top^\lhd$}
\UnaryInfC{$\Phi \vdash \top , N, P, \Gamma$}
\DisplayProof

\\[2.5ex]

\end{tabular}

\hrule
\hrule

\end{figure*}

\begin{figure*}
\caption{Core Equality Rules for \textsf{WS1}}
\centering
\label{coreequals}

\vspace{1ex}
\hrule
\hrule
\begin{tabular}{ccc}
\\[0.2ex]

\AxiomC{$(X ; \Theta \vdash \Gamma)[ \frac{z}{x} , \frac{z}{y} ]$}
\AxiomC{$X ; \Theta , x \neq y \vdash \Gamma$}
\LeftLabel{$\p_\match^{x,y,z}$}
\BinaryInfC{$X ; \Theta \vdash \Gamma$}
\DisplayProof

&

\AxiomC{}
\LeftLabel{$\p_{\neq}$}
\UnaryInfC{$X ; \Theta , x \neq x \vdash \Gamma$}
\DisplayProof

\\[2.5ex]
 
\end{tabular}

\hrule
\hrule

\end{figure*}

Finally, there are \emph{core equality rules} which deal with
equality, given in Figure \ref{coreequals}. We can interpret the core
equality rules at a model $(L,v)$ as follows:

\begin{itemize}
\item To interpret $\p_{\neq}$ (reflexivity of identity), we take the empty family of
  strategies, since there are no $\Theta$-satisfying
  $\mathcal{L}$-models if $\Theta$ contains $x \neq x$.
\item To interpret the matching rule $\p_\match^{x,y,z}$, we note that the collection of
  $\Theta$-satisfying $\mathcal{L}$-models can be decomposed into
  those where $x$ and $y$ are identified (the left-hand premise) and
  those where they are distinct (the right-hand premise).
\end{itemize}

Once  a discipline regarding where the matching rule
is applied has been introduced, proof search in this core subsystem is particularly
simple, as the form of the sequent to be proved  determines the choice of
final rule. We will later show that the core rules
are sufficient to denote any finitary family of uniform winning
strategies.

\label{focpol}
  We make a brief note on polarities and reversibility, and a
  comparison with focused proof systems. In such systems,
  polarisation is used to differentiate between connectives whose
  corresponding rules are \emph{reversible} or \emph{irreversible}
  \cite{And_Foc}. Irreversible rules act on positive formulas. An
  irreversible rule is one where (reading upwards) in applying the
  rule one must make some definite choice, a choice which could
  determine whether the proof search succeeds or not. Thus, additive
  disjunction introduction is always an irreversible rule, and in
  linear logic so is the tensor introduction rule, since a choice must
  be made regarding how the context is split.

  In \textsf{WS1}, the core introduction rule for tensor (as for all
  such rules) is additive, not multiplicative. Thus, this rule is
  reversible, and $\otimes$ is resultantly a negative connective. In
  contrast, $\parr$ is a positive connective as there are two
  different core introduction rules, which are not reversible. Thus,
  as well as the semantic motivation, we can view our distinction
  between positive and negative formulas in the same light as the
  polarities of focused systems.

  However, there is an important distinction. In focused systems, the
  proof search alternates between negative  phases, in
  which reversible rules are  applied, and positive phases, in which irreversible rules are applied. Analytic proof search in \textsf{WS} follows a
  different two-phase discipline, in which  we first \emph{decompose}
  the first formula of a sequent into a unit using the core
  introduction rules, and then \emph{collate} the tail formulas
  together using the core elimination rules. We will give an embedding
  of \textsf{LLP} inside \textsf{WS} in Section \ref{LLP-WS}.

\subsubsection{Other Rules}

\begin{figure*}[ht]
\caption{Non-core rules of \textsf{WS1}}
\begin{small}
\centering
\label{otherrules}

\vspace{1ex}
\hrule
\hrule
\begin{tabular}{cccc}
\\[0.2ex]

\AxiomC{$\Phi \vdash \Gamma^\ast, \Delta$}
\doubleLine
\LeftLabel{$\p_\mathbf{1}^\mathsf{T}$}
\UnaryInfC{$\Phi \vdash \Gamma^\ast, \mathbf{1}, \Delta$}
\DisplayProof

&

\AxiomC{$\Phi \vdash \Gamma^\ast, M, N, \Delta$}
\doubleLine
\LeftLabel{$\p_\otimes^\mathsf{T}$}
\UnaryInfC{$\Phi \vdash \Gamma^\ast, M \otimes N, \Delta$}
\DisplayProof

& 

\AxiomC{$\Phi \vdash \Gamma^\ast, M, N, \Delta$}
\LeftLabel{$\p_\mathsf{sym}^-$}
\UnaryInfC{$\Phi \vdash \Gamma^\ast, N, M, \Delta$}
\DisplayProof

& 

\AxiomC{$\Phi \vdash \Gamma^\ast, M, \Delta$}
\LeftLabel{$\p_\wk^-$}
\UnaryInfC{$\Phi \vdash \Gamma^\ast, \Delta$}
\DisplayProof

\\[2.5ex]

\AxiomC{$\Phi \vdash \Gamma^\ast, \mathbf{0} , \Delta$}
\doubleLine
\LeftLabel{$\p_\mathbf{0}^\mathsf{T}$}
\UnaryInfC{$\Phi \vdash \Gamma^\ast, \Delta$}
\DisplayProof

&

\AxiomC{$\Phi \vdash \Gamma^\ast, P, Q, \Delta$}
\doubleLine
\LeftLabel{$\p_\parr^\mathsf{T}$}
\UnaryInfC{$\Phi \vdash \Gamma^\ast, P \parr Q, \Delta$}
\DisplayProof

&

\AxiomC{$\Phi \vdash \Gamma^\ast, P, Q, \Delta$}
\LeftLabel{$\p_\mathsf{sym}^+$}
\UnaryInfC{$\Phi \vdash \Gamma^\ast, Q, P, \Delta$}
\DisplayProof

&

\AxiomC{$\Phi \vdash \Gamma^\ast, \Delta$}
\LeftLabel{$\p_\wk^+$}
\UnaryInfC{$\Phi \vdash \Gamma^\ast, P, \Delta$}
\DisplayProof


\\[3ex]

\multicolumn{2}{c}{
\AxiomC{$\Phi \vdash \Gamma^\ast, N^\perp, \Gamma_1$}
\AxiomC{$\Phi \vdash N, \Delta^+$}
\LeftLabel{$\p_\cut$}
\BinaryInfC{$\Phi \vdash \Gamma^\ast, \Delta^+ , \Gamma_1$}
\DisplayProof
}

&

\multicolumn{2}{c}{
\AxiomC{$\Phi \vdash N^\perp$}
\AxiomC{$\Phi \vdash N , Q$}
\LeftLabel{$\p_\cut^0$}
\BinaryInfC{$\Phi \vdash Q$}
\DisplayProof
}

\\[2.5ex]

\multicolumn{2}{c}{
\AxiomC{$\Phi \vdash N , Q , \Delta^+$}
\LeftLabel{$\p_{\id\oslash}$}
\UnaryInfC{$\Phi \vdash M , N , M^\perp \lhd Q , \Delta^+$}
\DisplayProof
}

&

\multicolumn{2}{c}{
\AxiomC{$\Phi \vdash \Gamma, P_i, \Delta$}
\LeftLabel{$\p_{\oplus i}^\mathsf{T}$}
\UnaryInfC{$\Phi \vdash \Gamma, P_1 \oplus P_2, \Delta$}
\DisplayProof
}

\\[2.5ex]

\multicolumn{2}{c}{
\AxiomC{$\Phi \vdash M, \Gamma, \Delta^+$}
\AxiomC{$\Phi \vdash N, \Delta_1^+$}
\LeftLabel{$\p_\mul$}
\BinaryInfC{$\Phi \vdash M, \Gamma, N, \Delta^+, \Delta_1^+$}
\DisplayProof
}

&



\AxiomC{}
\LeftLabel{$\p_\id$}
\UnaryInfC{$\Phi \vdash N, N^\perp$}
\DisplayProof 

&

\AxiomC{$\Phi \vdash M , N , N^\perp$}
\LeftLabel{$\p_{\mathsf{ana}}$}
\UnaryInfC{$\Phi \vdash !M , N^\perp$}
\DisplayProof

\\[3ex]

\multicolumn{2}{c}{
\AxiomC{$\Phi \vdash M,\Gamma,P$}
\AxiomC{$\Phi \vdash N,\Delta^+$}
\LeftLabel{$\p_\multimap$}
\BinaryInfC{$\Phi \vdash M,\Gamma,P \oslash N,\Delta^+$}
\DisplayProof
}

&

\multicolumn{2}{c}{
\AxiomC{$\Phi \vdash \Gamma, M_1 \& M_2, \Delta$}
\LeftLabel{$\p_{\&}^\mathsf{T}{}_i$}
\UnaryInfC{$\Phi \vdash \Gamma, M_i, \Delta$}
\DisplayProof
}

\\[2.5ex]

\AxiomC{$\Phi \vdash \Gamma , !M , \Delta$}
\LeftLabel{$\p_\der^!$}
\UnaryInfC{$\Phi \vdash \Gamma , M , \Delta$}
\DisplayProof

&

\AxiomC{$\Phi \vdash \Gamma , !M , \Delta$}
\LeftLabel{$\p_\con^!$}
\UnaryInfC{$\Phi \vdash \Gamma , !M , !M , \Delta$}
\DisplayProof

&

\AxiomC{$\Phi \vdash \Gamma , P , \Delta$}
\LeftLabel{$\p_\der^?$}
\UnaryInfC{$\Phi \vdash \Gamma , ?P , \Delta$}
\DisplayProof

&

\AxiomC{$\Phi \vdash \Gamma , ?P , ?P , \Delta$}
\LeftLabel{$\p_\con^?$}
\UnaryInfC{$\Phi \vdash \Gamma , ?P , \Delta$}
\DisplayProof

\\[2.5ex]

\multicolumn{2}{c}{
\AxiomC{$X ; \Theta \vdash \Gamma , \forall x . N , \Delta$}
\LeftLabel{$\p_\forall^\mathsf{T}$}
\RightLabel{$FV(s) \subseteq X$}
\UnaryInfC{$X ; \Theta \vdash \Gamma , N[s/x] , \Delta$}
\DisplayProof
}

&

\multicolumn{2}{c}{
\AxiomC{$X ; \Theta \vdash \Gamma , P[s/x] , \Delta$}
\LeftLabel{$\p_\exists^\mathsf{T}$}
\RightLabel{$FV(s) \subseteq X$}
\UnaryInfC{$X ; \Theta \vdash \Gamma , \exists x . P , \Delta$}
\DisplayProof
}

\\[2.5ex]






\end{tabular}

\hrule
\hrule
\end{small}
\end{figure*}

The non-core rules of \textsf{WS1} are given in Figure
\ref{otherrules}, with $\Delta^+$ ranging over lists of positive
formulas, $\Gamma^*$ over non-empty lists of formulas.  These rules
reflect some of the categorical structure enjoyed by our games model,
and allow straightforward interpretation of other logics and
programming languages inside \textsf{WS1}. They include a cut rule,
a multiplicative $\otimes$ rule, a restricted form of the exchange
rule, weakening, and so on. We will later see that these rules are
admissible with respect to the rules in Figures \ref{coreintros},
\ref{coreelims} and \ref{coreequals}, when restricted to the
exponential-free subsystem of \textsf{WS1}. Informally, we can
interpret each of these rules as follows:

\begin{itemize}
\item In the cases of $\p^\mathsf{T}_\otimes$,
  $\p_\mathbf{1}^\mathsf{T}$, $\p_\parr^\mathsf{T}$,
  $\p_\mathbf{0}^\mathsf{T}$, $\p_\sym^+$ and $\p_\sym^-$, the premise
  and conclusion are the same game, up to retagging, and can be
  interpreted using game isomorphisms.
\item In the cases of $\p_\&^\mathsf{T}{}_1$, $\p_\&^\mathsf{T}{}_2$,
  $\p_\wk^-$, $\p_\der^!$ a strategy on the conclusion can be obtained
  by using only part of the strategy on the premise. For example, for
  $\p_\wk^-$ we remove all moves in $M$.
\item In the cases of $\p_\oplus^\mathsf{T}{}_1$,
  $\p_\oplus^\mathsf{T}{}_2$, $\p_\wk^+$, $\p_\der^?$, a strategy on
  the conclusion can be obtained by using the strategy on the premise
  and ignoring the extra moves available to Player.
\item The $\p_\id$ rule requires a strategy on $N \multimap N$: we can
  use a \emph{copycat} strategy in which Player always switches
  component, playing the move that Opponent previously played. The
  ${\p_\id}_\oslash$ rule can be interpreted by playing copycat in the
  $M$ component.
\item The $\p_\cut$ and $\p_\cut^0$ rules can be interpreted by
  playing the two strategies given by the premises against each other
  in the $N$ component: ``parallel composition plus hiding''.
\item The $\p_\mul$ rule can be interpreted by combining the
  strategies given by the premises in a multiplicative manner:
  Opponent's moves in $M,\Gamma$ are responded to in accordance with
  the first premise, and moves in $N$ in accordance with the
  second. The $\p_\multimap$ rule can be interpreted similarly.
\item To interpret $\p_\con^?$, we can construct a strategy on the
  conclusion by identifying the two copies of $?P$ in the premise. To
  interpret $\p_\con^!$, we can construct a strategy on the conclusion
  by identifying the two copies of $!M$ in the conclusion.
\item We can interpret $\p_\ana$ using the following construction:
  given a map $N \multimap M \oslash N$, we may ``unwrap'' it an
  infinite number of times to yield a strategy on $N \multimap
  {!}M$. The $N$ component represents a parameter that can be used to pass
  information between the separate threads, to admit history-sensitive
  behaviour.
\end{itemize}


\subsubsection{Embedding of Intuitionistic Linear Logic}
For any negative formulas $M,N$, define $M \multimap N$ to be $N \lhd
M^\perp$. Thus any formula of first-order Intuitionistic Linear Logic
is a negative formula of \textsf{WS1}. We sketch an embedding into
\textsf{WS1} of proofs of \textsf{ILL} (over the connectives
$\otimes$,$\multimap$,$\forall$,$\&$,$\mathbf{1}$,$\bot$,$!$ and
(negative) atoms, formulated with left- and right- introduction rules
as in \cite{Scha_CatMod}).
 
\begin{proposition}For any proof $p$ of $M_1, \ldots, M_n \vdash N$ in
  \textsf{ILL} with free variables in $X$, there is a proof
  $\kappa(p)$ in \textsf{WS1} of $X ; \emptyset \vdash N , M_1^\perp,
  \ldots, M_n^\perp$.
\end{proposition}
\begin{proof}
  We show that for each rule of \textsf{ILL} there is a derivation in
  \textsf{WS1} of the conclusion from the premises.

The left $\otimes$ rule corresponds to $\p_\parr^\mathsf{T}$. For the right $\otimes$ rule, with $\Gamma = G_1, \ldots, G_n$ and $\Delta = D_1, \ldots, D_m$, we duplicate the proof and use $\p_\mul$ as follows:

\begin{footnotesize}
\begin{prooftree}
\AxiomC{$\vdash M , G_1, \ldots, G_n$}
\AxiomC{$\vdash N , D_1, \ldots, D_m$}
\LeftLabel{$\p_\mul$}
\BinaryInfC{$\vdash M , N , G_1, \ldots, G_n , D_1, \ldots, D_m$}

\AxiomC{$\vdash N , D_1, \ldots, D_m$}
\AxiomC{$\vdash M , G_1, \ldots, G_n$}
\LeftLabel{$\p_\mul$}
\BinaryInfC{$\vdash N , M , D_1, \ldots, D_m , G_1, \ldots, G_n$}
\LeftLabel{$\p_\mathsf{sym}^+$}
\UnaryInfC{$\vdots$}
\LeftLabel{$\p_\mathsf{sym}^+$}
\UnaryInfC{$\vdash N , M , G_1, \ldots, G_n , D_1, \ldots, D_m$}

\LeftLabel{$\p_\mathsf{\otimes}$}
\BinaryInfC{$\vdash M \otimes N , G_1, \ldots, G_n , D_1, \ldots, D_m$}
\end{prooftree}
\end{footnotesize}

\noindent The left $\mathbf{1}$ rule corresponds to
$\p_\mathbf{0}^\mathsf{T}$. The right $\mathbf{1}$ rule  corresponds
to $\p_\mathbf{1}$. The left $\multimap$ rule can be derived as
follows: 

\begin{footnotesize}
\begin{prooftree}
\AxiomC{$\vdash L , D_1, \ldots, D_m, N^\perp$}
\AxiomC{$\vdash M , G_1, \ldots, G_n$}
\LeftLabel{$\p_\multimap$}
\BinaryInfC{$\vdash L , D_1, \ldots, D_m, N^\perp \oslash M, G_1, \ldots, G_n$}
\LeftLabel{$\p_\mathsf{sym}^+$}
\UnaryInfC{$\vdots$}
\LeftLabel{$\p_\mathsf{sym}^+$}
\UnaryInfC{$\vdash L , G_1, \ldots, G_n, N^\perp \oslash M , D_1, \ldots, D_m$}
\end{prooftree}
\end{footnotesize}

\noindent The right $\multimap$ rule corresponds to $\p_\lhd$. The left $\&$ rules correspond to the $\p_\oplus^\mathsf{T}$ rules. The right $\&$ rule corresponds to $\p_\&$. The right-$\forall$ rule corresponds to $\p_\forall$ and the left-$\forall$ rule corresponds to $\p_\exists^\mathsf{T}$.

The dereliction, contraction and weakening rules for the exponential
correspond to $\p_\der^?$, $\p_\con^?$ and $\p_\wk^+$ respectively. We
next give the translation of the right $!$ rule (promotion). We first
assume $\Gamma$ consists of a single formula $L$.

\begin{footnotesize}
\begin{prooftree}
\AxiomC{$\vdash N , ?L^\perp$}

\AxiomC{}
\LeftLabel{$\p_\id$}
\UnaryInfC{$\vdash !L , ?L^\perp$}
\LeftLabel{$\p_\mul$}

\BinaryInfC{$\vdash N , !L, ?L^\perp, ?L^\perp$}
\LeftLabel{$\p_\con^?$}
\UnaryInfC{$\vdash N , !L, ?L^\perp$}
\LeftLabel{$\p_\mathsf{ana}$}
\UnaryInfC{$\vdash !N , ?L^\perp$}
\end{prooftree}
\end{footnotesize}

\noindent We will later refer to this derived rule as $\p_\prom$. If $\Gamma$ contains more than one formula, we use the equivalence of $!M \otimes !N$ and $!(M \& N)$ in \textsf{WS1}.

The first direction $p_1 \vdash !M \otimes !N \multimap !(M \& N)$ is
defined as follows:

\begin{footnotesize}
\begin{prooftree}

\AxiomC{}
\LeftLabel{$\p_\id$}
\UnaryInfC{$\vdash !M , ?M^\perp$}

\AxiomC{}
\LeftLabel{$\p_\id$}
\UnaryInfC{$\vdash !N , ?N^\perp$}

\LeftLabel{$\p_\mul$}
\BinaryInfC{$\vdash !M , !N , ?M^\perp , ?N^\perp$}
\LeftLabel{$\p_\con^!$}
\UnaryInfC{$\vdash !M , !M , !N , ?M^\perp , ?N^\perp$}
\LeftLabel{$\p_\der^!$}
\UnaryInfC{$\vdash M , !M , !N , ?M^\perp , ?N^\perp$}
\LeftLabel{$\p_\parr^\mathsf{T}$}
\UnaryInfC{$\vdash M , !M , !N , ?M^\perp \parr ?N^\perp$}
\LeftLabel{$\p_\otimes^\mathsf{T}$}
\UnaryInfC{$\vdash M , !M \otimes !N , ?M^\perp \parr ?N^\perp$}

\AxiomC{}
\LeftLabel{$\p_\id$}
\UnaryInfC{$\vdash !M , ?M^\perp$}

\AxiomC{}
\LeftLabel{$\p_\id$}
\UnaryInfC{$\vdash !N , ?N^\perp$}

\LeftLabel{$\p_\mul$}
\BinaryInfC{$\vdash !M , !N , ?M^\perp , ?N^\perp$}
\LeftLabel{$\p_\con^!$}
\UnaryInfC{$\vdash !M , !N , !N , ?M^\perp , ?N^\perp$}
\LeftLabel{$\p_\sym$}
\UnaryInfC{$\vdash !N , !M , !N , ?M^\perp , ?N^\perp$}
\LeftLabel{$\p_\der^!$}
\UnaryInfC{$\vdash N , !M , !N , ?M^\perp , ?N^\perp$}
\LeftLabel{$\p_\parr^\mathsf{T}$}
\UnaryInfC{$\vdash N , !M , !N , ?M^\perp \parr ?N^\perp$}
\LeftLabel{$\p_\otimes^\mathsf{T}$}
\UnaryInfC{$\vdash N , !M \otimes !N , ?M^\perp \parr ?N^\perp$}

\LeftLabel{$\p_\&$}
\BinaryInfC{$\vdash M \& N , !M \otimes !N , ?M^\perp \parr ?N^\perp$}
\LeftLabel{$\p_\mathsf{ana}$}
\UnaryInfC{$\vdash !(M \& N) , ?M^\perp \parr ?N^\perp$}

\end{prooftree}
\end{footnotesize}

\noindent The second direction $p_2 \vdash !(M \& N) \multimap !M \otimes !N$ is given as follows:

\begin{footnotesize}
\begin{prooftree}
\AxiomC{}
\LeftLabel{$\p_\id$}
\UnaryInfC{$\vdash M , M^\perp$}
\LeftLabel{$\p_\oplus^\mathsf{T}{}_1$}
\UnaryInfC{$\vdash M , M^\perp \oplus N^\perp$}
\LeftLabel{$\p_\der^?$}
\UnaryInfC{$\vdash M , ?(M^\perp \oplus N^\perp)$}
\LeftLabel{$\p_\prom$}
\UnaryInfC{$\vdash !M , ?(M^\perp \oplus N^\perp)$}

\AxiomC{}
\LeftLabel{$\p_\id$}
\UnaryInfC{$\vdash N , N^\perp$}
\LeftLabel{$\p_\oplus^\mathsf{T}{}_2$}
\UnaryInfC{$\vdash N , M^\perp \oplus N^\perp$}
\LeftLabel{$\p_\der^?$}
\UnaryInfC{$\vdash N , ?(M^\perp \oplus N^\perp)$}
\LeftLabel{$\p_\prom$}
\UnaryInfC{$\vdash !N , ?(M^\perp \oplus N^\perp)$}

\LeftLabel{${\p_\mul}_\otimes$}
\BinaryInfC{$\vdash !M \otimes !N , ?(M^\perp \oplus N^\perp), ?(M^\perp \oplus N^\perp)$}
\LeftLabel{$\p_\con^?$}
\UnaryInfC{$\vdash !M \otimes !N , ?(M^\perp \oplus N^\perp)$}
\end{prooftree}
\end{footnotesize}

\noindent We can then generalise $\p_\prom$ to 

\begin{footnotesize}
\begin{prooftree}

\AxiomC{$\vdash M, ?P_1 , ?P_2, \ldots, ?P_{n-1}, ?P_n$}
\UnaryInfC{$\vdots$}
\UnaryInfC{$\vdash M, ?(P_1 \oplus P_2 \oplus \ldots \oplus P_{n-1}) , ?P_n$}
\LeftLabel{$\p_\parr^{\mathsf{T}}$}
\UnaryInfC{$\vdash M, ?(P_1 \oplus P_2 \oplus \ldots \oplus P_{n-1}) \parr ?P_n$}

\AxiomC{$p_2$}
\LeftLabel{$\p_\cut$}
\BinaryInfC{$\vdash M, ?(P_1 \oplus P_2 \oplus \ldots \oplus P_{n-1} \oplus P_n)$}
\LeftLabel{$\p_\prom$}
\UnaryInfC{$\vdash !M, ?(P_1 \oplus P_2 \oplus \ldots \oplus P_{n-1} \oplus P_n)$}
\UnaryInfC{$\vdots$}
\UnaryInfC{$\vdash !M, ?(P_1 \oplus P_2), \ldots, ?P_{n-1}, ?P_n$}

\AxiomC{$p_1$}
\LeftLabel{$\p_\cut$}
\BinaryInfC{$\vdash !M, ?P_1 \parr ?P_2, \ldots, ?P_{n-1} , ?P_n$}
\LeftLabel{$\p_\parr^\mathsf{T}$}
\UnaryInfC{$\vdash !M, ?P_1 , ?P_2, \ldots, ?P_{n-1}, ?P_n$}
\end{prooftree}
\end{footnotesize}

\noindent and interpret the right ! rule of \textsf{ILL}. \qed
\end{proof}

A detailed proof-theoretic analysis of the properties of this
translation is beyond the scope of this paper. However, we note that
the translation is semantically natural, in the following sense. We
shall see in Section~\ref{sec:catsemantics} that the categorical models of~\textsf{WS1} have
(among other properties) the structure of a standard categorical model
of~\textsf{ILL}: they are~\emph{Lafont
  categories}~\cite{MelliesPA:catsll}. The semantics of the
quantifier-free fragment of~\textsf{ILL}
induced by translation into~\textsf{WS1} followed by interpretation in
a categorical model coincides with the expected semantics
of~\textsf{ILL} in a Lafont category.

\subsubsection{New Theorems}

We next sketch some examples of formulas that are not provable in
\textsf{ILL} but are provable in \textsf{WS1} --- i.e. they denote
games on which there are uniform winning history-sensitive strategies
which are expressible in \textsf{WS1}.

The formulas  \\ \\
$((A \otimes B \multimap \bot) \otimes (C \otimes D \multimap \bot) \multimap \bot) \multimap \\
((A \multimap \bot) \otimes (C \multimap \bot) \multimap \bot) \otimes  ((B \multimap \bot) \otimes (D \multimap \bot) \multimap \bot)$ \\
\\
\noindent are not provable, in general, in   intuitionistic linear logic (in particular, when $A,B,C,D$ are instantiated as negative atoms). They are a  counterpart in \textsf{ILL} of the \emph{medial} rule $[(A \otimes B) \parr (C \otimes D)] \multimap [(A \parr C) \otimes (B \parr D)]$, using an interpretation of depolarised formulas in a polarised setting following \cite{MT_RM}.

As observed by Blass  \cite{Bla_LL}, however, there are (uniform) history-sensitive winning strategies for medial. Informally, suppose:
\begin{itemize}
\item Opponent first choses the left hand component in the output (choice 1)
\item Opponent then chooses the right hand component in the input (choice 2)
\end{itemize}
\noindent Player can then play copycat in $C$. If Opponent then switches to the second output component $((B \multimap \bot) \otimes (D \multimap \bot) \multimap \bot)$, Player must enter copycat in $D$. But this decision relies on knowledge of Opponent's choice 2, which is not possible in an innocent setting and requires history-sensitive knowledge.

An outline \textsf{WS1} proof of this formula is given in Figure
\ref{medproof}. The use of the $\p_\otimes$ demonstrates where the
proof branches; there are four branches corresponding to the two uses
of $\p_\otimes$. In each of these four branches different ${\p_\parr}_i$
proof rules are chosen at the points labelled ${\p_\parr}_1$ here.

\begin{figure*}
\caption{Outline Proof of Medial}
\begin{footnotesize}
\centering
\label{medproof}

\vspace{1ex}
\hrule
\hrule

\begin{prooftree}
\AxiomC{}
\UnaryInfC{$\overline{\alpha} , \overline{\gamma} \vdash \top$}
\UnaryInfC{$\overline{\alpha} , \overline{\gamma} \vdash \top , \overline{\delta}$}
\UnaryInfC{$\overline{\alpha} , \overline{\gamma} \vdash \overline{\gamma} , \overline{\delta}$}
\UnaryInfC{$\overline{\alpha} , \overline{\gamma} \vdash \overline{\gamma} \parr \overline{\delta}$}
\UnaryInfC{$\overline{\alpha} , \overline{\gamma} \vdash \bot , (\bot \lhd (\top \oslash \beta) \parr (\top \oslash \delta)) , \overline{\gamma} \parr \overline{\delta}$}
\UnaryInfC{$\overline{\alpha} \vdash \gamma , (\bot \lhd (\top \oslash \beta) \parr (\top \oslash \delta)) , \overline{\gamma} \parr \overline{\delta}$}
\UnaryInfC{$\overline{\alpha} \vdash \top , \gamma , (\bot \lhd (\top \oslash \beta) \parr (\top \oslash \delta)) , \overline{\gamma} \parr \overline{\delta}$}
\UnaryInfC{$\overline{\alpha} \vdash (\top \oslash \gamma) , (\bot \lhd (\top \oslash \beta) \parr (\top \oslash \delta)) , \overline{\gamma} \parr \overline{\delta}$}
\UnaryInfC{$\overline{\alpha} \vdash (\top \oslash \gamma) \oslash (\bot \lhd (\top \oslash \beta) \parr (\top \oslash \delta)) , \overline{\gamma} \parr \overline{\delta}$}
\UnaryInfC{$\overline{\alpha} \vdash \bot , \overline{\gamma} \parr \overline{\delta} , (\top \oslash \gamma) \oslash (\bot \lhd (\top \oslash \beta) \parr (\top \oslash \delta))$}
\UnaryInfC{$\overline{\alpha} \vdash (\bot \lhd \overline{\gamma} \parr \overline{\delta}) , (\top \oslash \gamma) \oslash (\bot \lhd (\top \oslash \beta) \parr (\top \oslash \delta))$}
\UnaryInfC{$\overline{\alpha} \vdash \top , \overline{\beta} , (\bot \lhd \overline{\gamma} \parr \overline{\delta}) , (\top \oslash \gamma) \oslash (\bot \lhd (\top \oslash \beta) \parr (\top \oslash \delta))$}
\UnaryInfC{$\overline{\alpha} \vdash \overline{\alpha} , \overline{\beta} , (\bot \lhd \overline{\gamma} \parr \overline{\delta}) , (\top \oslash \gamma) \oslash (\bot \lhd (\top \oslash \beta) \parr (\top \oslash \delta))$}
\LeftLabel{${\p_\parr}_1$}
\UnaryInfC{$\overline{\alpha} \vdash \overline{\alpha} \parr \overline{\beta} , (\bot \lhd \overline{\gamma} \parr \overline{\delta}) , (\top \oslash \gamma) \oslash (\bot \lhd (\top \oslash \beta) \parr (\top \oslash \delta))$}
\UnaryInfC{$\overline{\alpha} \vdash \overline{\alpha} \parr \overline{\beta} \oslash (\bot \lhd \overline{\gamma} \parr \overline{\delta}) , (\top \oslash \gamma) \oslash (\bot \lhd (\top \oslash \beta) \parr (\top \oslash \delta))$}
\UnaryInfC{$\overline{\alpha} \vdash \bot , \top \oslash \gamma , (\bot \lhd (\top \oslash \beta) \parr (\top \oslash \delta)) , \overline{\alpha} \parr \overline{\beta} \oslash (\bot \lhd \overline{\gamma} \parr \overline{\delta})$}
\UnaryInfC{$\vdash \alpha , \top \oslash \gamma , (\bot \lhd (\top \oslash \beta) \parr (\top \oslash \delta)) , \overline{\alpha} \parr \overline{\beta} \oslash (\bot \lhd \overline{\gamma} \parr \overline{\delta})$}
\UnaryInfC{$\vdash \top , \alpha , \top \oslash \gamma , (\bot \lhd (\top \oslash \beta) \parr (\top \oslash \delta)) , \overline{\alpha} \parr \overline{\beta} \oslash (\bot \lhd \overline{\gamma} \parr \overline{\delta})$}
\UnaryInfC{$\vdash \top \oslash \alpha , \top \oslash \gamma , (\bot \lhd (\top \oslash \beta) \parr (\top \oslash \delta)) , \overline{\alpha} \parr \overline{\beta} \oslash (\bot \lhd \overline{\gamma} \parr \overline{\delta})$}
\LeftLabel{${\p_\parr}_1$}
\UnaryInfC{$\vdash (\top \oslash \alpha) \parr (\top \oslash \gamma) , (\bot \lhd (\top \oslash \beta) \parr (\top \oslash \delta)) , \overline{\alpha} \parr \overline{\beta} \oslash (\bot \lhd \overline{\gamma} \parr \overline{\delta})$}
\UnaryInfC{$\vdash (\top \oslash \alpha) \parr (\top \oslash \gamma) \oslash (\bot \lhd (\top \oslash \beta) \parr (\top \oslash \delta)) , \overline{\alpha} \parr \overline{\beta} \oslash (\bot \lhd \overline{\gamma} \parr \overline{\delta})$}
\UnaryInfC{$\vdash \bot , \overline{\alpha} \parr \overline{\beta} , (\bot \lhd \overline{\gamma} \parr \overline{\delta}) , (\top \oslash \alpha) \parr (\top \oslash \gamma) \oslash (\bot \lhd (\top \oslash \beta) \parr (\top \oslash \delta))$}
\UnaryInfC{$\vdash \bot \lhd \overline{\alpha} \parr \overline{\beta} , (\bot \lhd \overline{\gamma} \parr \overline{\delta}) , (\top \oslash \alpha) \parr (\top \oslash \gamma) \oslash (\bot \lhd (\top \oslash \beta) \parr (\top \oslash \delta))$}
\AxiomC{$\vdots$}
\LeftLabel{$\p_\otimes$}
\BinaryInfC{$\vdash (\bot \lhd \overline{\alpha} \parr \overline{\beta}) \otimes (\bot \lhd \overline{\gamma} \parr \overline{\delta}) , (\top \oslash \alpha) \parr (\top \oslash \gamma) \oslash (\bot \lhd (\top \oslash \beta) \parr (\top \oslash \delta))$}
\UnaryInfC{$\vdash \top , ((\bot \lhd \overline{\alpha} \parr \overline{\beta}) \otimes (\bot \lhd \overline{\gamma} \parr \overline{\delta})) , (\top \oslash \alpha) \parr (\top \oslash \gamma) \oslash (\bot \lhd (\top \oslash \beta) \parr (\top \oslash \delta))$}
\UnaryInfC{$\vdash \top \oslash ((\bot \lhd \overline{\alpha} \parr \overline{\beta}) \otimes (\bot \lhd \overline{\gamma} \parr \overline{\delta})) , (\top \oslash \alpha) \parr (\top \oslash \gamma) \oslash (\bot \lhd (\top \oslash \beta) \parr (\top \oslash \delta))$}
\UnaryInfC{$\vdash \bot , (\top \oslash \alpha) \parr (\top \oslash \gamma) , (\bot \lhd (\top \oslash \beta) \parr (\top \oslash \delta)) , \top \oslash ((\bot \lhd \overline{\alpha} \parr \overline{\beta}) \otimes (\bot \lhd \overline{\gamma} \parr \overline{\delta}))$}
\UnaryInfC{$\vdash (\bot \lhd (\top \oslash \alpha) \parr (\top \oslash \gamma)) , (\bot \lhd (\top \oslash \beta) \parr (\top \oslash \delta)) , \top \oslash ((\bot \lhd \overline{\alpha} \parr \overline{\beta}) \otimes (\bot \lhd \overline{\gamma} \parr \overline{\delta}))$}
\AxiomC{$\vdots$}
\LeftLabel{$\p_\otimes$}
\BinaryInfC{$\vdash (\bot \lhd (\top \oslash \alpha) \parr (\top \oslash \gamma)) \otimes (\bot \lhd (\top \oslash \beta) \parr (\top \oslash \delta)) , \top \oslash ((\bot \lhd \overline{\alpha} \parr \overline{\beta}) \otimes (\bot \lhd \overline{\gamma} \parr \overline{\delta}))$}
\end{prooftree}

\hrule
\hrule
\end{footnotesize}
\end{figure*}

Similarly, the following theorems of \textsf{WS1} are not provable in \textsf{ILL} but are provable in \textsf{WS1}:
\begin{itemize}
\item $[A \otimes (C \& D)] \& [B \otimes (C \& D)] \& [(A \& B) \otimes C] \& [(A \& B) \otimes D] \multimap \\ (A \& B) \otimes (C \& D)$, also discussed in \cite{Bla_LL}
\item $\phi_\mathsf{ex} \multimap \phi_\mathsf{ex} \otimes \phi_\mathsf{ex}$ where $\phi_\mathsf{ex} = (\phi \& (\phi \multimap \bot)) \multimap \bot$.
\end{itemize}


\subsection{Embedding Polarized Linear Logic in \textsf{WS1}}

Polarized Linear Logic (\textsf{LLP}) \cite{Lau_PG} is a proof system for a  polarisation of linear logic
into negative and positive formulas. As we have noted, this is entirely  different from the polarisation of \textsf{WS1} formulas employed here: each makes sense within the proof system within which it is defined. Here, we show how proofs of LLP may be represented inside \textsf{WS1} by translation, with two objectives:
\begin{itemize}
\item To clarify the relationship between the two logical systems, and their notions of polarisation.
\item  To capture both call-by-name and call-by-value $\lambda$-calculi via known (and elegant) translations into \textsf{LLP}, which may be composed with our embedding of  \textsf{LLP} into \textsf{WS1}. In the call-by-name case, this corresponds with interpretation via intuitionistic linear logic, whereas for call-by-value it is new. 
\end{itemize}
The formulas of \textsf{LLP} (over the units) are as follows:

\[
\begin{array}{ccccccccccccc}
  P & ::= & \mathbf{1} & \mid & \mathbf{0} & \mid & P \otimes Q & \mid &
 P \varoplus Q & \mid & {\downarrow N} & \mid & {!N} \\
N &  ::= & \bot & \mid & \top & \mid & M \parr N & \mid & M \& N &
\mid & {\uparrow P} & \mid & {?P}
\end{array}
\]
There is an operation $(-)^\perp$ exchanging polarity, swapping
$\mathbf{1}$ for $\bot$, $\mathbf{0}$ for $\top$, $\otimes$ for
$\parr$, and so on.
  The presentation of \textsf{LLP} given in \cite{Lau_PG} omits the
  linear lifts $\uparrow$ and $\downarrow$ of \textsf{MALLP}. We will
  include them in our presentation of \textsf{LLP} and its embedding.

A sequent of \textsf{LLP} is a list of \textsf{LLP}
formulas. The proof rules for Polarized Linear Logic are given in
Figure \ref{LLP-rules}. $\Gamma^-$ ranges over lists of negative
formulas, and $\Gamma'$ over lists where at most one formula is
positive. We say a negative \textsf{LLP} formula $N$ is \emph{reusable} (and
write $\mathsf{reuse}(N)$) if every occurrence of $\uparrow$ occurs
under a $?$. If we exclude the linear lifts $\uparrow$ and
$\downarrow$, all negative formulas are
reusable. $\mathsf{reuse}(\Gamma^-)$ holds if all formulas in
$\Gamma^-$ are reusable.

Each provable sequent has at most one positive formula, so we can
restrict our attention to sequents of this form. It is possible to
give semantics to \textsf{LLP} proofs as \emph{innocent} strategies
\cite{Lau_PG}, which do not have access to the entire history of play.

\begin{figure*}
\caption{Proof rules for \textsf{LLP}}
\centering
\label{LLP-rules}

\vspace{1ex}
\hrule
\hrule
\[
\AxiomC{\strut}
\LeftLabel{$ax$}
\UnaryInfC{$\vdash N , N^\perp$}
\DisplayProof
\quad\quad
\AxiomC{$\vdash \Gamma , N$}
\AxiomC{$\vdash \Delta , N^\perp$}
\LeftLabel{$cut$}
\BinaryInfC{$\vdash \Gamma , \Delta$}
\DisplayProof
\quad\quad 
\AxiomC{$\vdash \Gamma , A , B , \Delta$}
\LeftLabel{$ex$}
\UnaryInfC{$\vdash \Gamma , B , A , \Delta$}
\DisplayProof
\]
\[
\AxiomC{$\vdash \Gamma , P$}
\AxiomC{$\vdash \Delta , Q$}
\LeftLabel{$\otimes$}
\BinaryInfC{$\vdash \Gamma , \Delta, P \otimes Q$}
\DisplayProof
\quad\quad 
\AxiomC{$\vdash \Gamma , M , N$}
\LeftLabel{$\parr$}
\UnaryInfC{$\vdash \Gamma , M \parr N$}
\DisplayProof
\]
\[
\AxiomC{$\vdash \Gamma , P$}
\LeftLabel{$\oplus_1$}
\UnaryInfC{$\vdash \Gamma , P \oplus Q$}
\DisplayProof
\quad\quad
\AxiomC{$\vdash \Gamma , Q$}
\LeftLabel{$\oplus_2$}
\UnaryInfC{$\vdash \Gamma , P \oplus Q$}
\DisplayProof
\quad\quad
\AxiomC{$\vdash \Gamma , M$}
\AxiomC{$\vdash \Gamma , N$}
\LeftLabel{$\&$}
\BinaryInfC{$\vdash \Gamma , M \& N$}
\DisplayProof
\]
\[
\AxiomC{\strut}
\LeftLabel{$\mathbf{1}$}
\UnaryInfC{$\vdash \mathbf{1}$}
\DisplayProof
\quad\quad
\AxiomC{$\vdash \Gamma$}
\LeftLabel{$\bot$}
\UnaryInfC{$\vdash \Gamma , \bot$}
\DisplayProof
\quad\quad
\AxiomC{\strut}
\LeftLabel{$\top$}
\UnaryInfC{$\vdash \Gamma' , \top$}
\DisplayProof
\]
\[
\AxiomC{$\vdash \Gamma^- , N$}
\LeftLabel{$\downarrow$}
\UnaryInfC{$\vdash \Gamma^- , \downarrow N$}
\DisplayProof
\quad\quad
\AxiomC{$\vdash \Gamma , P$}
\LeftLabel{$\uparrow$}
\UnaryInfC{$\vdash \Gamma , \uparrow P$}
\DisplayProof
\]
\[
\AxiomC{$\vdash \Gamma^- , N$}
\LeftLabel{$!$}
\RightLabel{$\mathsf{reuse}(\Gamma^-)$}
\UnaryInfC{$\vdash \Gamma^- , !N$}
\DisplayProof
\quad\quad
\AxiomC{$\vdash \Gamma , P$}
\LeftLabel{$?d$}
\UnaryInfC{$\vdash \Gamma , ?P$}
\DisplayProof
\]
\[
\AxiomC{$\vdash \Gamma , N , N$}
\LeftLabel{$?c$}
\RightLabel{$\mathsf{reuse}(N)$}
\UnaryInfC{$\vdash \Gamma , N$}
\DisplayProof
\quad\quad
\AxiomC{$\vdash \Gamma$}
\LeftLabel{$?w$}
\RightLabel{$\mathsf{reuse}(N)$}
\UnaryInfC{$\vdash \Gamma , N$}
\DisplayProof
\]

\hrule
\hrule

\end{figure*}






\label{LLP-WS}
We next describe an embedding of \textsf{LLP} inside
\textsf{WS1}. Apart from some renaming of units, connectives in
\textsf{LLP} will be interpreted by the same connective in
\textsf{WS1}. 
Broadly speaking, positive formulas of \textsf{LLP} will be mapped to
negative formulas of \textsf{WS1}, and negative formulas of
\textsf{LLP} to positive formulas of \textsf{WS1}. However, under this
scheme there is a mismatch for the additives: we will therefore need
to map formulas of \textsf{LLP} to \emph{families} of \textsf{WS1}
formulas. The formulas that have 
a lift as their outermost connective will be mapped to singleton
families.


Let $\mathsf{WS1}^-$ denote the set of negative \textsf{WS1} formulas,
and $\mathsf{WS1}^+$ the set of positive \textsf{WS1} formulas.

\begin{definition}
  A \emph{finite family of negative (resp. positive) \textsf{WS1}
    formulas} is a pair $(I,f)$ where $I$ is a finite set and $f : I
  \rightarrow \mathsf{WS1}^-$ (resp. $I \rightarrow
  \mathsf{WS1}^+$). 
\end{definition}

For brevity, given such a family $F = (I,f)$ we will write $|F|$ for
$I$ and $F_x$ for $f(x)$. 
We will interpret a negative formula of \textsf{LLP} as a finite
family of positive $\mathsf{WS1}$ formulas, and a positive formula of
\textsf{LLP} as a finite family of negative $\mathsf{WS1}$
formulas. We describe this mapping in Figure \ref{MALLP-to-WSN}. 
Like \cite{MT_RM}, we decompose the polarity-reversing exponentials of
\textsf{LLP} into polarity-preserving exponentials of and
polarity-switching linear lifts.

\begin{figure*}
\caption{\textsf{LLP} formulas as families of \textsf{WS1} formulas}
\centering
\label{MALLP-to-WSN}

\vspace{1ex}
\hrule
\hrule

\begin{tabular}{|l|l|}
  \hline
  $A \in $ \textsf{LLP} & $i(A) \in $ \textsf{Fam WS1} \\
  \hline
  $\mathbf{1}$ & $ (\{ \ast \} , \ast \mapsto \mathbf{1})$ \\
  $\mathbf{0}$ & $ (\{ \ast \} , \ast \mapsto \bot)$ \\
  $P \otimes Q$ & $ (|i(P)| \times |i(Q)| , \langle x , y \rangle \mapsto i(P)_x \otimes i(Q)_y)$ \\
  $P \oplus Q$ & $ (|i(P)| \uplus |i(Q)| , [ \inn_1(x) \mapsto i(P)_x, \inn_2(y) \mapsto i(Q)_y ])$ \\
  $!N$ & $ (\{ \ast \} , \ast \mapsto ! \&_{j \in |i(N)|}(\bot \lhd i(N)_j))$ \\
  $\downarrow N$ & $ (\{ \ast \} , \ast \mapsto \&_{j \in |i(N)|}(\bot \lhd i(N)_j))$ \\
  \hline
  $\bot$ & $ (\{ \ast \} , \ast \mapsto \mathbf{0})$ \\
  $\top$ & $ (\{ \ast \} , \ast \mapsto \top)$ \\
  $M \parr N$ & $ (|i(M)| \times |i(N)| , \langle x , y \rangle \mapsto i(M)_x \parr i(N)_y)$ \\
  $M \& N$ & $ (|i(M)| \uplus |i(N)| , [ \inn_1(x) \mapsto i(M)_x, \inn_2(y) \mapsto i(N)_y ])$ \\
  $\uparrow P$ & $ (\{ \ast \} , \ast \mapsto \bigoplus_{j \in |i(P)|}(\top \oslash i(P)_j))$ \\
  $?P$ & $ (\{ \ast \} , \ast \mapsto ?\bigoplus_{j \in |i(P)|}(\top \oslash i(P)_j))$ \\
  \hline
\end{tabular}

\hrule
\hrule

\end{figure*}

Note that $|i(A^\perp)| = |i(A)|$ and $i(A^\perp)_y =
i(A)_y^\perp$. We translate proofs of \textsf{LLP} to families of
proofs of \textsf{WS1} in the following manner:

\begin{itemize}
\item Given an \textsf{LLP} proof $p$ of $\vdash N_1 , \ldots , N_n$
  and $x_i \in |i(N_i)|$ for each $i$, we construct a proof
  $i(p,\overrightarrow{x_i})$ of $\vdash \bot , i(N_1)_{x_1} , \ldots
  , i(N_n)_{x_n}$
\item Given an \textsf{LLP} proof $p$ of $\vdash N_1 , \ldots , N_i ,
  Q , N_{i+1} , \ldots , N_n$ and $x_i \in |i(N_i)|$ for each $i$, we
  construct a pair $i(p,\overrightarrow{x_i}) = (y,q)$ where $y \in
  |i(Q)|$ and $q$ is a proof of $\vdash i(Q)_y , i(N_1)_{x_1} , \ldots
  , i(N_n)_{x_n}$.
\end{itemize}


\begin{proposition}
  Suppose $N$ is reusable. Then for any $x$ in $|i(N)|$, there is a
  formula $Q$ and proofs $p \vdash !Q^\perp , i(N)_x$ and $p' \vdash
  i(N)_x^\perp , ?Q$ such that $\llbracket p \rrbracket$ and $\llbracket
  p' \rrbracket$ are inverses.
\label{llpexp}
\end{proposition}
\begin{proof}
Simple induction, making use of isomorphisms $!(M \& N) \cong !M \otimes !N$. \qed
\end{proof}

\begin{proposition}
  For each \textsf{LLP} formula $P$, $y \in |i(P)|$ and sequence of
  negative \textsf{WS1} formulas $\Delta^-$ there is a \textsf{WS1}
  proof $\p^\top_{P,y} \vdash i(P)_y, \Delta^- , \top$.
\label{toprule}
\end{proposition}
\begin{proof}
  Simple induction on $P$. \qed 
\end{proof}

\noindent We next show how each of the \textsf{LLP} proof rules is
translated. The translation is simple; we demonstrate some
representative cases.

\begin{itemize}
\item The $cut$ rule, with $p = cut(q,r)$: Suppose $\Gamma = N_1 ,
  \ldots , N_i , P , N_{i+1} , \ldots N_n$ and $\Delta = M_1 , \ldots
  , M_m$. Let $x_i \in |i(N_i)|$ and $y_i \in |i(M_i)|$. Then
  $i(r,\overrightarrow{y_i}) = (y,t)$ with $y \in |i(N^\perp)|$ and $t
  \vdash i(N^\perp)_y , i(M_1)_{y_1} , \ldots , i(M_n)_{y_n}$. Then
  $i(q,\overrightarrow{x_i},y) = (y',q')$ where $y' \in |i(P)|$ and
  $$q' \vdash i(P)_{y'} , i(N_1)_{x_1} , \ldots , i(N_n)_{x_n},
  i(N)_y.$$ Applying $\p_\cut$ to this proof and $t$ results in a proof
  $g$ of $$\vdash i(P)_{y'} , i(N_1)_{x_1} , \ldots , i(N_n)_{x_n} ,
  i(M_1)_{y_1} , \ldots , i(M_m)_{y_m}$$ and we set
  $i(p,\overrightarrow{x_i},\overrightarrow{y_i}) = (y',g)$.

  The case where $\Gamma = N_1 , \ldots , N_n$ and $\Delta = M_1 , \ldots ,
  M_m$ is similar.





\item The $\uparrow$ rule, with $p = \uparrow(q)$: Let $\Gamma = N_1 , \ldots ,
  N_n$ and $x_i \in |i(N_i)|$. Then $i(q,\overrightarrow{x_i}) =
  (y,q)$ where $q \vdash i(P)_y , i(N_1)_{x_1} , \ldots ,
  i(N_n)_{x_n}$. We set $i(p,\overrightarrow{x_i})$ to be the
  following proof:

\begin{prooftree}
\AxiomC{$q \vdash i(P)_y , i(N_1)_{x_1} , \ldots , i(N_n)_{x_n}$}
\UnaryInfC{$\vdash \top , i(P)_y , i(N_1)_{x_1} , \ldots , i(N_n)_{x_n}$}
\UnaryInfC{$\vdash \top \oslash i(P)_y , i(N_1)_{x_1} , \ldots , i(N_n)_{x_n}$}
\LeftLabel{${\p_\oplus}_y$}
\UnaryInfC{$\vdash \bigoplus_{j \in |i(P)|} \top \oslash i(P)_j , i(N_1)_{x_1} , \ldots , i(N_n)_{x_n}$}
\UnaryInfC{$\vdash \bot , i(N_1)_{x_1} \parr \ldots \parr i(N_n)_{x_n} \parr (\bigoplus_{j \in |i(P)|} \top \oslash i(P)_j)$}
\UnaryInfC{$\vdash \bot , i(N_1)_{x_1} , \ldots , i(N_n)_{x_n} , \bigoplus_{j \in |i(P)|} \top \oslash i(P)_j$}
\end{prooftree}


Note that in the semantics of this rule two moves are played:
the opening lift overall (O-move) and the opening lift
in the derelicted component (P-move), which corresponds to ``focusing'' on
that component.
\item The $?c$ rule, with $p = ?c(q)$: 

  If $\Gamma = N_1 , \ldots , N_n$ and $x_i \in |i(N_i)|$ and $x \in
  |i(N)|$ then $i(q,\overrightarrow{x_i},x,x)$ is a proof of $\vdash \bot ,
  i(N_1)_{x_1} , \ldots , i(N_n)_{x_n} , i(N)_x , i(N)_x$. We can
  apply Proposition \ref{llpexp} and use $?$-contraction in
  \textsf{WS1} to yield a proof $q'$ of \\ $\vdash \bot , i(N_1)_{x_1} ,
  \ldots , i(N_n)_{x_n} , i(N)_x$ and we set
  $i(p,\overrightarrow{x_i},x) = q'$.

  If $\Gamma = N_1 , \ldots , N_i , P , N_{i+1} , \ldots , N_n$ and
  $x_i \in |i(N_i)|$ and $x \in |i(N)|$ then
  $i(q,\overrightarrow{x_i},x,x) = (y,q')$ where $q' \vdash i(P)_y ,
  i(N_1)_{x_1} , \ldots , i(N_n)_{x_n} , i(N)_x , i(N)_x$. We can
  apply Proposition \ref{llpexp} and use $?$-contraction in
  \textsf{WS1} to yield a proof $q''$ of $$\vdash i(P)_y , i(N_1)_{x_1}
  , \ldots , i(N_n)_{x_n} , i(N)_x$$ and we set
  $i(p,\overrightarrow{x_i},x) = (y,q'')$.
\end{itemize}

\noindent We can hence interpret proofs in \textsf{LLP} as (families of) proofs
in \textsf{WS1}.

\section{Representing Imperative Programs and their Properties}

\subsection{Imperative Cell}

As an example of a proof of \textsf{WS1} capturing imperative
behaviour (and which does not correspond to a proof of intuitionistic
or polarized linear logic), we give a proof which denotes the Boolean
reference cell strategy described in Section \ref{impobjstrat}, the
\emph{cell} strategy of \cite{AMc_LSS}.

Recall that this is a strategy for the game $!(\mathbf{B} \&
\mathbf{Bi})$, where $\mathbf{B} = \bot \lhd \top \oplus \top$ and
$\mathbf{Bi} = (\bot \& \bot) \lhd \top$. We can parametrise the cell
by a starting value, yielding a strategy on $\mathbf{B} \multimap
{!}(\mathbf{B} \& \mathbf{Bi})$. We may obtain this strategy using a
finite strategy $p : \mathbf{B} \multimap (\mathbf{B} \& \mathbf{Bi})
\oslash \mathbf{B}$. The strategy $p$ is defined as follows, using the
naming conventions from Section \ref{impobjstrat}:

\[
\begin{array}{ccccccccccccccccccccccccccccccccccl}
\mathbf{B} & \multimap   & (\mathbf{B}& \&   &\mathbf{Bi}) &  \oslash & \mathbf{B} &  & &    & &  & &    & &  &  \\
          &        &            \mathtt{q}      &      &      &                        \\
   \mathtt{q}       &        &                &      &      &                        \\
   b       &        &                &      &      &                        \\
          &        &            b      &      &      &                        \\
           &        &            &          &            &      &                    \mathtt{q} &       \\
           &        &            &          &            &      &                    b &       \\
\hline
          &        &            &          &     \mathtt{in}(b)      &      &      &                        \\
           &        &            &          &     \mathtt{ok}      &      &      &                        \\

           &        &            &          &            &      &                    \mathtt{q} &       \\
           &        &            &          &            &      &                    b &       \\

  \end{array}
\]

\noindent To obtain the $\mathsf{cell}$ strategy, we consider an
infinite unwrapping $\leftmoon p \rightmoon : \mathbf{B} \multimap
{!}(\mathbf{B} \& \mathbf{Bi})$, as performed by the semantics of the
$\p_\mathsf{ana}$ rule.

\[
\begin{array}{ccccccccccccccccccccccccccccccccccl}
\mathbf{B} & \multimap   & (\mathbf{B}& \&   &\mathbf{Bi}) &  \oslash & \mathbf{B} & \multimap & (\mathbf{B}& \&   &\mathbf{Bi}) &  \oslash & ((\mathbf{B}& \&   &\mathbf{Bi}) &  \oslash & \ldots) \\
   & & & & & & & & & & \mathtt{in}(b) \\
          &        &            &          &    \mathtt{in}(b)      &      &      &                        \\
           &        &            &          &     \mathtt{ok}      &      &      &                        \\
   & & & & & & & & & & \mathtt{ok} \\

   & & & & & & & & & & & & \mathtt{q} \\
           &        &            &          &            &      &                    \mathtt{q} &       \\
           &        &            &          &            &      &                    b &       \\
   & & & & & & & & & & & & b \\
   & & & & & & & & & & & & \vdots \\
  \end{array}
\]


We can represent this strategy in our system using the anamorphism rule  $\p_\mathsf{ana}$: we may prove  $!(\mathbf{B} \& \mathbf{Bi}), \mathbf{B}^\perp$  by applying this rule to a proof of $(\mathbf{B} \& \mathbf{Bi}), \mathbf{B}, \mathbf{B}^\perp$. To obtain this, we apply the  product rule to a pair of proofs:
\begin{itemize}
\item $p_\mathtt{read}$, of $\mathbf{B},\mathbf{B}, \mathbf{B}^\perp$, corresponding to a function which reads its argument, returns it \emph{and} propagates it to the next call, and
\item $p_\mathtt{write}$, of  $\mathbf{Bi}, \mathbf{B}, \mathbf{B}^\perp$, corresponding to a function  which ignores its argument, accepts a Boolean input value and propagates it to the next call.
\end{itemize}
In this proof, if a rule is not labelled it is the unique applicable core rule, and some steps are omitted for brevity.

\begin{scriptsize}
\begin{prooftree}
\AxiomC{$p_\mathtt{write} : \vdash (\bot \& \bot) \lhd \top , \bot \lhd (\top \oplus \top) , \top \oslash (\bot \& \bot)$}
\AxiomC{$p_\mathtt{read} : \vdash \bot \lhd (\top \oplus \top) , \bot \lhd (\top \oplus \top) , \top \oslash (\bot \& \bot)$}
\BinaryInfC{$\vdash ((\bot \& \bot) \lhd \top) \& (\bot \lhd (\top \oplus \top)), \bot \lhd (\top \oplus \top) , \top \oslash (\bot \& \bot)$}
\LeftLabel{$\p_{\mathsf{ana}}$}
\UnaryInfC{$\vdash ! (((\bot \& \bot	) \lhd \top) \& (\bot \lhd (\top \oplus \top))) , \top \oslash (\bot \& \bot)$}
\end{prooftree}
\end{scriptsize}

where $p_\mathtt{write}$ is

\begin{footnotesize}
\begin{prooftree}
\AxiomC{}
\UnaryInfC{$\vdash \top$}
\UnaryInfC{$\vdash \top , (\top \oslash (\bot \& \bot))$}
\LeftLabel{${\p_\oplus}_1$}
\UnaryInfC{$\vdash \top \oplus \top , (\top \oslash (\bot \& \bot))$}
\LeftLabel{${\p_\parr}_1$}
\UnaryInfC{$\vdash (\top \oplus \top) \parr (\top \oslash (\bot \& \bot))$}
\UnaryInfC{$\vdash \bot , (\top \oplus \top) \parr (\top \oslash (\bot \& \bot))$}
\UnaryInfC{$\vdash (\bot \lhd (\top \oplus \top)) \lhd (\top \oslash (\bot \& \bot))$}
\UnaryInfC{$\vdash \top , (\bot \lhd (\top \oplus \top)) \lhd (\top \oslash (\bot \& \bot))$}
\UnaryInfC{$\vdash (\top \oslash (\bot \lhd (\top \oplus \top))) , \top \oslash (\bot \& 
\bot)$}
\LeftLabel{${\p_\parr}_1$}
\UnaryInfC{$\vdash (\top \oslash (\bot \lhd (\top \oplus \top))) \parr (\top \oslash (\bot \& \bot))$}
\UnaryInfC{$\vdash \bot , (\top \oslash (\bot \lhd (\top \oplus \top))) \parr (\top \oslash (\bot \& \bot))$}
\UnaryInfC{$\vdash \bot , \top , \bot \lhd (\top \oplus \top) , \top \oslash (\bot \& \bot)$}

\AxiomC{}
\UnaryInfC{$\vdash \top$}
\UnaryInfC{$\vdash \top , (\top \oslash (\bot \& \bot))$}
\LeftLabel{${\p_\oplus}_2$}
\UnaryInfC{$\vdash \top \oplus \top , (\top \oslash (\bot \& \bot))$}
\LeftLabel{${\p_\parr}_1$}
\UnaryInfC{$\vdash (\top \oplus \top) \parr (\top \oslash (\bot \& \bot))$}
\UnaryInfC{$\vdash \bot , (\top \oplus \top) \parr (\top \oslash (\bot \& \bot))$}
\UnaryInfC{$\vdash (\bot \lhd (\top \oplus \top)) \lhd (\top \oslash (\bot \& \bot))$}
\UnaryInfC{$\vdash \top , (\bot \lhd (\top \oplus \top)) \lhd (\top \oslash (\bot \& \bot))$}
\UnaryInfC{$\vdash (\top \oslash (\bot \lhd (\top \oplus \top))) , \top \oslash (\bot \& \bot)$}
\LeftLabel{${\p_\parr}_1$}
\UnaryInfC{$\vdash (\top \oslash (\bot \lhd (\top \oplus \top))) \parr (\top \oslash (\bot \& \bot))$}
\UnaryInfC{$\vdash \bot , (\top \oslash (\bot \lhd (\top \oplus \top))) \parr (\top \oslash (\bot \& \bot))$}
\UnaryInfC{$\vdash \bot , \top , \bot \lhd (\top \oplus \top) , \top \oslash (\bot \& \bot)$}

\BinaryInfC{$\vdash \bot \& \bot , \top , \bot \lhd (\top \oplus \top) , \top \oslash (\bot \& \bot)$}
\UnaryInfC{$\vdash (\bot \& \bot) \lhd \top , \bot \lhd (\top \oplus \top) , \top \oslash (\bot \& \bot)$}
\end{prooftree}
\end{footnotesize}

and $p_\mathtt{read}$ is

\begin{footnotesize}
\begin{prooftree}
\AxiomC{}
\UnaryInfC{$\vdash \top$}
\LeftLabel{${\p_\oplus}_1$}
\UnaryInfC{$\vdash \top \oplus \top$}
\UnaryInfC{$\vdash \bot , (\top \oplus \top)$}
\UnaryInfC{$\vdash \top , \bot \lhd (\top \oplus \top)$}
\LeftLabel{${\p_\oplus}_1$}
\UnaryInfC{$\vdash  \top \oplus \top , \bot \lhd (\top \oplus \top)$}
\UnaryInfC{$\vdash  \bot , (\top \oplus \top) \oslash (\bot \lhd (\top \oplus \top))$}

\AxiomC{}
\UnaryInfC{$\vdash \top$}
\LeftLabel{${\p_\oplus}_2$}
\UnaryInfC{$\vdash \top \oplus \top$}
\UnaryInfC{$\vdash \bot , (\top \oplus \top)$}
\UnaryInfC{$\vdash \top , \bot \lhd (\top \oplus \top)$}
\LeftLabel{${\p_\oplus}_2$}
\UnaryInfC{$\vdash  \top \oplus \top , \bot \lhd (\top \oplus \top)$}
\UnaryInfC{$\vdash  \bot , (\top \oplus \top) \oslash (\bot \lhd (\top \oplus \top))$}

\BinaryInfC{$\vdash  \bot \& \bot , (\top \oplus \top) \oslash (\bot \lhd (\top \oplus \top))$}
\UnaryInfC{$\vdash  \top , (\bot \& \bot) \lhd ((\top \oplus \top) \oslash (\bot \lhd (\top \oplus \top)))$}
\UnaryInfC{$\vdash  \top \oslash (\bot \& \bot) , (\top \oplus \top) \oslash (\bot \lhd (\top \oplus \top))$}
\LeftLabel{${\p_\parr}_2$}
\UnaryInfC{$\vdash ((\top \oplus \top) \oslash (\bot \lhd (\top \oplus \top))) \parr (\top \oslash (\bot \& \bot))$}
\UnaryInfC{$\vdash \bot , ((\top \oplus \top) \oslash (\bot \lhd (\top \oplus \top))) \parr (\top \oslash (\bot \& \bot))$}
\UnaryInfC{$\vdash \bot \lhd (\top \oplus \top) , \bot \lhd (\top \oplus \top) , \top \oslash (\bot \& \bot)$}
\end{prooftree}
\end{footnotesize}

We will later give categorical semantics to \textsf{WS1}, and so the
above proof provides a categorical account of this Boolean reference
cell, using a final coalgebraic property of the exponential.

We may use this proof to interpret declaration of a Boolean reference
in either call-by-name or call-by-value settings, by composition (cut)
with (the translation of) a term-in-context of the form $\Gamma,x:\var
\vdash M:T$.  Thus we may translate the recursion-free fragments of
\emph{Idealized Algol} \cite{Rey_ALGOL} and \emph{Reduced ML} over
finite datatypes into \textsf{WS1}, for example.

\subsection{State Encapsulation}
\textsf{WS1} is more expressive than total, finitary Idealized Algol:
for instance, we may use the anamorphism rule to capture structures
such as stacks, capable of storing an arbitrarily large amount of
data. A generalised programming construct which corresponds to this
capability is the \emph{encapsulation} operation which appears as the
$\mathsf{thread}$ operator in \cite{Wol_OO}, and as the
$\mathsf{encaps}$ strategy in \cite{Long_PLGM} where it is used for
constructing imperative objects in a model based on the same
underlying notion of game as used here. The operator has type $$(s
\rightarrow (o \times s)) \rightarrow s \rightarrow (1 \rightarrow
o).$$ Here $s$ is the type of the object's internal state. The first
argument represents an object which takes an explicit state of type
$s$, and returns a value of type $o$, together with an updated
state. The second argument represents an initial state. Encapsulation
returns an object of type $1 \rightarrow o$ (a ``thunk'' of type $o$)
in which the state $s$ is encapsulated --- i.e. hidden from the
environment, but shared between separate invocations of the object. On
first invocation (unthunking) the initial state is used as the input
state, and thereafter, each fresh call receives the output state from
the previous invocation as its input.


We can represent this operation in \textsf{WS1} using the
$\p_\mathsf{ana}$ rule. To do this, we consider a call-by-value
interpretation of types. We may translate call-by-value types as
positive formulas of \textsf{LLP}: $\phi^+(1) = \mathbf{1}$, $\phi^+(A
\times B) = \phi^+(A) \otimes \phi^+(B)$ and $\phi^+(A \rightarrow B)
= {!}(\phi^+(A)^\perp \parr \uparrow \phi^+(B))$\footnote{This is
  slightly different to the original embedding presented in
  \cite{Lau_PG}, which uses $?$ rather than $\uparrow$ in the
  translation of $\rightarrow$, allowing first-class continuations to
  be interpreted (the $\lambda\mu$-calculus). The translation adopted
  here is a form of \emph{linear CPS interpretation}.}. Thus by
composition with the embedding of \textsf{LLP} in \textsf{WS1}, we may
translate the types $o$ and $s$ as the families of WS-formulas $i
\circ \phi^+(s)$ and $i \circ \phi^+(o)$. Let us assume for
simplicity, that these are singleton families $\{S\}$ and $\{O\}$
respectively (i.e. $s$ represent products of function types). Then
$\mathsf{encaps}$ may be translated as a proof of $\vdash \bot, \top
\oslash i \circ \phi^+(1 \rightarrow o), S^\perp, i \circ \phi^+(s
\rightarrow (o \times s))^\perp$ --- i.e. $\vdash \bot, \top \oslash
! \uparrow (\mathbf{0} \parr \downarrow O) , S^\perp , ?\downarrow(S \otimes
\uparrow(O^\perp \parr S^\perp))$ --- as follows:



\begin{scriptsize}
\begin{prooftree}

\AxiomC{$a$}

\AxiomC{}
\LeftLabel{$\p_\id$}
\UnaryInfC{$\vdash !\uparrow(S^\perp \parr \downarrow(O \otimes S)) , ?\downarrow(S \otimes \uparrow(O^\perp \parr S^\perp))$}

\AxiomC{$b$}

\LeftLabel{$\p_\mul$}
\BinaryInfC{$\vdash \uparrow(\downarrow O) , !\uparrow(S^\perp \parr \downarrow(O \otimes S)) , S , ?\downarrow(S \otimes \uparrow(O^\perp \parr S^\perp)) , ?\downarrow(S \otimes \uparrow(O^\perp \parr S^\perp)) , S^\perp$}
\LeftLabel{$\p_\con^?$}
\UnaryInfC{$\vdash \uparrow(\downarrow O) , !\uparrow(S^\perp \parr \downarrow(O \otimes S)) , S , ?\downarrow(S \otimes \uparrow(O^\perp \parr S^\perp)) , S^\perp$}
\UnaryInfC{$\vdash \uparrow(\downarrow O) , !\uparrow(S^\perp \parr \downarrow(O \otimes S)) \otimes S , ?\downarrow(S \otimes \uparrow(O^\perp \parr S^\perp)) \parr S^\perp$}
\LeftLabel{$\p_\mathsf{ana}$}
\UnaryInfC{$\vdash !\uparrow(\downarrow O) , ?\downarrow(S \otimes \uparrow(O^\perp \parr S^\perp)) \parr S^\perp$}
\LeftLabel{$\p_\cut$}
\BinaryInfC{$\vdash !\uparrow(\mathbf{0} \parr \downarrow O) , ?\downarrow(S \otimes
\uparrow(O^\perp \parr S^\perp)) \parr S^\perp$}

\LeftLabel{$\p_\parr^\mathsf{T}$}
\UnaryInfC{$\vdash !\uparrow(\mathbf{0} \parr
  \downarrow O) , ?\downarrow(S \otimes \uparrow(O^\perp \parr
  S^\perp)) , S^\perp$}
\LeftLabel{$\p_\sym^+$}
\UnaryInfC{$\vdash !\uparrow(\mathbf{0} \parr
  \downarrow O) , S^\perp , ?\downarrow(S \otimes \uparrow(O^\perp \parr
  S^\perp))$}
\UnaryInfC{$\vdash \top , !\uparrow(\mathbf{0} \parr
  \downarrow O) , S^\perp , ?\downarrow(S \otimes \uparrow(O^\perp \parr
  S^\perp))$}
\UnaryInfC{$\vdash \bot , \top \oslash !\uparrow(\mathbf{0} \parr
  \downarrow O) , S^\perp , ?\downarrow(S \otimes \uparrow(O^\perp \parr
  S^\perp))$}

\end{prooftree}
\end{scriptsize}

\noindent where $a$ is the evident isomorphism $\vdash ! \uparrow
(\mathbf{0} \parr \downarrow O) , ? \downarrow \uparrow O^\perp$
and $b$ is: 

\begin{footnotesize}
\begin{prooftree}
\AxiomC{}
\LeftLabel{$\p_\id$}
\UnaryInfC{$\vdash S , S^\perp$}

\AxiomC{}
\LeftLabel{$\p_\id$}
\UnaryInfC{$\vdash O , O^\perp$}

\AxiomC{}
\LeftLabel{$\p_\id$}
\UnaryInfC{$\vdash S , S^\perp$}
\LeftLabel{$\p_\mul$}
\BinaryInfC{$\vdash O , S , O^\perp , S^\perp$}
\UnaryInfC{$\vdash O \oslash S , O^\perp , S^\perp$}
\LeftLabel{$\p_\parr^\mathsf{T}$}
\UnaryInfC{$\vdash O \oslash S , O^\perp \parr S^\perp$}
\UnaryInfC{$\vdash \top , O \oslash S , O^\perp \parr S^\perp$}
\UnaryInfC{$\vdash O^\perp \parr S^\perp \parr (\downarrow O \oslash S)$}
\UnaryInfC{$\vdash \bot , O^\perp \parr S^\perp , \downarrow O \oslash S$}
\LeftLabel{${\p_\mul}_\otimes$}
\BinaryInfC{$\vdash S \otimes \uparrow(O^\perp \parr S^\perp) , \downarrow O \oslash S, S^\perp$}
\LeftLabel{$\p_\sym^+$}
\UnaryInfC{$\vdash S \otimes \uparrow(O^\perp \parr S^\perp) , S^\perp , \downarrow O \oslash S$}
\UnaryInfC{$\vdash \top , S \otimes \uparrow(O^\perp \parr S^\perp) , \downarrow O \oslash S, S^\perp$}
\UnaryInfC{$\vdash \downarrow(S \otimes \uparrow(O^\perp \parr S^\perp)) , \downarrow O \oslash S, S^\perp$}
\UnaryInfC{$\vdash \bot , (\downarrow O \oslash S) \parr \downarrow(S \otimes \uparrow(O^\perp \parr S^\perp)) \parr S^\perp$}
\UnaryInfC{$\vdash \uparrow(\downarrow O) , S , \downarrow(S \otimes \uparrow(O^\perp \parr S^\perp)) , S^\perp$}
\LeftLabel{$\p_\der^?$}
\UnaryInfC{$\vdash \uparrow(\downarrow O) , S , ?\downarrow(S \otimes \uparrow(O^\perp \parr S^\perp)) , S^\perp$}
\end{prooftree}
\end{footnotesize}

\subsection{Coroutines}
We may also give a proof denoting  a \emph{coroutining} operation,
permitting a form of deterministic multithreading, defined as a
strategy in \cite{Lai_COC,Lai_FPC}. 
 In a call-by-name setting, this corresponds to an operation taking
 two terms $s$, $t$ of type $\com \rightarrow \com$, and returning a
 command which runs $s$: when (and if) $s$ calls its argument, control
 passes to $t$.  When $t$ calls its 
  argument, control is passed back to $s$, and so on, until either
  $s$ or $t$ terminates. 
We can define a coroutining operator $\mathsf{cocomp} \vdash \Sigma ,
?(\Sigma^\perp \oslash !\Sigma) , ?(\Sigma^\perp \oslash !\Sigma)$, where $\Sigma = \top \lhd \bot$.
We first give a
proof $o$ of $(!\Sigma \multimap \bot) \multimap !\Sigma$.

\begin{footnotesize}
\begin{prooftree}
\AxiomC{}
\LeftLabel{$\p_\id$}
\UnaryInfC{$\vdash \bot \lhd ?(\top \oslash \bot) , \top \oslash !(\bot \lhd \top)$}
\UnaryInfC{$\vdash \top , \bot \lhd ?(\top \oslash \bot) , \top \oslash !(\bot \lhd \top)$}
\UnaryInfC{$\vdash \top \oslash (\bot \lhd ?(\top \oslash \bot)) , \top \oslash !(\bot \lhd \top)$}
\UnaryInfC{$\vdash \bot , \top \oslash !(\bot \lhd \top) , \top \oslash (\bot \lhd ?(\top \oslash \bot))$}
\UnaryInfC{$\vdash \bot \lhd \top , !(\bot \lhd \top) , \top \oslash (\bot \lhd ?(\top \oslash \bot))$}
\UnaryInfC{$\vdash !(\bot \lhd \top) , \top \oslash (\bot \lhd ?(\top \oslash \bot))$}
\UnaryInfC{$\vdash \top , !(\bot \lhd \top) , \top \oslash (\bot \lhd ?(\top \oslash \bot))$}
\UnaryInfC{$\vdash \top \oslash !(\bot \lhd \top) , \top \oslash (\bot \lhd ?(\top \oslash \bot))$}
\UnaryInfC{$\vdash \bot , \top \oslash (\bot \lhd ?(\top \oslash \bot)) , \top \oslash !(\bot \lhd \top)$}
\UnaryInfC{$\vdash \bot , \top , \bot \lhd ?(\top \oslash \bot) , \top \oslash !(\bot \lhd \top)$}
\UnaryInfC{$\vdash \bot \lhd \top , \bot \lhd ?(\top \oslash \bot) , \top \oslash !(\bot \lhd \top)$}
\LeftLabel{$\p_\ana$}
\UnaryInfC{$\vdash ! (\bot \lhd \top) , \top \oslash !(\bot \lhd \top)$}
\UnaryInfC{$\vdash ! (\bot \lhd \top) \lhd (\top \oslash !(\bot \lhd \top))$}
\end{prooftree}
\end{footnotesize}


\noindent We next define a proof $o' \vdash (!\Sigma \multimap \Sigma)
\multimap \bot \multimap !\Sigma$, which connects the output move of the
first argument to the Player-move in the second argument.

\begin{footnotesize}
\begin{prooftree}

\AxiomC{$o \vdash !\Sigma , \top \oslash !\Sigma$}


\AxiomC{}
\UnaryInfC{$\vdash \bot , \top$}

\AxiomC{}
\UnaryInfC{$\vdash !\Sigma , ?\Sigma^\perp$}

\LeftLabel{${\p_\mul}_\otimes$}
\BinaryInfC{$\vdash \bot \otimes !\Sigma , ?\Sigma^\perp , \top$}
\UnaryInfC{$\vdash \top , \bot , !\Sigma , ?\Sigma^\perp , \top$}
\UnaryInfC{$\vdash (\top \oslash \bot) \oslash !\Sigma , ?\Sigma^\perp , \top$}
\UnaryInfC{$\vdash \bot , ?\Sigma^\perp , \top , (\top \oslash \bot) \oslash !\Sigma$}
\UnaryInfC{$\vdash \bot \lhd ?\Sigma^\perp , \top , (\top \oslash \bot) \oslash !\Sigma$}

\LeftLabel{$\p_\cut$}
\BinaryInfC{$\vdash !\Sigma , \top , (\top \oslash \bot) \oslash !\Sigma$}
\UnaryInfC{$\vdash !\Sigma \lhd \top , \Sigma^\perp \oslash !\Sigma$}
\end{prooftree}
\end{footnotesize}

\noindent We can then define $\mathsf{cocomp}$.

\begin{footnotesize}
\begin{prooftree}

\AxiomC{}
\UnaryInfC{$\vdash \top$}
\UnaryInfC{$\vdash \bot , \top$}

\AxiomC{}
\UnaryInfC{$\vdash \top$}
\UnaryInfC{$\vdash \bot , \top$}

\BinaryInfC{$\vdash \bot \otimes \bot , \top$}
\UnaryInfC{$\vdash \top , \bot , \bot , \top$}
\UnaryInfC{$\vdash (\top \oslash \bot) \oslash \bot , \top$}
\UnaryInfC{$\vdash \bot , \top , (\top \oslash \bot) \oslash \bot$}
\UnaryInfC{$\vdash \Sigma , \Sigma^\perp \oslash \bot$}

\AxiomC{}
\LeftLabel{$\p_\id$}
\UnaryInfC{$\vdash \Sigma , \Sigma^\perp$}

\AxiomC{}
\LeftLabel{$\p_\id$}
\UnaryInfC{$\vdash !\Sigma , ?\Sigma^\perp$}
\AxiomC{}
\LeftLabel{$\p_\id$}
\UnaryInfC{$\vdash \bot , \top$}
\LeftLabel{$\p_\multimap$}
\BinaryInfC{$\vdash !\Sigma , ?\Sigma^\perp \oslash \bot, \top$}
\LeftLabel{$\p_\sym^+$}
\UnaryInfC{$\vdash !\Sigma , \top , ?\Sigma^\perp \oslash \bot$}

\LeftLabel{$\p_\multimap$}
\BinaryInfC{$\vdash \Sigma , \Sigma^\perp \oslash !\Sigma , \top , ?\Sigma^\perp \oslash \bot$}
\LeftLabel{$\p_\sym^+$}
\UnaryInfC{$\vdash \Sigma , \top , \Sigma^\perp \oslash !\Sigma , ?\Sigma^\perp \oslash \bot$}
\UnaryInfC{$\vdash \Sigma \lhd \top , \Sigma^\perp \oslash !\Sigma , ?\Sigma^\perp \oslash \bot$}

\LeftLabel{$\p_\cut$}
\BinaryInfC{$\vdash \Sigma , \Sigma^\perp \oslash !\Sigma , ?
\Sigma^\perp \oslash \bot$}

\AxiomC{$o'$}

\LeftLabel{$\p_\cut$}
\BinaryInfC{$\vdash \Sigma , \Sigma^\perp \oslash !\Sigma , \Sigma^\perp \oslash !\Sigma$}
\end{prooftree}
\end{footnotesize}


\label{cocompdef}




\label{ialgemb}
\subsection{Specifying Properties of Programs} 
The formulas of \textsf{WS1} are more expressive than the types of languages such as Idealized Algol, and hence they enable the behaviour of history sensitive strategies to be specified both more abstractly and more precisely. For example, formulas can specify the order in which arguments are interrogated, how many times they are interrogated, and relationships between inputs and outputs of ground type (using the first-order structure).

\subsubsection{Data-Independent Programming}
\label{aqdi}
We can use quantifiers to represent \emph{data-independent} structures
such as cells and stacks, where the underlying ground type at a given
$\mathcal{L}$-structure $L$ is $|L|$. As a formula/game, this ground
type is represented by $\mathbf{V} = \bot \lhd \exists x . \top$ --- a
dialogue in this game consists of Opponent playing a question move $q$
and Player responding with an element of $|L|$. We can represent a
stream of such values using the formula $!\mathbf{V}$.

Let $\mathbf{Vi} = \forall x . \bot \lhd \top$ represent an `input
version' of $\mathbf{V}$, where Opponent plays an $|L|$ value and
Player then accepts it, analogous to $\mathbf{Bi}$ above. The type of
a stack object can then be given by the formula ${!}(\mathbf{V} \&
\mathbf{Vi})$, with a ``pop'' and a ``push'' method. We give a proof
denoting the behaviour of such a stack, parametrised by a starting
stack, of type $!\mathbf{V} \multimap {!}(\mathbf{V} \&
\mathbf{Vi})$.

\begin{footnotesize}
\begin{prooftree}
\AxiomC{}
\LeftLabel{$\p_\id$}
\UnaryInfC{$\vdash !(\bot \lhd \exists x . \top) , ?(\top \oslash \forall x . \bot)$}
\LeftLabel{$\p_\con^!$}
\UnaryInfC{$\vdash !(\bot \lhd \exists x . \top), !(\bot \lhd \exists x . \top) , ?(\top \oslash \forall x . \bot)$}
\LeftLabel{$\p_\der^!$}
\UnaryInfC{$\vdash \bot \lhd \exists x . \top, !(\bot \lhd \exists x . \top) , ?(\top \oslash \forall x . \bot)$}

\AxiomC{}
\LeftLabel{$\p_\id$}
\UnaryInfC{$\{ x \} ; \vdash !(\bot \lhd \exists x . \top) , ?(\top \oslash \forall x . \bot)$}
\UnaryInfC{$\{ x \} ; \vdash \top , !(\bot \lhd \exists x . \top) , ?(\top \oslash \forall x . \bot)$}
\LeftLabel{$\p_\exists^x$}
\UnaryInfC{$\{ x \} ; \vdash \exists x . \top , !(\bot \lhd \exists x . \top) , ?(\top \oslash \forall x . \bot)$}
\UnaryInfC{$\{ x \} ; \vdash \bot , \exists x . \top , !(\bot \lhd \exists x . \top) , ?(\top \oslash \forall x . \bot)$}
\UnaryInfC{$\{ x \} ; \vdash \bot \lhd \exists x . \top , !(\bot \lhd \exists x . \top) , ?(\top \oslash \forall x . \bot)$}
\UnaryInfC{$\{ x \} ; \vdash !(\bot \lhd \exists x . \top) , ?(\top \oslash \forall x . \bot)$}
\UnaryInfC{$\{ x \} ; \vdash \top , !(\bot \lhd \exists x . \top) , ?(\top \oslash \forall x . \bot)$}
\UnaryInfC{$\{ x \} ; \vdash \bot , \top, !(\bot \lhd \exists x . \top) , ?(\top \oslash \forall x . \bot)$}
\UnaryInfC{$\vdash \forall x . \bot , \top, !(\bot \lhd \exists x . \top) , ?(\top \oslash \forall x . \bot)$}
\UnaryInfC{$\vdash \forall x . \bot \lhd \top, !(\bot \lhd \exists x . \top) , ?(\top \oslash \forall x . \bot)$}

\BinaryInfC{$\vdash (\bot \lhd \exists x . \top) \& (\forall x . \bot \lhd \top), !(\bot \lhd \exists x . \top) , ?(\top \oslash \forall x . \bot)$}
\LeftLabel{$\p_\ana$}
\UnaryInfC{$\vdash !((\bot \lhd \exists x . \top) \& (\forall x . \bot \lhd \top)), ?(\top \oslash \forall x . \bot)$}
\end{prooftree}
\end{footnotesize}

Once again, we use $\p_\mathsf{ana}$ to obtain the infinite
behaviour, applied to a proof $q$ of $!\mathbf{V} \multimap
(\mathbf{V} \& \mathbf{Vi}) \oslash !\mathbf{V}$. The strategy denoted
by $q$ performs as `copycat' in the $!\mathbf{V} \multimap \mathbf{V}
\oslash !\mathbf{V}$ component, and in the $!\mathbf{V} \multimap \mathbf{Vi}
\oslash !\mathbf{V}$ component behaves as follows:

\[
\begin{array}{ccccccccccccccccccccccccccccccccccl}

!\mathbf{V} & \multimap   & \mathbf{Vi} & \oslash & !\mathbf{V} \\
        &        &            \mathtt{in}(v)      &      &      &                        \\
        &        &            \mathtt{ok}      &      &      &                        \\
        &        &                  &      &  \mathtt{q}    &                        \\
        &        &                  &      &    v  &                        \\
  \end{array}
\]

\noindent and then enters copycat.


\subsubsection{Good Variables}

One respect in which  the game semantics of  Idealized Algol (and other imperative languages) fails to reflect its syntax fully is in the existence in the model of \emph{bad variables} which do not return the last value assigned to them \cite{AMc_LSS}. In \textsf{WS1} we may define formulas for which the only  proof denotes a good variable.

The formula $\mathbf{worm} = \mathbf{Bi} \oslash {!}\mathbf{B}$
represents a Boolean variable which can be written once, then read
many times. One proof/strategy of this formula will indeed be a valid
Boolean cell: if Opponent plays \texttt{inputX} then Player responds
with \texttt{ok}, if Opponent then tries to read the cell \texttt{q},
then Player responds with \texttt{X}. But there are also bad
variables: for example, the read method may always return
\texttt{True} regardless of what was written.

To exclude such behaviour, we can replace the input/output moves with
atoms. Define $\mathbf{B}^{\phi,\psi} = \bot \lhd (\overline{\phi} \oplus
\overline{\psi})$ and $\mathbf{Bi}^{\phi,\psi} = (\phi \& \psi) \lhd \top$, with
$\mathbf{worm}^{\phi,\psi} = \mathbf{Bi}^{\phi,\psi} \oslash {!}\mathbf{B}^{\phi,\psi}$. If $\phi$ and
$\psi$ are assigned $\tru$, then this denotes the same dialogue as
$\mathbf{worm}$. However, the denotation of any proof of
$\mathbf{worm}^{\phi,\psi}$ at such a model must be the good variable
strategy. The rule for $\phi$ (and semantically, uniformity of
strategies) ensures that $\phi$ must be played before
$\overline{\phi}$, and $\psi$ before $\overline{\psi}$. Consequently,
Player can only respond with a particular Boolean value in the
\texttt{read} component if that same value has previously been given
as an input in the \texttt{write} component, so good-variable
behaviour is assured. The following proof of this formula uses only
the core rules and the promotion rule.


\begin{footnotesize}
\begin{prooftree}
\AxiomC{}
\UnaryInfC{$\overline{\phi} \vdash \top$}
\UnaryInfC{$\overline{\phi} \vdash \overline{\phi}$}
\LeftLabel{${\p_\oplus}_1$}
\UnaryInfC{$\overline{\phi} \vdash \overline{\phi} \oplus \overline{\psi}$}
\UnaryInfC{$\overline{\phi} \vdash \bot , \overline{\phi} \oplus \overline{\psi}$}
\UnaryInfC{$\overline{\phi} \vdash \bot \lhd \overline{\phi} \oplus \overline{\psi}$}
\LeftLabel{$\mathsf{prom}$}
\UnaryInfC{$\overline{\phi} \vdash !(\bot \lhd \overline{\phi} \oplus \overline{\psi})$}
\UnaryInfC{$\overline{\phi} \vdash \top , !(\bot \lhd \overline{\phi} \oplus \overline{\psi})$}
\UnaryInfC{$\overline{\phi} \vdash \top \oslash !(\bot \lhd \overline{\phi} \oplus \overline{\psi})$}
\UnaryInfC{$\overline{\phi} \vdash \bot , \top \oslash !(\bot \lhd \overline{\phi} \oplus \overline{\psi})$}
\UnaryInfC{$\overline{\phi} \vdash \bot , \top , !(\bot \lhd \overline{\phi} \oplus \overline{\psi})$}
\UnaryInfC{$\vdash \phi , \top , !(\bot \lhd \overline{\phi} \oplus \overline{\psi})$}

\AxiomC{}
\UnaryInfC{$\overline{\psi} \vdash \top$}
\UnaryInfC{$\overline{\psi} \vdash \overline{\psi}$}
\LeftLabel{${\p_\oplus}_2$}
\UnaryInfC{$\overline{\psi} \vdash \overline{\phi} \oplus \overline{\psi}$}
\UnaryInfC{$\overline{\psi} \vdash \bot , \overline{\phi} \oplus \overline{\psi}$}
\UnaryInfC{$\overline{\psi} \vdash \bot \lhd \overline{\phi} \oplus \overline{\psi}$}
\LeftLabel{$\mathsf{prom}$}
\UnaryInfC{$\overline{\psi} \vdash !(\bot \lhd \overline{\phi} \oplus \overline{\psi})$}
\UnaryInfC{$\overline{\psi} \vdash \top , !(\bot \lhd \overline{\phi} \oplus \overline{\psi})$}
\UnaryInfC{$\overline{\psi} \vdash \top \oslash !(\bot \lhd \overline{\phi} \oplus \overline{\psi})$}
\UnaryInfC{$\overline{\psi} \vdash \bot , \top \oslash !(\bot \lhd \overline{\phi} \oplus \overline{\psi})$}
\UnaryInfC{$\overline{\psi} \vdash \bot , \top , !(\bot \lhd \overline{\phi} \oplus \overline{\psi})$}
\UnaryInfC{$\vdash \psi , \top , !(\bot \lhd \overline{\phi} \oplus \overline{\psi})$}

\BinaryInfC{$\vdash (\phi \& \psi) , \top , !(\bot \lhd \overline{\phi} \oplus \overline{\psi})$}
\UnaryInfC{$\vdash (\phi \& \psi) \lhd \top , !(\bot \lhd \overline{\phi} \oplus \overline{\psi})$}
\UnaryInfC{$\vdash ((\phi \& \psi) \lhd \top) \oslash !(\bot \lhd \overline{\phi} \oplus \overline{\psi})$}
\end{prooftree}
\end{footnotesize}

\noindent We cannot use $!$ to obtain a formula which admits only an
arbitrarily reusable `good variable', but we can obtain finite
approximations. For example, the formula $$\mathbf{Bi}^{\alpha,\beta}
\oslash (\mathbf{worm}^{\phi,\psi} ~ \& ~ (\mathbf{B}^{\alpha,\beta}
\oslash \mathbf{worm}^{\phi,\psi}) ~ \& ~ (\mathbf{B}^{\alpha,\beta}
\oslash (\mathbf{B}^{\alpha,\beta} \oslash
\mathbf{worm}^{\phi,\psi})))$$ models a good variable that can be
written to twice, and can be read at most twice before the second
write. Strategies on such formulas then approximate our reusable cell
strategy above on $!(\mathbf{B} \& \mathbf{Bi})$.


\section{Categorical Semantics for \textsf{WS1}}
\label{sec:catsemantics}

To give a  formal semantics for our logic, we first introduce a notion of  categorical model which captures everything except the first-order  structure  (quantifiers and atoms). We shall use notation $\eta : F \Rightarrow G : \mathcal{C} \rightarrow
\mathcal{D}$ to mean $\eta$ is a natural transformation from $F$ to
$G$ with $F,G : \mathcal{C} \rightarrow \mathcal{D}$.

First, we define some categories of games that will form the intended
instance of our categorical model. Objects in these categories will be
negative games, and an arrow $A \rightarrow B$ will be a strategy on
$A \multimap B$. We can compose strategies using ``parallel
composition plus hiding''. Suppose $\sigma : A \multimap B$ and $\tau
: B \multimap C$, define $$\sigma \| \tau = \{ s \in (M_A + M_B +
M_C)^\ast : s|_1 \in P_A \wedge s|_2 \in P_B \wedge s|_3 \in P_C \}$$
and set $$\tau \circ \sigma = \{ s|_{1,3} : s \in \sigma \| \tau \}.$$
It is well-known that $\tau \circ \sigma$ is a well-formed strategy on
$A \multimap C$ (see e.g. \cite{AJ_MLL}).

\begin{proposition}
  Composition is associative, and there is an identity $A \rightarrow
  A$ given by the copycat strategy: $\{ s \in P_{A \multimap A} :
  \gamma(s) \}$ where $\gamma(s)$ holds if and only if $t|_1 = t|_2$
  for all even-length prefixes $t$ of $s$.
\end{proposition}

\begin{definition}
  The category $\mathcal{G}$ has negative games as objects, and a map
  $\sigma : A \rightarrow B$ is a strategy on $A \multimap B$ with
  composition and identity as above.
\end{definition}

\noindent This category has been studied extensively in e.g.
\cite{Lam_SGL,Cur_SS,Long_PLGM}, and has equivalent presentations
using graph games \cite{HS_GG} and locally Boolean domains
\cite{Lai_LBD}.

If $A$, $B$ and $C$ are bounded, $\sigma : A \multimap B$ and $\tau :
B \multimap C$ are total then $\tau \circ \sigma$ is also total. Total
strategies do not compose for unbounded games, however. Winning
strategies on unbounded games do compose \cite{Hy_GS}, and the
identity strategy is winning.



\noindent 
\begin{definition}
  The category $\mathcal{W}$ has negative games as objects and winning
  strategies as maps.
\end{definition}

\noindent A map $\sigma : A \rightarrow B$ is \emph{strict} if it
responds to Opponent's first move with a move in $A$, if it responds at
all. Strict strategies are closed under composition and the identity
is strict.

\begin{definition}
  The category $\mathcal{G}_s$ has negative games as objects and
  strict strategies as maps.
  The category $\mathcal{W}_s$ has negative games as objects and
  strict winning strategies as maps.
\end{definition}

\noindent Isomorphisms in $\mathcal{W}$ correspond to forest
isomorphisms and all isomorphisms are total and strict
\cite{Lau_CIT}. 

Each of the above categories can be endowed with  symmetric monoidal structure, given by $(I, \otimes)$ where $I$ is the
empty game $\mathbf{1}$ and the action of $\otimes$ on objects is as defined in Section \ref{gameconnectives}.


\subsection{Sequoidal Closed Structure}

The notions of \emph{sequoidal category} and \emph{sequoidal closed category} were first introduced in \cite{Lai_HOS}.

\begin{definition}
A \emph{sequoidal category} consists of:
\begin{itemize}
\item A symmetric monoidal category $(\mathcal{C}, I, \otimes)$ (we will call the relevant isomorphisms $\assoc : (A \otimes B) \otimes C \cong A \otimes (B \otimes C)$, $\li : I \otimes A \cong A$, $\ri : A \otimes I \cong I$ and $\sym : A \otimes B \cong B \otimes A$)
\item A category $\mathcal{C}_s$
\item A right-action $\varoslash$ of $\mathcal{C}$ on $\mathcal{C}_s$. That is, a functor $\_\varoslash\_ : \mathcal{C}_s \times \mathcal{C} \rightarrow \mathcal{C}_s$ with natural isomorphisms $\unit_\oslash : A \varoslash I \cong A$ and $\passoc : A \varoslash (B \otimes C) \cong (A \varoslash B) \varoslash C$ satisfying the following coherence conditions \cite{JK_Act}:
\begin{diagram}
A \oslash (B \otimes (C \otimes D)) & \rTo^{\passoc} & (A \oslash B) \oslash (C \otimes D) & \rTo^{\passoc} & ((A \oslash B) \oslash C) \oslash D \\
\dTo^{\id \oslash \assoc}            &        &  & \ruTo^{\passoc \oslash \id} \\
A \oslash ((B \otimes C) \otimes D) & \rTo^{\passoc} & (A \oslash (B \otimes C)) \oslash D \\
\end{diagram}
\begin{diagram}
A \oslash (I \otimes B)   & \rTo^{\passoc} &  (A \oslash I) \oslash B & & & A \oslash (B \otimes I) & \rTo^\passoc & (A \oslash B) \oslash I \\
\dTo^{\id \oslash \li} & \ldTo^{\unit_\oslash \oslash \id} & & \hspace{5pt} & & \dTo^{\id \oslash \ri} & \ldTo^{\unit_\oslash} \\
A \oslash B & & & & & A \oslash B \\
\end{diagram}
\item A functor $J : \mathcal{C}_s \rightarrow \mathcal{C}$
\item A natural transformation $\wk : J(\_) \otimes \_ \Rightarrow J(\_ \varoslash \_)$ satisfying further coherence conditions \cite{Lai_HOS}:
\begin{diagram}
JA \otimes I & \rTo^{\ri} & JA & & (JA \otimes B) \otimes C & \rTo^{\wk \otimes \id} & J(A \oslash B) \otimes C & \rTo^{\wk} & J((A \oslash B) \oslash C) \\
\dTo^{\wk} & \ruTo^{J(\unit_\oslash)} & & \hspace{2pt} & \dTo^{\assoc}  & & & \ruTo^{J(\passoc)} \\
J(A \oslash I) & & & & JA \otimes (B \otimes C) & \rTo^{\wk} & J(A \oslash (B \otimes C)) \\
\end{diagram}
\end{itemize}
\end{definition}

\begin{definition}
  An \emph{inclusive sequoidal category} is a sequoidal category in
  which $\mathcal{C}_s$ is a full-on-objects subcategory of
  $\mathcal{C}$ containing $\wk$ and the monoidal isomorphisms; $J$ is
  the inclusion functor; and $J$ reflects isomorphisms.
\end{definition}

\noindent We can identify this structure in our categories of games:
we can extend the left-merge operator $\oslash$ to an action
$\mathcal{G}_s \times \mathcal{G} \rightarrow \mathcal{G}_s$. If
$\sigma : A \rightarrow B$ and $\tau : C \rightarrow D$ then $\sigma
\oslash \tau : A \oslash C \rightarrow B \oslash D$ plays as $\sigma$
between $A$ and $B$ and as $\tau$ between $C$ and $D$. The strictness
of $\sigma$ guarantees that this
yields a valid strategy on $(A \oslash C) \multimap (B \oslash
D)$. The isomorphisms $\passoc$ and
$\unit_\oslash$ exist, and there is a natural copycat strategy $\wk :
M \otimes N \rightarrow M \oslash N$ in $\mathcal{G}_s$, all
satisfying the required axioms \cite{Lai_FPC}. The functor $J$
reflects isomorphisms as the inverse of strict isomorphisms are
strict. Thus $(\mathcal{G},\mathcal{G}_s)$ forms an inclusive
sequoidal category; as does $(\mathcal{W},\mathcal{W}_s)$. 

\begin{definition}
  An inclusive sequoidal category is \emph{Cartesian} if
  $\mathcal{C}_s$ has finite products preserved by $J$ (we will write
  $t_A$ for the unique map $A \rightarrow 1$). It is
  \emph{decomposable} if the natural transformations $\dec = \langle
  \wk , \wk \circ \sym \rangle : A \otimes B \Rightarrow (A \varoslash
  B) \times (B \varoslash A) : \mathcal{C}_s \times \mathcal{C}_s
  \rightarrow \mathcal{C}_s$ and $\dec^0 = \mathsf{t}_I : I
  \Rightarrow 1 : \mathcal{C}_s$ are isomorphisms (so, in particular,
  $(\mathcal{C},\otimes,I)$ is an affine SMC).

  A Cartesian sequoidal category is \emph{distributive} if the natural
  transformations $\dist = \langle \pi_1 \varoslash \id_C , \pi_2
  \varoslash \id_C \rangle : (A \times B) \varoslash C \Rightarrow (A
  \varoslash C) \times (B \varoslash C) : \mathcal{C}_s \times
  \mathcal{C}_s \times \mathcal{C} \rightarrow \mathcal{C}_s$ and
  $\dist_0 = \mathsf{t}_{1 \varoslash C} : 1 \varoslash C \Rightarrow
  1 : \mathcal{C} \rightarrow \mathcal{C}_s$ are isomorphisms.
\end{definition}

\noindent We write $\dist^0 : I \oslash C \cong I$ for the isomorphism
$(\dec^0)^{-1} \circ \dist_0 \circ (\dec^0 \oslash \id)$.

In the game categories defined above, $M \& N$ is a product of $M$ and
$N$, and the empty game $I$ is a terminal object as well as the
monoidal unit. The decomposability and distributivity isomorphisms
above exist as natural copycat morphisms \cite{Lai_FPC}. In fact,
$\mathcal{W}$ and $\mathcal{G}$ have all small products, following the
construction in Section \ref{gameconnectives}, with the corresponding
distributivity isomorphism with respect to $\oslash$.

\begin{definition}
A \emph{sequoidal closed category} is an inclusive sequoidal category where $\mathcal{C}$ is symmetric monoidal closed and the map $f \mapsto \Lambda(f \circ \wk)$ defines a natural isomorphism $\Lambda_s : \mathcal{C}_s(B \varoslash A, C) \Rightarrow \mathcal{C}_s(B, A \multimap C)$.
\end{definition}

\noindent We can show that $\mathcal{G}$ and $\mathcal{W}$ are
sequoidal closed, with the internal hom given by $\multimap$
\cite{Lai_FPC}.

In any sequoidal closed category, define $\app_s : (A \multimap B)
\varoslash A \rightarrow B$ as $\Lambda_s^{-1}(\id)$, and $\app : (A
\multimap B) \otimes A \rightarrow B = \Lambda^{-1}(\id)$, noting that
$\app = \app_s \circ \wk$. If $f : A \rightarrow B$ let $\Lambda_I(f)
: I \rightarrow A \multimap B$ denote the name of $f$, i.e. $\Lambda(f
\circ \ri)$. Write  $\Lambda_I^{-1}$ for the inverse operation.

\begin{proposition}
  In any sequoidal closed category, $\multimap$ restricts to a functor $\mathcal{C}^\op \times \mathcal{C}_s \rightarrow \mathcal{C}_s$ with natural isomorphisms $\unit_\multimap : I \multimap A \cong A$ and $\passoc_\multimap : A \otimes B \multimap C \cong A \multimap (B \multimap C)$ in $\mathcal{C}_s$.
\end{proposition}
\begin{proof}
 We need to show that if $g$ is in $\mathcal{C}_s$ then $f \multimap g$ is in $\mathcal{C}_s$. But $f \multimap g = \Lambda(g \circ \app \circ (\id \otimes f)) = \Lambda(g \circ \app_s \circ \wk \circ (\id \otimes f)) = \Lambda(g \circ \app_s \circ (\id \oslash f) \circ \wk) = \Lambda_s (g \circ \app_s \circ (\id \oslash f))$ which is in $\mathcal{C}_s$.

In any symmetric monoidal category the isomorphisms $\unit_\multimap$ and $\passoc_\multimap$ exist, but we must show that they are strict.
\begin{itemize}
\item $\unit_\multimap : I \multimap A \rightarrow A$ is given by $\app \circ \ri^{-1}$. This $\app_s \circ \wk \circ \ri^{-1} = \app_s \circ \unit_\oslash^{-1}$ which is a map in $\mathcal{C}_s$.
\item $\passoc_\multimap : A \otimes B \multimap C \cong A \multimap (B \multimap C)$ is given by $\Lambda(\Lambda(\app \circ \assoc)) = \Lambda(\Lambda(\app_s \circ \wk \circ \assoc))
= \Lambda(\Lambda(\app_s \circ \passoc^{-1} \circ \wk \circ (\wk \otimes \id)))
= \Lambda(\Lambda(\app_s \circ \passoc^{-1} \circ \wk) \circ \wk)
= \Lambda_s(\Lambda_s (\app_s \circ \passoc^{-1}))$ which is in $\mathcal{C}_s$.
\end{itemize}
The inverses of the above maps are strict as $J$ reflects isomorphisms.
\qed
\end{proof}

\label{WScatdef}

\noindent In distributive, decomposable sequoidal closed categories we can also define the following natural transformations:
\begin{itemize}
\item The isomorphism $\psym : (A \oslash B) \oslash C \cong (A \oslash C) \oslash B$ given by $\passoc \circ (\id \oslash \sym) \circ \passoc^{-1}$.
\item The isomorphism $\psym_\multimap : C \multimap (B \multimap A) \cong B \multimap (C \multimap A)$ given by $\passoc_\multimap \circ (\sym \multimap \id) \circ \passoc_\multimap^{-1}$
\item The isomorphism $\dist_\multimap : A \multimap (B \times C) \rightarrow (A \multimap B) \times (A \multimap C)$ given by $\langle \id \multimap \pi_1 , \id \multimap \pi_2 \rangle$, whose inverse is $\Lambda\langle \app \circ (\pi_1 \otimes \id) , \app \circ (\pi_2 \otimes \id) \rangle$. This isomorphism exists in any monoidal closed category with products.
\item The map $\af : A \Rightarrow I$ given by $(\dec^0)^{-1} \circ \mathsf{t}_A$.
\item The isomorphism $\dist_\multimap^0 : A \multimap I \rightarrow I$ given by $\af$ whose inverse is $\Lambda(\ri \circ (\id \otimes \af))$. We must check that these are inverses: $\af \circ \Lambda(\ri \circ (\id \otimes \af)) = \id$ as both are maps into the terminal object, and $\Lambda(\ri \circ (\id \otimes \af)) \circ \af = \Lambda(\ri \circ (\af \otimes \id) \circ (\af \otimes \id)) = \Lambda(\app) = \id$ as required. We know that $\ri \circ (\af \otimes \id) \circ (\af \otimes \id) = \app$ as both are maps into the terminal object.
\end{itemize}

We can use the structure described above to model the negative
connectives of \textsf{WS1}. We will represent positive connectives
indirectly, inspired by the fact that strategies on the positive game
$P$ correspond to strategies on the negative game $\uparrow P =
P^\perp \multimap o$ where $o$ is the one-move game $\bot$. The object
$o$ satisfies a special property: an internalised version of
\emph{linear functional extensionality} \cite{A_ADFC}.

\begin{definition}
An object $o$ in a sequoidal closed category satisfies \emph{linear functional extensionality} if the natural transformation $\lfe : (B \multimap o) \oslash A \Rightarrow (A \multimap B) \multimap o : \mathcal{C} \times \mathcal{C}^\op \rightarrow \mathcal{C}_s$ given by $\Lambda_s (\app_s \circ (\id \oslash \app) \circ (\id \oslash \sym) \circ \passoc^{-1})$ is an isomorphism. 
\end{definition}

\noindent The linear functional extensionality property is characteristic of our \emph{history sensitive, locally alternating}  
games model \cite{Lai_FPC}: it does not hold  in other sequoidal closed
categories (e.g. Conway games \cite{Lai_HOS}). 

Using linear functional extensionality we can give a natural
isomorphism $\mathsf{abs} : o \oslash A \cong o$ by noticing that $o
\oslash A \cong (I \multimap o) \oslash A \cong (A \multimap I)
\multimap o \cong I \multimap o \cong o$, and thus setting
$\mathsf{abs} = \unit_\multimap \circ ((\dist_\multimap^0)^{-1}
\multimap \id) \circ \lfe \circ (\unit_\multimap^{-1} \oslash \id)$.

\subsection{Coalgebraic Exponential Comonoid}

We next consider the categorical status of the exponential operator
$!$. We interpret the core introduction rules for the exponentials, and the key anamorphism rule, by requiring that it is the carrier for a \emph{final coalgebra} of the functor $X \mapsto N \oslash X$.

Recall that a coalgebra for a functor $F : \mathcal{C} \rightarrow
\mathcal{C}$ is an object $A$ and a map $A \rightarrow F(A)$. A
\emph{final coalgebra} is a terminal object in the category of
coalgebras, that is a coalgebra $\alpha : Z \rightarrow F(Z)$ such
that for any $f : A \rightarrow F(A)$ there is a unique $\leftmoon f
\rightmoon : A \rightarrow Z$ such that $\alpha \circ \leftmoon f
\rightmoon = F(\leftmoon f \rightmoon) \circ f$. 
\begin{diagram}
A & \rTo^{f} & F(A) \\
\dTo^{\leftmoon f \rightmoon} & & \dTo_{F(\leftmoon f \rightmoon)} \\
Z & \rTo_{\alpha} & F(Z) \\
\end{diagram}
We call $\leftmoon f \rightmoon$ the \emph{anamorphism} of 
$f$. Note in particular that if $(Z,\alpha)$ is a final coalgebra for $F$, then $\alpha$ is an isomorphism, with inverse $\alpha^{-1} = \leftmoon F(\alpha) \rightmoon$.

 In $\cal{W}$ we define a coalgebra $(\mathop{!N},\alpha)$ by taking   $\alpha : \mathop{!}N \rightarrow N \oslash \mathop{!}N$ to be the evident copycat strategy which relabels $\inn_1(a)$ on the right to $(a,1)$ on
the left and $\inn_2(a,n)$ on the right to $(a,n+1)$ on the left.

\begin{proposition}
  $(\mathop{!}N, \alpha)$ is the final coalgebra of the functor $N \oslash \_$
  in the category $\mathcal{G}$.
  \label{bangcoalgG}
\end{proposition}
\begin{proof}
  Let $\sigma : M \rightarrow N \oslash M$. Define $\leftmoon \sigma
  \rightmoon_n : M \rightarrow (N \oslash \_)^n(M)$ by $\leftmoon
  \sigma \rightmoon_0 = \id$ and $\leftmoon \sigma \rightmoon_{n+1} =
  (\id \oslash \_)^n(\sigma) \circ \leftmoon \sigma \rightmoon_n$.
  \begin{diagram}
  M & \rTo^{\leftmoon \sigma \rightmoon_n} & (N \oslash \_)^n(M) & \rTo^{(\id \oslash \_)^n(\sigma)} & (N \oslash \_)^n(N \oslash M) = (N \oslash \_)^{n+1}(M) \\
  \end{diagram}

  The strategy $\leftmoon \sigma \rightmoon_n$ is a partial
  approximant to $\leftmoon \sigma \rightmoon : M \rightarrow
  \mathop{!} N$. 
  We can show by induction on $n$ that $\leftmoon \sigma
  \rightmoon_{n+1} = (\id \oslash \leftmoon \sigma \rightmoon_n) \circ
  \sigma$.
  Similarly, we can define $\alpha_k : \mathop{!}N \cong (N \oslash
  \_)^k(\mathop{!}N) : \alpha_k^{-1}$ by performing the above
  construction on $\alpha$. Consider the sequence of maps $M
  \rightarrow \mathop{!}N$ defined by $s_k= \alpha_k^{-1} \circ (\id
  \oslash \_)^k(\epsilon) \circ \leftmoon \sigma \rightmoon_k$ for $k
  \in \omega$. We can show that $s_{k + 1} \sqsupseteq s_k$ by
  induction on $k$, and so $(s_k)$ is a chain. Set $\leftmoon \sigma
  \rightmoon = \bigsqcup \alpha_k^{-1} \circ (\id \oslash
  \_)^k(\epsilon) \circ \leftmoon \sigma \rightmoon_k$, where
  $\epsilon$ is the empty strategy. It is well-known that
  $\mathcal{G}$ is cpo-enriched with bottom element $\epsilon$
  \cite{Lai_FPC}.

  We wish to show that $\leftmoon \sigma \rightmoon$ is the unique
  strategy such that $\alpha \circ \leftmoon \sigma \rightmoon = (\id
  \oslash \leftmoon \sigma \rightmoon) \circ \sigma$. To show that the
  equation holds, note that $\alpha \circ \leftmoon \sigma \rightmoon
  = \alpha \circ \bigsqcup \alpha_k^{-1} \circ (\id \oslash
  \_)^k(\epsilon) \circ \leftmoon \sigma \rightmoon_k = \alpha \circ
  \bigsqcup \alpha_{k+1}^{-1} \circ (\id \oslash \_)^{k + 1}(\epsilon)
  \circ \leftmoon \sigma \rightmoon_{k + 1} = \bigsqcup \alpha \circ
  \alpha_{k + 1}^{-1} \circ (\id \oslash \_)^{k + 1}(\epsilon) \circ
  \leftmoon \sigma \rightmoon_{k + 1} = \bigsqcup (\id \oslash
  \alpha_k^{-1}) \circ (\id \oslash (\id \oslash \_)^k(\epsilon))
  \circ (\id \oslash \leftmoon \sigma \rightmoon_k) \circ \sigma =
  (\id \oslash \bigsqcup(\alpha_k^{-1} \circ (\id \oslash
  \_)^k(\epsilon) \circ \leftmoon \sigma \rightmoon_k)) \circ \sigma =
  (\id \oslash \leftmoon \sigma \rightmoon) \circ \sigma$.

  For uniqueness, suppose that $\gamma : M \rightarrow \mathop{!} N$
  is such that $\alpha \circ \gamma = (\id \oslash \gamma) \circ
  \sigma$. We wish to show that $\gamma = \leftmoon \sigma \rightmoon
  = \bigsqcup \alpha_k^{-1} \circ (\id \oslash \_)^k(\epsilon) \circ
  \leftmoon \sigma \rightmoon_k$. To see that $\gamma \sqsupseteq
  \leftmoon \sigma \rightmoon$, it suffices to show that $\gamma$ is
  an upper bound of the chain, i.e. $\gamma \sqsupseteq \alpha_k^{-1}
  \circ (\id \oslash \_)^k(\epsilon) \circ \leftmoon \sigma
  \rightmoon_k$ for each $k$. This can be shown using a simple induction on $k$.
  To see that $\gamma \sqsubseteq \leftmoon \sigma \rightmoon$, we
  show that each play in $\gamma$ is also in $\leftmoon \sigma
  \rightmoon$. Consider a play $s \in \gamma : M \rightarrow
  \mathop{!}N$. Since $s$ is finite, it must visit only a finite
  number of copies of $N$ --- say, $k$ copies. Then $s$ is also a play
  in $\alpha_{k}^{-1} \circ (\id \oslash \_)^k(\epsilon) \circ
  \alpha_k \circ \gamma$.

  It is thus sufficient to show that $(\id \oslash \_)^k(\epsilon)
  \circ \alpha_k \circ \gamma \sqsubseteq (\id \oslash \_)^k(\epsilon)
  \circ \leftmoon \sigma \rightmoon_k$. This is achieved by a simple
  induction on $k$. \qed

\end{proof}

\begin{proposition}
  $(\mathop{!}N, \alpha)$ is the final coalgebra of $N \oslash \_$
  in the category $\mathcal{W}$.
  \label{bangcoalgW}
\end{proposition}
\begin{proof}
 It suffices to show that if  $\sigma : M \rightarrow N \oslash M$ is a winning strategy, then $\leftmoon \sigma \rightmoon$ is winning.

  To see that $\leftmoon \sigma \rightmoon$ is total, let $s \in
  \leftmoon \sigma \rightmoon$ and $so \in P_{\mathop{!}N}$. Then $so$
  visits only finite $k$ many copies of $N$, and so up to retagging it
  is a play in $M \rightarrow (N \oslash \_)^k(M)$, and $s$ a play in
  $\leftmoon \sigma \rightmoon_k$. By totality of $\leftmoon \sigma
  \rightmoon_k$, there is a move $p$ with $sop \in \leftmoon \sigma
  \rightmoon_k$. Then, up to retagging, $sop$ is also a play in
  $\leftmoon \sigma \rightmoon$.

  We next need to check that each infinite play with all even prefixes
  in $\leftmoon \sigma \rightmoon$ is winning. Let $s$ be such an
  infinite play, with $s|_M$ winning. We must show that
  $s|_{\mathop{!}N}$ is winning, i.e. $s|_{(N,i)}$ is winning for each
  $i$. The infinite play $s$ corresponds to an infinite interaction
  sequence:

  \begin{diagram}
    M & \rTo^\sigma & N \oslash M & \rTo^{\id \oslash \sigma} & N \oslash (N \oslash M) & \rTo^{\id \oslash (\id \oslash \sigma)} & \ldots \\
    \vdots \\
  \end{diagram}

  Then $s|_{(N,i)}$ can also be found in the $i$th column of the above
  interaction sequence. By hiding all columns other than the first and
  the $i$th, we see a play in $M \rightarrow (N \oslash \_)^i(M)$
  in $\leftmoon \sigma \rightmoon_i$. The first column is $s|_M$
  (which is winning), and the $i$th component of the second is
  $s|_{(N,i)}$. Since $\leftmoon \sigma \rightmoon_i$ is a
  winning strategy, this play is winning, by the winning condition for
  $\oslash$. \qed
\end{proof}

Recall that the monoidal unit of a distributive sequoidal category is a terminal object. Thus we may define operations corresponding to dereliction and promotion:
\begin{itemize} 
\item $\der_N:\mathop{!}N \rightarrow N  = \unit_\oslash \circ (\id \oslash \term) \circ \alpha$.
\item Given any symmetric comonoid $(B,\eta,\delta)$, and morphism $f : B \rightarrow N$, let $f^\dagger: B \rightarrow !N$ be the (comonoid morphism) $\leftmoon \wk \circ (f \otimes \id) \circ \delta \rightmoon$.   
\end{itemize}

To interpret  the contraction rule, we  require a further coalgebraic property.
\begin{definition}A decomposable, distributive  sequoidal category $\mathcal{C}$ has \emph{coalgebraic monoidal exponentials} if:
\begin{itemize}
\item For any object $A$, the  endofunctor $A \oslash \_$ has a specified  final coalgebra   $(\mathop{!}A,\alpha_A)$.
\item For any objects $A,B$, $(!A \otimes !B,\alpha_{A,B})$ is a final coalgebra for the endofunctor   $(A \times B) \oslash \_$, where $\alpha_{A,B}: !A \otimes !B \rightarrow (A \times B) \oslash (!A \otimes !B)$ is the isomorphism:
  \begin{eqnarray*}
& & !A \otimes !B \cong (!A \oslash !B) \times (!B \oslash !A) \cong ((A
\oslash !A) \oslash !B) \times ((B \oslash !B) \oslash !A) \\
&  \cong & (A \oslash (!A \otimes !B)) \times (B \oslash (!A \otimes !B)) \cong (A
\times B) \oslash (!A \otimes !B)
  \end{eqnarray*}
\end{itemize}
\end{definition}
The second requirement is equivalent to requiring that the morphism $\langle \der_A \otimes \term,\term \otimes \der_B\rangle^\dagger$ from $!A \otimes !B$ to $!(A \times B)$ is an isomorphism. Thus we may define a comonoid $(!A,\delta:!A \rightarrow !A \otimes !A,\term: A \rightarrow I)$, where $\delta$ is the anamorphism of the map $\dist_{A,A,!A} \circ \langle \alpha_A, \alpha_A\rangle$.     
\begin{proposition}If $\C$ has coalgebraic monoidal exponentials then  $(!A,\delta,\term)$ is the cofree commutative comonoid on $A$.  

\end{proposition}
\begin{proof}In other words, the forgetful functor from the category of comonoids on $\C$ into the category $\C$ has a left adjoint which sends $A$ to $(!A,\delta,\term)$. The unit  of this adjunction is the dereliction $\der_A:{!}A \rightarrow A$: for any 
$f:B \rightarrow A$, $f^\dagger:B \rightarrow !A$ is the unique comonoid morphism such that $\der \circ f^\dagger$. (Uniqueness follows from finality of $!A$.) \qed
\end{proof}

This cofree commutative comonoid can also be constructed using the
technique described in \cite{MTT_EFE}. This approach builds the
exponential as a limit of finitary \emph{symmetric tensor powers},
that is, finite tensor products subject to a quotient so that the
order that the components are played in is irrelevant. Our use of the
asymmetric $\oslash$ enforces a strict left-to-right order, providing
a concrete (albeit less generally applicable) alternative to such quotienting.

\begin{proposition}
The sequoidal closed categories $\mathcal{W}$ and $\mathcal{G}$ are both equipped with coalgebraic monoidal exponentials.
\label{coexpcomWG}
\end{proposition}
\begin{proof}
  Follows from Propositions \ref{bangcoalgG}, \ref{bangcoalgW} and the
  fact $\mathop{!}$ is the cofree commutative comonoid in $\mathcal{G}$
  and $\mathcal{W}$ \cite{Lai_FPC}. \qed
\end{proof}

\begin{definition}
  A \emph{WS!-category} is a distributive, decomposable sequoidal
  closed category with an object $o$ satisfying linear functional
  extensionality and  coalgebraic monoidal exponentials.
\end{definition}

\begin{proposition}
  The categories $(\mathcal{G},\mathcal{G}_s)$ and
  $(\mathcal{W},\mathcal{W}_s)$ enjoy the structure of an WS!-category.
\label{WScats}
\end{proposition}

\subsection{Semantics of Rules}
We may now describe the interpretation of the rules of our logic
(other than those for atoms, quantifiers and equality) in a
\textsf{WS!}-category $\C$. Suppose that, for a given context of
variables and atoms $\Phi$, we have an interpretation of formulas and
sequents over $\Phi$ as objects of $\C$, satisfying the following: 


\[
\begin{array}{lcllcl}
  \llbracket \Phi \vdash \mathbf{1} \rrbracket & = & I & \llbracket \Phi \vdash \mathbf{0} \rrbracket & = & I \\
  \llbracket \Phi \vdash \bot \rrbracket & = & o &  \llbracket \Phi \vdash \top \rrbracket & = & o \\
  \llbracket \Phi \vdash M \otimes N
  \rrbracket & = & \llbracket \Phi \vdash M \rrbracket \otimes \llbracket \Phi \vdash N \rrbracket & \llbracket \Phi \vdash P \parr Q \rrbracket & = & \llbracket \Phi \vdash P \rrbracket
  \otimes \llbracket \Phi \vdash Q \rrbracket \\
  \llbracket \Phi \vdash M \& N \rrbracket & = & \llbracket \Phi \vdash M \rrbracket \times \llbracket \Phi \vdash N \rrbracket & \llbracket \Phi \vdash P \oplus Q \rrbracket & = & \llbracket \Phi \vdash P \rrbracket \times
  \llbracket\Phi \vdash  Q \rrbracket \\
  \llbracket \Phi \vdash M \oslash N \rrbracket & = & \llbracket \Phi \vdash M \rrbracket \oslash
  \llbracket \Phi \vdash N \rrbracket &
  \llbracket \Phi \vdash P \lhd Q \rrbracket & = & \llbracket \Phi \vdash P \rrbracket \oslash
  \llbracket \Phi \vdash Q \rrbracket \\
  \llbracket \Phi \vdash M \lhd Q \rrbracket & = & \llbracket \Phi \vdash Q \rrbracket \multimap \llbracket \Phi \vdash M \rrbracket &
  \llbracket \Phi \vdash P \oslash N \rrbracket & = & \llbracket \Phi \vdash N \rrbracket \multimap 
  \llbracket \Phi \vdash P \rrbracket \\ 
  \llbracket \Phi \vdash !N \rrbracket & = & !\llbracket \Phi \vdash N \rrbracket & \llbracket \Phi \vdash ?P \rrbracket & = & ! \llbracket \Phi \vdash P \rrbracket \\
\end{array}
\]
\[
\begin{array}{lcl}


\llbracket \Phi \vdash M , \Gamma , N \rrbracket & = & \llbracket \Phi \vdash M , \Gamma \rrbracket \oslash \llbracket \Phi \vdash N \rrbracket \\
\llbracket \Phi \vdash M , \Gamma , P \rrbracket & = & \llbracket \Phi \vdash P \rrbracket \multimap \llbracket \Phi \vdash M , \Gamma \rrbracket \\

\llbracket \Phi \vdash P , \Gamma , N \rrbracket & = & \llbracket \Phi \vdash N \rrbracket \multimap \llbracket \Phi \vdash P , \Gamma \rrbracket \\
\llbracket \Phi \vdash P , \Gamma , Q \rrbracket & = & \llbracket \Phi \vdash P , \Gamma \rrbracket \oslash \llbracket \Phi \vdash Q \rrbracket \\

\end{array}
\]

(For atom and quantifier-free formulas, these equations \emph{define} an interpretation of formulas and sequents in $\C$.) Then we may give an interpretation of each proof rule except those for atoms, quantifiers, and equality  as an operation on morphisms in $\C$. These typically involve an operation on the head formula of the sequence ``under'' a context consisting of its tail, and so we define distributivity maps to allow this:

We define endofunctors  $\llbracket \Gamma \rrbracket^b$ on $\mathcal{\C}_s$ for each context (possibly empty list of formulas) $\Gamma$ and $b \in
\{ +, - \}$.
below. \\ \vspace{1pt}
\[
\begin{array}{lclclcl}

\llbracket \epsilon \rrbracket^+ & = & \id & &
\llbracket \epsilon \rrbracket^- & = & \id \\
\llbracket \Gamma, M \rrbracket^+ & = & \llbracket M \rrbracket \multimap \llbracket \Gamma \rrbracket^+ & &
\llbracket
\Gamma, P \rrbracket^- & = & \llbracket P \rrbracket \multimap \llbracket
\Gamma \rrbracket^- \\
\llbracket \Gamma, P \rrbracket^+ & = & \llbracket \Gamma \rrbracket^+ \oslash \llbracket P \rrbracket & &
\llbracket \Gamma, M \rrbracket^- & = & \llbracket
\Gamma \rrbracket^- \oslash \llbracket M \rrbracket \\ \vspace{1pt}

\end{array}
\]
\begin{proposition}
  For any sequent $A , \Gamma$ we have $\llbracket A, \Gamma
  \rrbracket = \llbracket \Gamma \rrbracket^b(\llbracket A
  \rrbracket)$ where $b$ is the polarity of $A$.
\end{proposition}
\begin{proof}
A simple induction on $\Gamma$. \qed
\end{proof}

\begin{proposition}
  For any context $\Gamma$, $\llbracket \Gamma \rrbracket^b$ preserves
  products. \label{contprod}
\end{proposition}
\begin{proof}
 Using the distributivity of $\times$ over $\oslash$ and $\multimap$, we can construct isomorphisms $\mathsf{dist_{b,\Gamma}} : \llbracket
  \Gamma \rrbracket^b(A \times B) \cong \llbracket \Gamma
  \rrbracket^b(A) \times \llbracket \Gamma \rrbracket^b (B)$ and
  $\mathsf{dist_{b,\Gamma}^{0}} : \llbracket \Gamma \rrbracket^b(I)
  \cong I$ by induction on
  $\Gamma$. \\

\end{proof}

\subsection{Semantics of Proof Rules}
Define $\sigma:\llbracket\vdash \Gamma\rrbracket$ if:
\begin{itemize}
\item $\Gamma = N,\Gamma'$, and $\sigma:I \rightarrow   \llbracket \Gamma \rrbracket$ in $\C$.
\item $\Gamma = P,\Gamma'$, and $\sigma:\llbracket \Gamma \rrbracket \rightarrow o$ in $\C$.
\end{itemize} 
Semantics of the core rules as operations on morphisms are given in Figure \ref{WS-sem} and the
other rules in Figures \ref{WS-sem2} and \ref{WS-sem3}. The rules
involving the exponential are treated separately in Figure
\ref{WS!-sem}. Note that in each case, the interpretation in the WS!-category $\W$ agrees with the informal exposition in Section \ref{proofinterp}.

In the semantics of $\p_\cut$ we use an additional construction. If $\tau : I \rightarrow \llbracket N, \Delta \rrbracket$ define (strict) $\tau^{\circ-}_{M,\Gamma} : \llbracket M , \Gamma , N^\perp \rrbracket
\rightarrow \llbracket M , \Gamma , \Delta \rrbracket$ to be $\unit_\multimap \circ (\tau \multimap \id_{\llbracket M,\Gamma \rrbracket})$ if $|\Delta| = 0$ and $\passoc_\multimap^{n} \circ (\Lambda^{-n}\Lambda_I^{-1} \tau \multimap \id_{\llbracket M,\Gamma \rrbracket})$ if $|\Delta| = n + 1$.  Define (strict) $\tau^{\circ+}_{P,\Gamma} : \llbracket P , \Gamma , \Delta \rrbracket \rightarrow \llbracket P , \Gamma , N^\perp \rrbracket$ to be $(\id_{\llbracket P,\Gamma \rrbracket} \oslash \tau) \circ \unit_\oslash^{-1}$ if $|\Delta| = 0$ and $(\id \oslash \Lambda^{-n}\Lambda_I^{-1} \tau) \circ ((\id_{\llbracket P,\Gamma  \rrbracket} \oslash \sym) \circ \passoc^{-1})^{n}$ if $|\Delta| = n + 1$. In some of the rules in Figure \ref{WS-sem3} we omit some $\passoc$ isomorphisms for clarity.

\begin{figure*}
\caption{Categorical Semantics for \textsf{WS1} (core rules)}
\begin{small}
\centering
\label{WS-sem}

\vspace{1ex}
\hrule
\hrule

\begin{tabular}{cc}
\\[0.2ex]

\AxiomC{}
\LeftLabel{$\p_\mathbf{1}$}
\UnaryInfC{$(\dist_{-,\Gamma}^0)^{-1} : \llbracket \vdash \mathbf{1}, \Gamma \rrbracket$}
\DisplayProof

&

\AxiomC{}
\LeftLabel{$\p_\top$}
\UnaryInfC{$\id_o : \llbracket \vdash \top \rrbracket$}
\DisplayProof

\\[3ex]

\AxiomC{$\sigma : \llbracket \vdash M,N,\Gamma \rrbracket$}
\AxiomC{$\tau : \llbracket \vdash N,M,\Gamma \rrbracket$}
\LeftLabel{$\p_\otimes$}
\BinaryInfC{$\llbracket \Gamma \rrbracket^-(\dec^{-1}) \circ \dist_{-,\Gamma}^{-1} \circ \langle \sigma, \tau \rangle : \llbracket \vdash M \otimes N, \Gamma \rrbracket$}
\DisplayProof

&

\AxiomC{$\sigma : \llbracket \vdash M, \Gamma \rrbracket$}
\AxiomC{$\tau : \llbracket \vdash N, \Gamma \rrbracket$}
\LeftLabel{$\p_\&$}
\BinaryInfC{$ \dist_{-,\Gamma}^{-1} \circ \langle \sigma, \tau \rangle : \llbracket \vdash M \& N, \Gamma \rrbracket$}
\DisplayProof

\\[3ex]

\AxiomC{$ \sigma : \llbracket \vdash Q,P,\Gamma \rrbracket$}
\LeftLabel{${\p_\parr}_2$}
\UnaryInfC{$ \sigma \circ \llbracket \Gamma \rrbracket^+(\wk \circ \sym): \llbracket \vdash P \parr Q, \Gamma \rrbracket$}
\DisplayProof

&

\AxiomC{$ \sigma : \llbracket \vdash P,Q,\Gamma \rrbracket$}
\LeftLabel{${\p_\parr}_1$}
\UnaryInfC{$ \sigma \circ \llbracket \Gamma \rrbracket^+(\wk): \llbracket \vdash P \parr Q, \Gamma \rrbracket$}
\DisplayProof

\\[3ex]

\AxiomC{$ \sigma : \llbracket \vdash P, \Delta \rrbracket$}
\LeftLabel{${\p_\oplus}_1$}
\UnaryInfC{$\sigma \circ \llbracket \Delta \rrbracket^+(\pi_1): \llbracket \vdash P \oplus Q, \Delta \rrbracket $}
\DisplayProof

&

\AxiomC{$ \sigma : \llbracket \vdash Q, \Delta \rrbracket$}
\LeftLabel{${\p_\oplus}_2$}
\UnaryInfC{$\sigma \circ \llbracket \Delta \rrbracket^+(\pi_2): \llbracket \vdash P \oplus Q, \Delta \rrbracket $}
\DisplayProof

\\[3ex]

\AxiomC{$ \sigma : \llbracket \vdash \bot , P \parr Q, \Gamma \rrbracket$}
\LeftLabel{$\p_\bot^\parr$}
\UnaryInfC{$\llbracket \Gamma \rrbracket^- (\passoc_\multimap \circ (\sym \multimap \id)) \circ \sigma : \llbracket \vdash \bot , P, Q, \Gamma \rrbracket$}
\DisplayProof

&

\AxiomC{$ \sigma : \llbracket \vdash P \rrbracket $}
\LeftLabel{$\p_\bot^+$}
\UnaryInfC{$ \Lambda_I(\sigma) : \llbracket \vdash \bot , P \rrbracket$}
\DisplayProof

\\[3ex]

\AxiomC{$ \sigma : \llbracket \vdash \bot , P \oslash N, \Gamma \rrbracket$}
\LeftLabel{$\p_\bot^\oslash$}
\UnaryInfC{$\llbracket \Gamma \rrbracket^-(\lfe^{-1}) \circ \sigma : \llbracket \vdash \bot, P, N, \Gamma \rrbracket$}
\DisplayProof

&

\AxiomC{$ \sigma : \llbracket \vdash N \rrbracket$}
\LeftLabel{$\p_\top^-$}
\UnaryInfC{$ \unit_\multimap \circ (\sigma \multimap \id) : \llbracket \vdash \top , N \rrbracket$}
\DisplayProof

\\[3ex]

\AxiomC{$ \sigma : \llbracket \vdash \top , M \otimes N , \Gamma \rrbracket$}
\LeftLabel{$\p_\top^\otimes$}
\UnaryInfC{$ \sigma \circ \llbracket \Gamma \rrbracket^+((\sym \multimap \id) \circ \passoc_\multimap^{-1} ) : \llbracket \vdash \top, M, N, \Gamma \rrbracket$}
\DisplayProof

& 

\AxiomC{$ \sigma : \llbracket \vdash \top , N \lhd P , \Gamma \rrbracket $}
\LeftLabel{$\p_\top^\lhd$}
\UnaryInfC{$ \sigma \circ \llbracket \Gamma \rrbracket^+(\lfe) : \llbracket \vdash \top, N, P, \Gamma \rrbracket$}
\DisplayProof

\\[3ex]

\AxiomC{$\sigma : \llbracket \vdash \bot , \Gamma \rrbracket$}
\LeftLabel{$\p_\bot^-$}
\UnaryInfC{$\llbracket \Gamma \rrbracket^-(\mathsf{abs}^{-1}) \circ \sigma : \llbracket \vdash \bot , N , \Gamma \rrbracket$}
\DisplayProof

&

\AxiomC{$\sigma : \llbracket \vdash A,P,\Gamma \rrbracket$}
\LeftLabel{$\p_\lhd$}
\UnaryInfC{$\sigma : \llbracket \vdash A \lhd P, \Gamma \rrbracket$}
\DisplayProof

\\[3ex]

\AxiomC{$\sigma : \llbracket \top , \Gamma \rrbracket$}
\LeftLabel{$\p_\top^+$}
\UnaryInfC{$ \sigma \circ \llbracket \Gamma \rrbracket^+ (\mathsf{abs}): \llbracket \top , P , \Gamma \rrbracket$}
\DisplayProof

&

\AxiomC{$\sigma : \llbracket \vdash A,N,\Gamma \rrbracket $}
\LeftLabel{$\p_\oslash$}
\UnaryInfC{$\sigma : \llbracket \vdash A \oslash N, \Gamma \rrbracket$}
\DisplayProof

\\[2.5ex]

\end{tabular}

\hrule
\hrule
\end{small}
\end{figure*}

\begin{figure*}
\caption{Categorical Semantics for \textsf{WS1} (other rules, part 1)}
\begin{small}
\centering
\label{WS-sem2}

\vspace{1ex}
\hrule
\hrule

\begin{tabular}{cc}
\\[0.2ex]

\AxiomC{$\sigma : \llbracket \vdash M', \Gamma, M, N, \Delta \rrbracket$}
\UnaryInfC{$\llbracket \Delta \rrbracket^-(\psym) \circ \sigma : \llbracket \vdash M', \Gamma, N, M, \Delta \rrbracket$}
\DisplayProof

&

\AxiomC{$\sigma : \llbracket \vdash P, \Gamma, M, N, \Delta \rrbracket$}
\UnaryInfC{$\sigma \circ \llbracket \Delta \rrbracket^+(\psym_\multimap) : \llbracket \vdash P, \Gamma, N, M, \Delta \rrbracket$}
\DisplayProof

\\[4ex]

\AxiomC{$\sigma : \llbracket \vdash M, \Gamma, P, Q, \Delta \rrbracket$}
\UnaryInfC{$\llbracket \Delta \rrbracket^-(\psym_\multimap) \circ \sigma : \llbracket \vdash M, \Gamma, Q, P, \Delta \rrbracket$}
\DisplayProof

&

\AxiomC{$\sigma : \llbracket \vdash P', \Gamma, P, Q, \Delta \rrbracket$}
\UnaryInfC{$\sigma \circ \llbracket \Delta \rrbracket^+(\psym) : \llbracket \vdash P', \Gamma, Q, P, \Delta \rrbracket$}
\DisplayProof

\\[4ex]

\AxiomC{$\sigma : \llbracket \vdash P, \Gamma, M, \Delta \rrbracket$}
\UnaryInfC{$\sigma \circ \llbracket \Delta \rrbracket^+((\af \multimap \id ) \circ \unit_\multimap^{-1}) : \llbracket \vdash P, \Gamma, \Delta \rrbracket$}
\DisplayProof

&

\AxiomC{$\sigma : \llbracket \vdash N, \Gamma, M, \Delta \rrbracket$}
\UnaryInfC{$\llbracket \Delta \rrbracket^-(\unit_\oslash \circ (\id \oslash \af)) \circ \sigma : \llbracket \vdash N, \Gamma, \Delta \rrbracket$}
\DisplayProof

\\[4ex]

\AxiomC{$\sigma : \llbracket \vdash M, \Gamma, \Delta \rrbracket$}
\UnaryInfC{$\llbracket \Delta \rrbracket^-((\af \multimap \id) \circ \unit_\multimap^{-1}) \circ \sigma : \llbracket \vdash M, \Gamma, P, \Delta \rrbracket $}
\DisplayProof

&

\AxiomC{$\sigma : \llbracket \vdash P, \Gamma, \Delta^+ \rrbracket$}
\UnaryInfC{$\sigma \circ \llbracket \Delta \rrbracket^-(\unit_\oslash \circ (\id \oslash \af)): \llbracket \vdash P, \Gamma, Q, \Delta^+ \rrbracket$}
\DisplayProof

\\[4ex]

\AxiomC{$\sigma : \llbracket \vdash N, \Gamma, \Delta \rrbracket$}
\UnaryInfC{$\llbracket \Delta \rrbracket^-(\unit_\oslash^{-1}) \circ \sigma : \llbracket \vdash N, \Gamma, \mathbf{1}, \Delta \rrbracket$}
\DisplayProof

&

\AxiomC{$\sigma : \llbracket \vdash P, \Gamma, \Delta \rrbracket$}
\UnaryInfC{$\sigma \circ \llbracket \Delta \rrbracket^+(\unit_\multimap) : \llbracket \vdash P, \Gamma, \mathbf{1}, \Delta \rrbracket $}
\DisplayProof

\\[2.5ex]








\AxiomC{$\sigma : \llbracket \vdash M, \Gamma, \Delta \rrbracket$}
\UnaryInfC{$\llbracket \Delta \rrbracket^-(\unit_\multimap^{-1}) \circ \sigma : \llbracket \vdash M, \Gamma, \mathbf{0}, \Delta \rrbracket$}
\DisplayProof

&

\AxiomC{$\sigma : \llbracket \vdash P, \Gamma, \Delta \rrbracket$}
\UnaryInfC{$\sigma \circ \llbracket \Delta \rrbracket^+(\unit_\oslash) : \llbracket \vdash P, \Gamma, \mathbf{0}, \Delta \rrbracket$}
\DisplayProof

\\[4ex]

\AxiomC{$\sigma : \llbracket \vdash M', \Gamma, M, N, \Delta \rrbracket$}
\UnaryInfC{$\llbracket \Delta \rrbracket^-(\passoc)\circ \sigma : \llbracket \vdash M', \Gamma, M \otimes N, \Delta \rrbracket$}
\DisplayProof

&

\AxiomC{$\sigma : \llbracket \vdash P, \Gamma, M, N, \Delta \rrbracket$}
\UnaryInfC{$\sigma \circ \llbracket \Delta \rrbracket^+(\passoc_\multimap) : \llbracket \vdash P, \Gamma, M \otimes N, \Delta \rrbracket$}
\DisplayProof

\\[4ex]

\AxiomC{$\sigma : \llbracket \vdash P', \Gamma, P, Q, \Delta \rrbracket$}
\UnaryInfC{$\sigma \circ \llbracket \Delta \rrbracket^+(\passoc^{-1}) : \llbracket \vdash P', \Gamma, P \parr Q, \Delta \rrbracket$}
\DisplayProof

&

\AxiomC{$\sigma : \llbracket \vdash M, \Gamma, P, Q, \Delta \rrbracket$}
\UnaryInfC{$\llbracket \Delta \rrbracket^+(\passoc_\multimap^{-1}) \circ \sigma : \llbracket \vdash M, \Gamma, P \parr Q, \Delta \rrbracket$}
\DisplayProof

\\[4ex]

\AxiomC{$\sigma : \llbracket \vdash M, \Gamma, P_i, \Delta \rrbracket$}
\UnaryInfC{$\llbracket \Delta \rrbracket^-(\pi_i \multimap \id) \circ \sigma : \llbracket \vdash M, \Gamma, P_1 \oplus P_2, \Delta \rrbracket$}
\DisplayProof

&

\AxiomC{$\sigma : \llbracket \vdash Q, \Gamma, P_i, \Delta \rrbracket$}
\UnaryInfC{$\sigma \circ \llbracket \Delta \rrbracket^+(\id \oslash \pi_i) : \llbracket \vdash Q, \Gamma, P_1 \oplus P_2, \Delta \rrbracket$}
\DisplayProof

\\[4ex]

\AxiomC{$\sigma : \llbracket \vdash N, \Gamma, M_1 \& M_2, \Delta \rrbracket$}
\UnaryInfC{$\llbracket \Delta \rrbracket^-(\id \oslash \pi_i) \circ \sigma : \llbracket \vdash N, \Gamma, M_i, \Delta \rrbracket$}
\DisplayProof

&

\AxiomC{$\sigma : \llbracket \vdash Q, \Gamma, M_1 \& M_2, \Delta \rrbracket$}
\UnaryInfC{$\sigma \circ \llbracket \Delta \rrbracket^+(\pi_i \multimap \id) : \llbracket \vdash Q, \Gamma, M_i, \Delta \rrbracket$}
\DisplayProof

\\[2.5ex]

\end{tabular}

\vspace{4ex}

\hrule
\hrule
\end{small}
\end{figure*}

\begin{figure*}
\caption{Categorical Semantics for \textsf{WS1} (other rules, part 2)}
\begin{small}
\centering
\label{WS-sem3}

\vspace{1ex}
\hrule
\hrule

\begin{tabular}{cc}
\\[0.2ex]

\AxiomC{$\sigma : \llbracket \vdash M, \Gamma, \Delta^+ \rrbracket$}
\AxiomC{$\tau : \llbracket \vdash N, \Delta_1^+ \rrbracket$}
\LeftLabel{$\p_\mul$}
\BinaryInfC{$ \Lambda_I\Lambda(\wk \circ (\Lambda_I^{-1}(\sigma) \otimes \Lambda_I^{-1}(\tau))) : \llbracket \vdash M, \Gamma, N, \Delta^+, \Delta_1^+ \rrbracket$}
\DisplayProof

\\[4ex]

\AxiomC{$\sigma : \llbracket \vdash M, \Gamma, N^\perp, \Gamma_1 \rrbracket$}
\AxiomC{$\tau : \llbracket \vdash N, \Delta^+ \rrbracket$}
\LeftLabel{$\p_\cut$}
\BinaryInfC{$\llbracket \Gamma_1 \rrbracket^-(\tau^{\circ-}_{M,\Gamma}) \circ \sigma : \llbracket \vdash M , \Gamma, \Delta^+ , \Gamma_1 \rrbracket $}
\DisplayProof

\\[4ex]

\AxiomC{$\sigma : \llbracket \vdash P, \Gamma, N^\perp, \Gamma_1 \rrbracket$}
\AxiomC{$\tau : \llbracket \vdash N, \Delta^+ \rrbracket$}
\LeftLabel{$\p_\cut$}
\BinaryInfC{$\sigma \circ \llbracket \Gamma_1 \rrbracket^+(\tau^{\circ+}_{P,\Gamma}) : \llbracket \vdash P , \Gamma, \Delta^+ , \Gamma_1 \rrbracket $}
\DisplayProof

\\[4ex]

\AxiomC{}
\LeftLabel{$\p_\id$}
\UnaryInfC{$\Lambda_I(\id) : \llbracket \vdash N , N^\perp \rrbracket$}
\DisplayProof

\\[2.5ex]

\AxiomC{$\sigma : \llbracket \vdash N^\perp \rrbracket$}
\AxiomC{$\tau : \llbracket \vdash N , Q \rrbracket$}
\LeftLabel{$\p_\cut^0$}
\BinaryInfC{$\sigma \circ \Lambda_I^{-1}(\tau) : \llbracket \vdash Q \rrbracket$}
\DisplayProof

\\[4ex]

\AxiomC{$\sigma : \llbracket \vdash M , \Gamma , P \rrbracket$}
\AxiomC{$\tau : \llbracket \vdash N , \Delta^+ \rrbracket$}
\LeftLabel{$\p_\multimap$}
\BinaryInfC{$\psym_\multimap \circ \Lambda_I(\Lambda_I^{-1}(\tau) \multimap \Lambda_I^{-1}(\sigma)) : \llbracket \vdash M , \Gamma , P \oslash N , \Delta^+ \rrbracket$}
\DisplayProof

\\[4ex]

\AxiomC{$\sigma : \llbracket \vdash N , Q , \Delta^+ \rrbracket$}
\LeftLabel{${\p_\id}_\oslash$}
\UnaryInfC{$\Lambda_I\Lambda((\id \oslash \Lambda^{-1}\Lambda_I^{-1}(\sigma) \circ \sym) \circ \passoc_\oslash \circ \wk \circ \sym) : \llbracket \vdash M , N , M^\perp \lhd Q , \Delta^+ \rrbracket$}
\DisplayProof

\\[2.5ex]

\end{tabular}
\hrule
\hrule
\end{small}
\end{figure*}

\begin{figure*}[ht]
\caption{Semantics for \textsf{WS1} --- Exponential Rules}
\begin{small}
\centering
\label{WS!-sem}

\vspace{1ex}
\hrule
\hrule
Core rules: \\
\begin{tabular}{cc}
\\[0.2ex]

\AxiomC{$\sigma : \llbracket \vdash N , !N , \Gamma \rrbracket$}
\LeftLabel{$\p_!$}
\UnaryInfC{$\llbracket \Gamma \rrbracket^-(\alpha^{-1}) \circ \sigma : \llbracket \vdash !N , \Gamma \rrbracket$}
\DisplayProof

&

\AxiomC{$\sigma : \llbracket \vdash P , ?P , \Gamma \rrbracket$}
\LeftLabel{$\p_?$}
\UnaryInfC{$\sigma \circ \llbracket \Gamma \rrbracket^+(\alpha) : \llbracket \vdash ?P , \Gamma \rrbracket$}
\DisplayProof

\\[2.5ex]

\end{tabular}

Other rules: \\

\begin{prooftree}
\AxiomC{$\sigma : \llbracket \vdash M , P^\perp , P \rrbracket$}
\LeftLabel{$\p_{\mathsf{ana}}$}
\UnaryInfC{$\Lambda_I(\leftmoon \Lambda_I^{-1}(\sigma) \rightmoon) : \llbracket \vdash !M , P \rrbracket$}
\end{prooftree} 

\begin{tabular}{cc}

\AxiomC{$\sigma : \llbracket \vdash P , \Gamma , !M , \Delta \rrbracket$}
\UnaryInfC{$\sigma \circ \llbracket \Delta \rrbracket^+(\der \multimap \id): \llbracket \vdash P , \Gamma , M , \Delta \rrbracket$}
\DisplayProof

&

\AxiomC{$\sigma : \llbracket \vdash !M , \Delta \rrbracket$}
\UnaryInfC{$\llbracket \Delta \rrbracket^-(\der) \circ \sigma : \llbracket \vdash M , \Delta \rrbracket$}
\DisplayProof

\\[2.5ex]

\AxiomC{$\sigma : \llbracket \vdash P , \Gamma , !M , \Delta \rrbracket$}
\UnaryInfC{$\sigma \circ \llbracket \Delta \rrbracket^+((\delta \multimap \id) \circ \passoc_\multimap^{-1}): \llbracket \vdash P , \Gamma , !M , !M , \Delta \rrbracket$}
\DisplayProof

&

\AxiomC{$\sigma : \llbracket \vdash N , \Gamma , !M , \Delta \rrbracket$}
\UnaryInfC{$\llbracket \Delta \rrbracket^-(\id \oslash \der) \circ \sigma : \llbracket \vdash N , \Gamma , M , \Delta \rrbracket$}
\DisplayProof

\\[2.5ex]

\AxiomC{$\sigma : \llbracket \vdash N , \Gamma , !M , \Delta \rrbracket$}
\UnaryInfC{$\llbracket \Delta \rrbracket^-(\passoc \circ (\id \oslash \delta)) \circ \sigma : \llbracket \vdash N , \Gamma , !M , !M , \Delta \rrbracket$}
\DisplayProof

&

\AxiomC{$\sigma : \llbracket \vdash !M , \Delta \rrbracket$}
\UnaryInfC{$\llbracket \Delta \rrbracket^-(\con) \circ \sigma : \llbracket \vdash !M , !M , \Delta \rrbracket$}
\DisplayProof

\\[2.5ex]

\AxiomC{$\sigma : \llbracket \vdash M , \Gamma , ?P , ?P , \Delta \rrbracket$}
\UnaryInfC{$\llbracket \Delta \rrbracket^-((\delta \multimap \id) \circ \passoc_\multimap^{-1}) \circ \sigma: \llbracket \vdash M , \Gamma , ?P , \Delta \rrbracket$}
\DisplayProof

&

\AxiomC{$\sigma : \llbracket \vdash ?P , ?P , \Delta \rrbracket$}
\UnaryInfC{$\sigma \circ \llbracket \Delta \rrbracket^+(\con) : \llbracket \vdash ?P , \Delta \rrbracket$}
\DisplayProof

\\[2.5ex]

\AxiomC{$\sigma : \llbracket \vdash Q , \Gamma , ?P , ?P , \Delta \rrbracket$}
\UnaryInfC{$\sigma \circ \llbracket \Delta \rrbracket^+(\passoc \circ (\id \oslash \delta)) : \llbracket \vdash Q , \Gamma , ?P , \Delta \rrbracket$}
\DisplayProof

&

\AxiomC{$\sigma : \llbracket \vdash P , \Delta \rrbracket$}
\UnaryInfC{$\sigma \circ \llbracket \Delta \rrbracket^+(\der) : \llbracket \vdash ?P , \Delta \rrbracket$}
\DisplayProof

\\[2.5ex]

\AxiomC{$\sigma : \llbracket \vdash Q , \Gamma , P , \Delta \rrbracket$}
\UnaryInfC{$\sigma \circ \llbracket \Delta \rrbracket^+(\id \oslash \der) : \llbracket \vdash Q , \Gamma , ?P , \Delta \rrbracket$}
\DisplayProof

&

\AxiomC{$\sigma : \llbracket \vdash M , \Gamma , P , \Delta \rrbracket$}
\UnaryInfC{$\llbracket \Delta \rrbracket^-(\der \multimap \id) \circ \sigma: \llbracket \vdash M , \Gamma , ?P , \Delta \rrbracket$}
\DisplayProof

\\[2.5ex]

\end{tabular}

\hrule
\hrule
\end{small}
\end{figure*}

\section{Semantics of atoms, quantifiers and equality}
We shall now complete the semantics of \textsf{WS1} by interpreting atoms and quantifiers based on our categories of games and strategies. (The requisite structure could be axiomatised for any WS!-category, but we shall not do so here.) 
We have seen that a sequent $X;\Theta \vdash \Gamma$ of \textsf{WS1}
can be interpreted as a family of games, indexed over
$\Theta$-satisfying $\mathcal{L}$-models over $X$. We shall interpret a
proof of $X ; \Theta \vdash \Gamma$ as a uniform family of strategies for each such game. 

For example, the family denoted by $\top \oslash (\phi \lhd \top)$ has games of the following form:

\begin{center}
\includegraphics[scale=0.6]{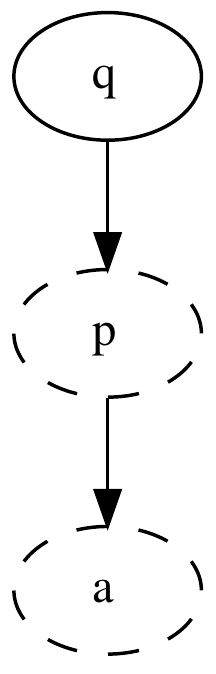}
\end{center}

\noindent Here we represent the forest of plays $P_A$ directly. The moves in dotted circles are only available if $(L,v) \models \overline{\phi}$. There is a unique total strategy on the (positive) game above in both cases, and this family is uniform in the sense that the strategy on models which satisfy $\phi$ is a \emph{substrategy} of the strategy on models satisfying $\overline{\phi}$ --- if $(L,v) \models \phi$ and $(L',v') \models \overline{\phi}$ then $\sigma_{\llbracket \top \oslash (\phi \lhd \top) \rrbracket(L,v)} \subseteq \sigma_{\llbracket \top \oslash (\phi \lhd \top) \rrbracket(L',v')}$.

In contrast, consider the formula $\bot \lhd (\overline{\phi} \oplus (\top \oslash \phi))$. The game forest is given as follows, using the same notation as above:

\begin{center}
\includegraphics[scale=0.6]{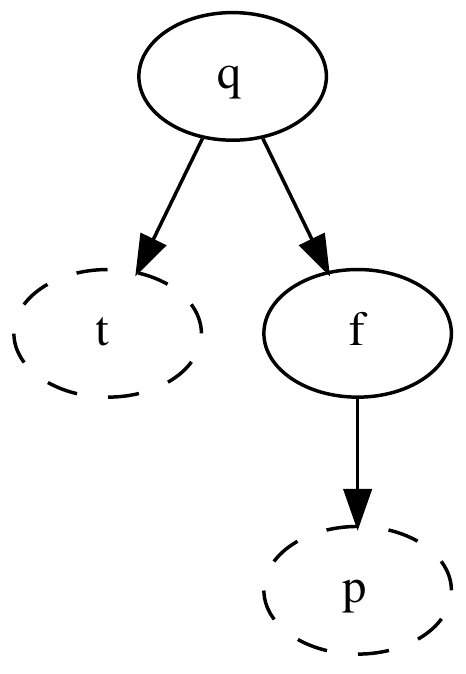}
\end{center}

\noindent There is a family of strategies on this (negative) game: if
$\phi$ is true, Player plays \texttt{f} and if $\overline{\phi}$ is
true, Player plays \texttt{t}. However, this strategy is not uniform
as the choice of second move depends on the truth value of $\phi$ in
the appropriate $\mathcal{L}$-structure. Correspondingly, the formula is not
provable in \textsf{WS1}. 

We now formalise this notion of uniformity of strategies as a
naturality property.

\subsection{Uniform Strategies}

\subsubsection{Game Embeddings}

We wish to formalise categorically the notion of a game $A$ being a subgame of $B$: we can then state that a family of strategies is uniform if whenever $A$ is a subgame of $B$, the restriction of $\sigma_B$ to $A$ is $\sigma_A$. If we consider games as trees, we require a tree embedding from $P_A$ into $P_B$. We use the following machinery: 

\begin{definition}
Let $\mathcal{C}$ be a poset-enriched category. The category $\mathcal{C}_e$ has the same objects as $\mathcal{C}$ and a map $A \rightarrow B$ in $\mathcal{C}_e$ consists of a pair $(i_f , p_f)$ where $i_f : A \rightarrow B$ and $p_f : B \rightarrow A$ in $\mathcal{C}$, such that $p_f \circ i_f = \id$ and $i_f \circ p_f \sqsubseteq \id$.
\begin{itemize}
\item The identity is given by $(\id, \id)$.
\item For composition, set $(i_f, p_f) \circ (i_g , p_g) = (i_f \circ i_g , p_g \circ p_f)$. We need to check this is a valid pairing: $p_{f \circ g} \circ i_{f \circ g} = p_g \circ p_f \circ i_f \circ i_g = p_g \circ \id \circ i_g = \id$ and $i_{f \circ g} \circ p_{f \circ g} = i_f \circ i_g \circ p_g \circ p_f \sqsubseteq i_f \circ \id \circ p_f = i_f \circ p_f \sqsubseteq \id$.
\item It is clear that composition is associative and that $f = f \circ \id = \id \circ f$.
\end{itemize}
\end{definition}

\noindent Let $\mathcal{G}$ denote the poset-enriched category of games and (not-necessarily winning) strategies, and $\mathcal{G}_s$ its subcategory of strict strategies, with $\sqsubseteq$ given by strategy inclusion. A tree embedding of $A$ into $B$ corresponds to a map $A \rightarrow B$ in $\mathcal{G}_e$.


\begin{proposition}
If $f : A \rightarrow B$ in $\mathcal{G}_e$ then $i_f$ and $p_f$ are strict.
\end{proposition}
\begin{proof}
If $i_f$ responds to an opening move in $B$ with a move in $B$ then so does $i_f \circ p_f$ and so $i_f \circ p_f \sqsubseteq \id$ fails. Similarly, if $p_f$ responds to an opening move in $A$ with a move in $A$ then so does $p_f \circ i_f$ and so $p_f \circ i_f = \id$ fails. \qed
\end{proof}

\noindent We can thus define identity-on-objects functors $i : \mathcal{G}_e \rightarrow \mathcal{G}_s$ and $p : \mathcal{G}_e \rightarrow \mathcal{G}_s^\op$.
We can show that our operations on games lift to functors on $\mathcal{G}_e$.

\begin{proposition}
All of the operations $\multimap$,$\oslash$,$\otimes$,$\&$,$!$ extend to covariant (bi)functors on $\mathcal{G}_e$.
\end{proposition}
\begin{proof}
  Each case exploits functoriality and monotonicity of the relevant
  operation. We just give an example: set $(i,p) \multimap (i',p') =
  (p \multimap i' , i \multimap p')$. Then $(i \multimap p') \circ (p
  \multimap i') = (p \circ i) \multimap (p' \circ i') = \id \multimap
  \id = \id$ and $(p \multimap i') \circ (i \multimap p') = (i \circ
  p) \multimap (i' \circ p') \sqsubseteq \id \multimap \id = \id$. \qed
\end{proof}

\subsubsection{Lax natural Transformations}

Given an embedding $e : A \rightarrow B$ and strategies $\sigma_A : A$, $\sigma_B : B$, $\sigma_B$ \emph{restricts to} $\sigma_A$ if $\sigma_A = p_e \circ \sigma_B$. We generalise this idea using the notion of \emph{lax natural transformations}.

\begin{definition}
Let $\mathcal{C}$ be a category, $\mathcal{D}$ a poset-enriched category and $F, G : \mathcal{C} \rightarrow \mathcal{D}$. A \emph{lax natural transformation} $F \Rightarrow G$ is a family of arrows $\mu_A : F(A) \rightarrow G(A)$ such that $\eta_B \circ F(f) \sqsupseteq G(f) \circ \eta_A$.
\end{definition}
\begin{diagram}
F(A) & \rTo^{\mu_A} & G(A) \\
\dTo^{F(f)} & \sqsupseteq & \dTo_{G(f)} \\
F(B) & \rTo_{\mu_B} & G(B) \\
\end{diagram}

\noindent We can compose lax natural transformations using vertical
composition. There is also a form of horizontal composition, provided
that one of the two functors is the identity: Let $H,G : \mathcal{C}
\rightarrow \mathcal{D}$ and $\mu : G \Rightarrow H$ a lax natural
transformation. Then a) if $F : \mathcal{B} \rightarrow \mathcal{C}$
then there is a lax natural transformation $\mu F : G \circ F
\Rightarrow H \circ F$ given by $(\mu F)_A = \mu_{F(A)}$ and b) if $J
: \mathcal{D} \rightarrow \mathcal{E}$ is monotonic then there is a
lax natural transformation $J \mu : J \circ G \rightarrow J \circ H$
given by $(J \mu)_A = J(\mu_A)$.

\subsubsection{Uniform Winning Strategies}
\label{uniwinstrats}

\begin{definition}
Let $F , G : \mathcal{C} \rightarrow \mathcal{G}_e$. A \emph{uniform strategy} from $F$ to $G$ is a lax natural transformation $\sigma : i \circ F \Rightarrow i \circ G$. A \emph{uniform total strategy} is a uniform strategy $\sigma$ where each $\sigma_A$ is total. A \emph{uniform winning strategy} is a uniform strategy where each $\sigma_A$ is winning.
\end{definition}

\noindent If $f : A \rightarrow B$, the lax naturality condition is that $i_{G(f)} \circ \sigma_A \sqsubseteq \sigma_B \circ i_{F(f)}$. Thus $\sigma_A = p_{G(f)} \circ i_{G(f)} \circ \sigma_A \sqsubseteq p_{G(f)} \circ \sigma_B \circ i_{F(f)}$. But since $\sigma_A$ is total, it is maximal in the ordering $\sqsubseteq$ and we must have $\sigma_A = p_{G(f)} \circ \sigma_B \circ i_{F(f)}$. Similarly, we see that $\sigma_A = p_{G(f)} \circ \sigma_B \circ i_{F(f)}$ implies the lax naturality condition as $i_{G(f)} \circ \sigma_A = i_{G(f)} \circ p_{G(f)} \circ \sigma_B \circ i_{F(f)} \sqsubseteq \sigma_B \circ i_{F(f)}$. Thus, lax naturality captures the fact that $\sigma_A$ is determined by $\sigma_B$ via restriction. If $F$ is the constant functor $\kappa_I$, this reduces to $\sigma_A = p_{G(e)} \circ \sigma_B$.

We can construct a WS-category of uniform strategies over a
base category $\mathcal{C}$. Let $\mathcal{G}^\mathcal{C}$ be the
category where:
\begin{itemize}
\item Objects are functors $\mathcal{C} \rightarrow \mathcal{G}_e$
\item An arrow $F \rightarrow G$ is a uniform strategy $F
  \Rightarrow G$
\item Composition is given by vertical composition of lax natural
  transformations
\item The identity on a functor $F$ is given by the lax natural
  transformation $\eta : F \Rightarrow F$ where $\eta_A =
  \id_{F(A)}$. It is clear that this is lax natural.

\end{itemize}

\noindent Similarly, we can construct a category $\mathcal{W}^\mathcal{C}$ of functors and uniform winning strategies.

\begin{proposition}
$\mathcal{G}^\mathcal{C}$ is a WS!-category.
\label{uniws}
\end{proposition}
\begin{proof}
  We first exhibit the symmetric monoidal structure. $F \otimes G$ is
  defined to be $\otimes \circ (F \times G) \circ \Delta$ where
  $\Delta : \mathcal{C} \rightarrow \mathcal{C} \times \mathcal{C}$ is
  the diagonal. So, $(F \otimes G)(A) = F(A) \otimes G(A)$. On arrows, 
  we set $(\eta \otimes \rho)_A = \eta_A \otimes \rho_A$. 
  We need to show that if $f : L
  \rightarrow K$ then $(i_{A(f)} \otimes i_{C(f)}) \circ (\eta_K
  \otimes \rho_K) \sqsupseteq (i_{B(f)} \otimes i_{D(f)}) \circ
  (\eta_L \otimes \rho_L)$. That is, we need to show that $(i_{A(f)}
  \circ \eta_K) \otimes (i_{C(f)} \circ \rho_K) \sqsupseteq (i_{B(f)}
  \circ \eta_L) \otimes (i_{D(f)} \circ \rho_L)$. But this is clear by
  lax naturality of $\eta$ and $\rho$ and monotonicity of $\otimes$.

  The tensor unit $I$ is the constant functor, sending all objects to
  the game $I$ and arrows to $\id_I$.

  The morphisms $\assoc$, $\ri$, $\li$ and $\sym$ are defined
  pointwise: for example, $(\assoc_{F,G,H})_X =
  \assoc_{F(X),G(X),H(X)}$.  To check for lax naturality, we must use
  horizontal composition. For example, consider the map $\assoc : (F
  \otimes G) \otimes H \rightarrow F \otimes (G \otimes H)$ defined
  pointwise as described. The domain is $(F \otimes G) \otimes H =
  ((\_ \otimes \_) \otimes \_) \circ (i \circ F \times i \circ G
  \times i \circ H) \circ \Delta_3$ where $\Delta_3$ is the diagonal
  functor $\mathcal{C} \rightarrow \mathcal{C} \times \mathcal{C}
  \times \mathcal{C}$. Similarly, the codomain is $(\_ \otimes (\_
  \otimes \_)) \circ (i \circ F \times i \circ G \times i \circ H)
  \circ \Delta_3$. We can thus see that $\assoc$ is equal to the
  horizontal composition $\assoc J$ where $J = (i \circ F \times i
  \circ G \times i \circ H) \circ \Delta_3$ and $\assoc$ is the
  natural transformation $\_ \otimes (\_ \otimes \_) \Rightarrow (\_
  \otimes \_) \otimes \_$ in $\mathcal{G}_s$.
\begin{diagram}
\mathcal{C} & \rTo^{\Delta_3} & \mathcal{C} \times \mathcal{C} \times
\mathcal{C} & \rTo^{i \circ F \times i \circ G \times i \circ H} &
\mathcal{G}_s \times \mathcal{G}_s \times \mathcal{G}_s & \rTo^{\_
  \otimes (\_ \otimes \_)} & \mathcal{G}_s \\
& \dImplies^{\id} & & \dImplies^{\id} & & \dImplies^{\assoc} \\
\mathcal{C} & \rTo^{\Delta_3} & \mathcal{C} \times \mathcal{C} \times
\mathcal{C} & \rTo^{i \circ F \times i \circ G \times i \circ H} &
\mathcal{G}_s \times \mathcal{G}_s \times \mathcal{G}_s & \rTo^{(\_
  \otimes \_) \otimes \_} & \mathcal{G}_s \\
\end{diagram}

  One can similarly express the other monoidal isomorphisms in this
  way to see lax naturality. The coherence equations of symmetric
  monoidal categories inherit pointwise from $\mathcal{G}$.

  Symmetric monoidal closure, products, sequoidal closure and linear
  functional extensionality lift pointwise from $\mathcal{G}$ using
  horizontal composition. We can also show that the coalgebraic monoidal
  exponential structure lifts from $\mathcal{G}$. \qed
\end{proof}

\begin{proposition}
$\mathcal{W}^\mathcal{C}$ is a WS!-category.
\end{proposition}
\begin{proof}
  We proceed precisely as in Proposition \ref{uniws}, lifting the
  structure of a WS!-category in $\mathcal{W}$ to that in
  $\mathcal{W}^\mathcal{C}$. In particular, pointwise-winningness of
  the relevant morphisms in $\mathcal{W}^\mathcal{C}$ inherits from
  the winningness in $\mathcal{W}$. \qed
\end{proof}


\subsection{Quantifiers}

\subsubsection{Category of $\mathcal{L}$-structures}

\begin{definition}
  Given a set of variables $X$ and set of atomic formulas $\Theta$, we
  let $\mathcal{M}_\Theta^X$ denote the category of
  $\Theta$-satisfying $\mathcal{L}$-models over $X$. Objects are
  $\mathcal{L}$-models over $X$ that satisfy each formula in
  $\Theta$. A morphism $(L,v) \rightarrow (L',v')$ is a map $f : |L|
  \rightarrow |L'|$ such that:
\begin{itemize}
\item For each $x \in X$, $v'(x) = f(v(x))$
\item If $(L,v) \models \overline{\phi}(\overrightarrow{a})$ for $\overrightarrow{a} \in |L|^{\ar(\phi)}$ then $(L',v') \models \overline{\phi}(\overrightarrow{f(a)})$
\item For each function symbol $g$ in $\mathcal{L}$, $f(I_L(g)(\overrightarrow{a})) = I_{L'}(g)(\overrightarrow{f(a)})$.
\end{itemize}
\end{definition}

\noindent Note that since the positive atoms include inequality, such morphisms must be injective. Also note that if $f : (L,v) \rightarrow (L',v')$ and $(L,v) \models \overline{\phi}(\overrightarrow{s})$ then $(L',v') \models \overline{\phi}(\overrightarrow{s})$.

If $v$ is a valuation on $X$, define $v[x \mapsto l]$ on $X \cup \{ x \}$ to be the valuation sending $y$ to $v(y)$ if $y \neq x$, and $x$ to $l$. Given $f : (L,v) \rightarrow (M,w)$ in $\mathcal{M}_X^\Theta$ and $s$ a term with $FV(s) \subseteq X$, $f$ is also a map $(L,v[x \mapsto v(s)]) \rightarrow (M,w[x \mapsto w(s)])$ in $\mathcal{M}_{X \cup \{ x \}}^\Theta$. We know that $f$ preserves all of the valuations other than $x$, and for $x$ we see that $f(v[x \mapsto v(s)](x)) = f(v(s)) = w(s) = w[x \mapsto w(s)](x)$. \\

\noindent We will give semantics of sequents $X ; \Theta \vdash
\Gamma$ as functors $\mathcal{M}_\Theta^X \rightarrow \mathcal{G}_e$,
and proofs as uniform winning strategies. 

\subsubsection{Quantifiers as Adjoints}

In this section, we will describe an adjunction that will allow us to
interpret the quantifiers. 

\begin{itemize}
\item If $FV(s) \subseteq X$ we can define a functor $\mathsf{set}^x_s : \mathcal{M}_X^\Theta \rightarrow \mathcal{M}_{X \uplus \{x\}}^\Theta$ by $\mathsf{set}^x_s(L,v) = (L,v[x \mapsto v(s)])$ and if $f : (L,v) \rightarrow (M,w)$ we set $\mathsf{set}^x_s(f) = f$. We need to check that $\mathsf{set}^x_s(f)$ is a valid morphism. We know that $\mathsf{set}^x_s(f)$ preserves all variables in $X$, and $\mathsf{set}^x_s(f)(v[x \mapsto v(s)](x)) = f(v(s)) = w(s) = w[x \mapsto w(s)](x)$ as required. It is clear that $\mathsf{set}^x_s$ is functorial.

From this we can extract a functor $\mathsf{set}'^x_s : \mathcal{W}^{\mathcal{M}_{X \uplus \{ x \}}^\Theta} \rightarrow \mathcal{W}^{\mathcal{M}_X^\Theta}$, mapping $F$ to $F \circ \mathsf{set}^x_s$, with an action on arrows defined by horizontal composition.

\item Provided $x$ does not occur in $\Theta$, there is an evident forgetful functor $U_x : \mathcal{M}_{X \uplus \{x\}}^\Theta \rightarrow \mathcal{M}_X^\Theta$ mapping $(L,v)$ to $(L,v-x)$. From this we can extract a functor $U'_x : \mathcal{W}^{\mathcal{M}_X^\Theta} \rightarrow \mathcal{W}^{\mathcal{M}_{X \uplus \{ x \}}^\Theta}$ mapping $F$ to $F \circ U_x$, with an action on arrows defined by horizontal composition. Note that $U_x \circ \mathsf{set}^x_s = \id$ and so $\mathsf{set}'^x_s \circ U'_x = \id$.
\end{itemize}

We will show that $U'_x$ has a right adjoint $\forall x . \_ $. Assuming empty $\Gamma$, this allows us to interpret the rules $\p_\forall$ and $\p_\exists$.

\begin{definition}
Let $\mathcal{C}$ be a category. We define the category $\FamInj(\mathcal{C})$. An object is a set $I$ and a family of $\mathcal{C}$-objects $\{ A_i : i \in I \}$. An arrow $\{ A_i : i \in I \} \rightarrow \{ B_j : j \in J \}$ is a pair $(f, \{ f_i : i \in I \})$ where $f$ is an injective function $I \rightarrow J$ and each $f_i : A_i \rightarrow B_{f(i)}$. We will often write such a map as $(f , \{ f_i \})$ when we wish to leave the indexing set implicit.
\begin{itemize}
\item Composition is defined by $(f,\{f_i\}) \circ (g,\{g_i\}) = (f \circ g , \{ f_{g(i)} \circ g_i \})$.
\item The identity $\{ A_i : i \in I \} \rightarrow \{A_i : i \in I \}$ is given by $(\id, \{ \id_{A_i} \})$.
\item Satisfaction of the categorical axioms is inherited from $\mathcal{C}$.
\end{itemize}
\end{definition}

\begin{definition}
Let $F : \mathcal{C} \rightarrow \mathcal{D}$. We define $\FamInj(F) : \FamInj(\mathcal{C}) \rightarrow \FamInj(\mathcal{D})$. On objects, $\FamInj(F)(\{ A_i : i \in I \}) = \{ F(A_i) : i \in I \}$. On arrows, we set $\FamInj(F)(f,\{f_i\}) = (f , \{ F(f_i) \})$.
\end{definition}

%

\noindent We define a distributivity functor $\mathsf{dst} : \FamInj(\mathcal{C}) \times \mathcal{D} \rightarrow \FamInj(\mathcal{C} \times \mathcal{D})$ by $\dst(\{A_i : i \in I \}, B) = \{ (A_i,B) : i \in I \}$ and $\dst((f,\{f_i\}),g) = (f,\{(f_i,g)\})$.

Suppose $F$ is an object in $\mathcal{W}^{\mathcal{M}_{X \uplus \{ x
    \}}^\Theta}$ (a functor $\mathcal{M}_{X \uplus \{ x \}}^\Theta
\rightarrow \mathcal{G}_e$). We define $\forall x . F$ as an object in
$\mathcal{W}^{\mathcal{M}_X^\Theta}$ (a functor $\mathcal{M}_X^\Theta
\rightarrow \mathcal{G}_e$). We first define a product functor
$\mathsf{prod} : \FamInj(\mathcal{G}_e) \rightarrow \mathcal{G}_e$. On
objects, $\fprod$ sends $\{ G_i : i \in I \}$ to $\prod_{i \in I}
G_i$. On arrows, let $f : \{ G_j : j \in J \} \rightarrow \{ H_h : h
\in H \}$. The embedding part of $\mathsf{prod}(f)$ is given by
$\langle g_h \rangle_h$ where $g_h = i_{f_j} \circ \pi_j$ if $h =
f(j)$ and $\epsilon$ otherwise. The projection part is given by
$\langle p_{f_j} \circ \pi_{f(j)} \rangle_j$. We can check that
$\fprod$ defines a functor into $\mathcal{G}_e$. Finally, given $F :
\mathcal{M}_{X \uplus \{ x \}}^\Theta \rightarrow \mathcal{G}_e$ we
define $\forall x . F : \mathcal{M}_X^\Theta \rightarrow
\mathcal{G}_e$ to be $\mathsf{prod} \circ \FamInj(F) \circ
\mathsf{add}_x$.

\begin{proposition}
The functor $U_x' : \mathcal{W}^{\mathcal{M}_X^\Theta} \rightarrow \mathcal{W}^{\mathcal{M}_{X \uplus \{ x \}}^\Theta}$ has a right adjoint given by $\forall x . \_ = \fprod \circ \FamInj(\_) \circ \mathsf{add}_x : \mathcal{W}^{\mathcal{M}_{X \uplus \{ x \}}^\Theta} \rightarrow \mathcal{W}^{\mathcal{M}_X^\Theta}$
\end{proposition}

\begin{proof}
  We must first give the unit of this adjunction. For each $F$, we
  must give a uniform winning strategy $\eta : U'_x(\forall x . F)
  \Rightarrow F$. Such an $\eta$ is a winning uniform strategy
  $\mathsf{prod} \circ \FamInj(F) \circ \mathsf{add}_x \circ U_x
  \Rightarrow F$. Note that $(\mathsf{prod} \circ \FamInj(F) \circ
  \mathsf{add}_x \circ U_x)(L,v) = \mathsf{prod}(\{ F(L,v-x[x \mapsto
  l]) : l \in L \} ) = \prod_{l \in L} F(L,v[x \mapsto l])$. Thus
  $\eta_{(L,v)}$ must be a winning strategy $\prod_{l \in L} F(L,v[x
  \mapsto l]) \rightarrow F(L,v)$ and we take $\eta_{(L,v)} =
  \pi_{v(x)}$. One can check that this transformation is lax natural.

Given $f : U'_x(F) \rightarrow G$ we must show that there is a unique $\hat{f} : F \rightarrow \forall x . G$ such that $f = \eta_G \circ U_x'(\hat{f})$. Let $f$ be such a uniform winning strategy. Then we must give winning strategies $\hat{f}_{(L,v)} : F(L,v) \rightarrow \prod_{l \in L} G(L,v[x \mapsto l])$. Set $\hat{f}_{(L,v)} = \langle h_l \rangle_l$ where $h_l : F(L,v) \rightarrow G(L,v[x \mapsto l])$ is defined by $f_{(L,v[x \mapsto l])}$. We can check that $\hat{f}$ satisfies lax naturality.


We next need to show that $\hat{f}$ satisfies the universal property. Firstly, we must show that $f = \eta_G \circ U_x'(\hat{f})$. It suffices to show that for each $(L,v)$, $f_{(L,v)} = ((\eta_G) \circ U_x'(\hat{f}))_{L,v}$. Composition in is given by vertical composition. Thus, the RHS is given by $\pi_{v(x)} \circ \langle f_{(L,v[x \mapsto l])} \rangle_l = f_{(L,v[x \mapsto v(x)])} = f_{(L,v)}$ as required.

We need to show that $\hat{f} : F \rightarrow \forall x . G$ is the
unique uniform strategy satisfying $f = \eta_G \circ
U_x'(\hat{f})$. Suppose $h : F \rightarrow \forall x . G$ in
$\mathcal{W}^{\mathcal{M}_X^\Theta}$ satisfies this property. Then
given $(L,v)$ in $\mathcal{M}_{X \uplus \{ x \}}^\Theta$, we know that
$f_{(L,v)} = \eta_{G(L,v)} \circ h_{(L,v - x)} = \pi_{v(x)} \circ
h_{(L,v - x)}$. Let $(L,v) \in \mathcal{M}_X^\Theta$. We must show that $h_{(L,v)} = \hat{f}_{(L,v)} = \langle f_{(L,v[x \mapsto l])} \rangle_l$. Thus we need to show that for each $l$, $\pi_l \circ h_{(L,v)} = f_{(L,v[x \mapsto l])}$. But consider the model $(L,v[x \mapsto l])$. This is $f_{(L,v[x \mapsto l])} = \pi_{v[x \mapsto l](x)} \circ h_{(L, v[x \mapsto l] - x)} = \pi_l \circ h_{(L,v)}$, as required. \qed
\end{proof}

If $N : \mathcal{M}^\Theta_{X \uplus \{ x \}} \rightarrow
\mathcal{G}_e$ then on objects $\llbracket \forall x . N
\rrbracket(L,v) = \prod_{l \in |L|} \llbracket N \rrbracket(L,v[x
\mapsto l])$. For the action of $\forall x . N$ on arrows, suppose $f
: (L,v) \rightarrow (L',w)$. Then $\llbracket \forall x . N
\rrbracket(f) : \prod_{l \in |L|} \llbracket N \rrbracket(L,v[x
\mapsto l]) \rightarrow \prod_{l \in |L'|} \llbracket N
\rrbracket(L',w[x \mapsto l])$ is given as follows: The embedding part
(left to right) is given by $\langle g_{m} \rangle_{m}$ where $g_{m} =
\epsilon$ if $m$ is not in the image of $f$, and $g_{m} = i \llbracket
N \rrbracket(f) \circ \pi_l$ if $m = f(l)$ (note in this case $l$ is
unique by injectivity of $f$). The projection part is given by
$\langle p \llbracket N \rrbracket (f) \circ \pi_{f(l)} \rangle_l$.

Consider the map $\set'^x_s(\eta) : \forall x . F =
\set'^x_s(U'_x(\forall x . F)) \rightarrow \set'^x_s(F)$ in the
category $\mathcal{M}_X^\Theta$. Pointwise, $\set'^x_s(\eta)_{(L,v)} :
\prod_{l \in L} F(L,v[x \mapsto l]) \rightarrow F(L,v[x \mapsto
v(s)])$ is given by $\pi_{v(s)}$, and so we will write $\pi_s$ for
this map.


\subsection{Semantics of Sequents}

We  define the semantics of sequents $X ; \Theta \vdash \Gamma$ as
functors $\mathcal{M}_X^\Theta \rightarrow \mathcal{G}_e$ inductively, via the equations given in the previous section, extended with the following interpretations of atoms and quantifiers:  

\begin{small}
\[
\begin{array}{lcllcl}
  \llbracket \Phi \vdash \phi(\overrightarrow{s}) \rrbracket(L,v) & = & I $ if $ (L,v) \models \overline{\phi}(\overrightarrow{s}) & \llbracket \Phi \vdash \overline{\phi}(\overrightarrow{s}) \rrbracket(L,v) & = & I $ if $ (L,v) \models \overline{\phi}(\overrightarrow{s}) \\
  \llbracket \Phi \vdash \phi(\overrightarrow{s}) \rrbracket(L,v) & = & o $ if $ (L,v) \models \phi(\overrightarrow{s}) & \llbracket \Phi \vdash \overline{\phi}(\overrightarrow{s}) \rrbracket(L,v) & = & o $ if $ (L,v) \models \phi(\overrightarrow{s}) \\
\end{array}
\]
\[
\begin{array}{lcl}
\llbracket X ; \Theta \vdash \forall x . N \rrbracket & = & \forall x . \llbracket X \uplus \{ x \} ; \Theta \vdash N \rrbracket \\
\llbracket X ; \Theta \vdash \exists x . P \rrbracket & = & \forall x . \llbracket X \uplus \{ x \} ; \Theta \vdash P \rrbracket \\
\end{array}
\]
\end{small}


In the case of atoms, the functors are specified pointwise on objects,
and we must also define the (functorial) action on arrows. Let $f :
(L,v) \rightarrow (L',v')$. If the truth value of
$\phi(\overrightarrow{s})$ is the same in $(L,v)$ and $(L',v)$, we use
the identity embedding $(\id,\id)$. If the truth value of
$\phi(\overrightarrow{s})$ is different, we must have $(L,v) \models
\phi(\overrightarrow{s})$ and $(L',v) \models
\overline{\phi}(\overrightarrow{s})$ since morphisms in
$\mathcal{M}_X$ preserve truth of positive atoms. Thus we need an
embedding $I \rightarrow o$. We can take $(\epsilon_{I \multimap o},
\epsilon_{o \multimap I})$ where $\epsilon_A$ is the strategy
containing just the empty sequence. Note that $\epsilon_{o \multimap
  I} \circ \epsilon_{I \multimap o} = \epsilon_I = \id_I$ and
$\epsilon_{I \multimap o} \circ \epsilon_{o \multimap I} = \epsilon
\sqsubseteq \id_o$ ($\epsilon$ is the bottom element with respect to
$\sqsubseteq$).

We must check functoriality. We have already noted that if the truth
value of $\phi(\overrightarrow{s})$ is the same in $(L,v)$ and
$(L',v')$ then $\llbracket \phi(\overrightarrow{s}) \rrbracket(f) =
\id$, so in particular $\llbracket \phi(\overrightarrow{s})
\rrbracket(\id) = \id$. For composition, suppose $f : (L,v) \rightarrow (L',v')$ and $g : (L',v') \rightarrow (L'',v'')$. We can consider the truth value of $\phi(\overrightarrow{s})$ in each of these models (only some cases are possible, as morphisms preserve truth of positive atoms). \\

\noindent \begin{tabular}{|c|c|c|l|}
\hline
$(L,v) \models$ & $(L',v') \models$ & $(L'',v'') \models$ & $\llbracket \phi(\overrightarrow{s}) \rrbracket(g) \circ \llbracket \phi(\overrightarrow{s}) \rrbracket(g) = \llbracket \phi(\overrightarrow{s}) \rrbracket(g \circ f)$ \\
\hline
$\phi(\overrightarrow{s})$ & $\phi(\overrightarrow{s})$ & $\phi(\overrightarrow{s})$ & $(\id , \id) \circ (\id , \id) = (\id , \id)$ \\
$\phi(\overrightarrow{s})$ & $\phi(\overrightarrow{s})$ & $\overline{\phi}(\overrightarrow{s})$ & $(\epsilon , \epsilon) \circ (\id , \id) = (\epsilon, \epsilon)$\\
$\phi(\overrightarrow{s})$ & $\overline{\phi}(\overrightarrow{s})$ & $\overline{\phi}(\overrightarrow{s})$ & $(\id , \id) \circ (\epsilon , \epsilon) = (\epsilon , \epsilon)$ \\
$\overline{\phi}(\overrightarrow{s})$ & $\overline{\phi}(\overrightarrow{s})$ & $\overline{\phi}(\overrightarrow{s})$ & $(\id , \id) \circ (\id , \id) = (\id , \id)$ \\
\hline
\end{tabular}








\subsection{Semantics of Proofs}
We now extend the semantics of proof rules given in the previous section with interpretations for the rules for quantifiers, atoms and equality, completing the semantics of \textsf{WS1}. 

\noindent We first show that if $x \not \in FV(\Gamma)$ there is an
isomorphism $\dist_\Gamma : \llbracket \forall x . A , \Gamma
\rrbracket \cong \forall x . \llbracket A , \Gamma \rrbracket$ in
$\mathcal{W}^{\mathcal{M}_X^\Theta}$. Observe that there is a natural isomorphism $$\dist_\oslash :
\_ \oslash \_ \circ (\mathsf{prod} \times \id) \Rightarrow
\mathsf{prod} \circ \FamInj(\_ \oslash \_) \circ \mathsf{dst} :
\FamInj(\mathcal{G}_s) \times \mathcal{G}_s \rightarrow
\mathcal{G}_s$$ which is concretely  a family of winning strategies
$$\mathsf{prod}(\{ G_i : i \in I \}) \oslash M \rightarrow \mathsf{prod}(\{ G_i \oslash M : i \in I \})$$ given by
$\dist_\oslash = \langle \pi_i \oslash \id \rangle_i$. Each
$\dist_\oslash$ is a natural isomorphism in $\mathcal{W}_s$.

Similarly, we can define a natural isomorphism $$\dist_\multimap :
\fprod(\{ M \multimap G_i : i \in I \}) \cong M \multimap \fprod(\{
G_i : i \in I \})$$ between functors $$\_ \multimap \_ \circ
(\mathsf{prod} \times \id) \Rightarrow \mathsf{prod} \circ \FamInj(\_
\multimap \_) \circ \mathsf{dst} : \FamInj(\mathcal{G}_s) \times
\mathcal{G}_s \rightarrow \mathcal{G}_s.$$

For each $\Gamma$, we can then construct a
map $$\mathsf{dist}_{b,\Gamma} : \llbracket \Gamma \rrbracket^b_1
\circ (\mathsf{prod} \times \id) \cong \fprod \circ \FamInj(\llbracket
\Gamma \rrbracket^b_1) \circ \mathsf{dst} : \FamInj(\mathcal{G}_s)
\times \mathcal{M}_X^\Theta \rightarrow \mathcal{G}_s$$ proceeding by
induction on $\Gamma$.

Finally, given a sequent $A,\Gamma$ we define $\dist_{\Gamma}$ as the
following horizontal composition, where $b$ is the polarity of $A$. It
is easy to see by checking pointwise that the functor $\forall x
. \llbracket A , \Gamma \rrbracket$ is equal to the given
decomposition.
\begin{scriptsize}
\begin{diagram}
\forall x . \llbracket A , \Gamma \rrbracket : & \mathcal{M}_X^\Theta & \rTo^{\langle \mathsf{add}_x , \id \rangle} & \FamInj(\mathcal{M}_{X \uplus \{ x \}}^\Theta) \times \mathcal{M}_X^\Theta & \rTo^{\FamInj(A) \times \id} & \FamInj(\mathcal{G}_s) \times \mathcal{M}_X^\Theta & \rTo^{\fprod \circ \FamInj(\llbracket \Gamma \rrbracket^b_1) \circ \mathsf{dst}} & \mathcal{G}_s \\
& & \dImplies^\id & & \dImplies^{\id} & & \dImplies^{\mathsf{dist}_{b,\Gamma}^{-1}} \\
\llbracket \forall x . A , \Gamma \rrbracket : & \mathcal{M}_X^\Theta & \rTo^{\langle \mathsf{add}_x , \id \rangle} & \FamInj(\mathcal{M}_{X \uplus \{ x \}}^\Theta) \times \mathcal{M}_X^\Theta & \rTo^{\FamInj(A) \times \id} & \FamInj(\mathcal{G}_s) \times \mathcal{M}_X^\Theta & \rTo^{\llbracket \Gamma \rrbracket^b_1 \circ (\mathsf{prod} \times \id)} &  \mathcal{G}_s \\
\end{diagram}
\end{scriptsize}

\noindent Since $\dist_\Gamma$ is a natural isomorphism, and pointwise
winning, it is an isomorphism in $\mathcal{W}^{\mathcal{M}_X^\Theta}$.

\begin{proposition}
  $\pi_{v(x)} \circ {\dist_{\Gamma}}_{(L,v)} = \llbracket \Gamma
  \rrbracket^b(\pi_{v(x)})$
\label{contprodpoint}
\end{proposition}
\begin{proof}
  We can check this by induction on $\Gamma$, as in Proposition
  \ref{contprod}. \qed
\end{proof}

We next give semantics to the rules involving atoms and
quantifiers. We first introduce some notation. Suppose $\mathcal{C}$
is the coproduct of two categories $\mathcal{D}$ and $\mathcal{E}$
(the disjoint union of the two categories, where there are no maps
between them). If $F : \mathcal{C} \rightarrow \mathcal{G}_e$ we write
$F|_\mathcal{D}$ and $F|_\mathcal{E}$ for the restriction of $F$ to
$\mathcal{D}$ and $\mathcal{E}$ respectively. If $\eta : F \Rightarrow
G$ then we can restrict $\eta$ to a natural transformation
$F|_\mathcal{D} \Rightarrow G|_\mathcal{D}$, and we write
$\eta|_\mathcal{D}$ for this restriction. If $\eta : F|_\mathcal{D}
\Rightarrow G|_\mathcal{D}$ and $\sigma : F|_\mathcal{E} \Rightarrow
G|_{\mathcal{E}}$ then we write $[ \eta , \sigma
]_{\mathcal{D},\mathcal{E}}$ for the lax natural transformation
defined by $[ \eta , \sigma ]_A = \eta_A$ if $A \in \mathcal{D}$ and
$[\eta , \sigma]_A = \sigma_A$ if $A \in \mathcal{E}$. Lax naturality
of $[\eta , \sigma]$ inherits from lax naturality of $\eta$ and
$\sigma$, since there are no maps between $\mathcal{D}$ and
$\mathcal{E}$ when viewed as subcategories of $\mathcal{C}$.  If
$\mathcal{C} = \mathcal{M}_X^\Theta$ then we will write $[\eta ,
\sigma]_{\alpha , \beta}$ for $[\eta , \sigma]_{\mathcal{M}_X^{\Theta
    , \alpha} , \mathcal{M}_X^{\Theta , \beta}}$.

We construct an isomorphism $$H_{x,y,z} : \mathcal{M}_X^{\Theta,x=y}
\cong \mathcal{M}_{X / \{x,y\} \uplus \{ z \}}^{\Theta[ \frac{z}{x} ,
  \frac{z}{y} ]} : H^{-1}_{x,y,z}$$ with $H_{x,y,z}(M,v) = (M,v [z
\mapsto v(x)] - x - y)$ and $H^{-1}(M,v) = (M,v[x \mapsto v(z),y
\mapsto v(z)] - z)$. We can show that $\llbracket (X ; \Theta \vdash
\Gamma) [ \frac{z}{x} , \frac{z}{y} ] \rrbracket = \llbracket X ;
\Theta , x = y \vdash \Gamma \rrbracket H_{x,y,z}^{-1}$ by induction
on $\Gamma$.

Semantics of the rules involving atoms and quantifiers are given in
Figure \ref{WS1-sem}. We must justify lax naturality of
$\p_\mathsf{at-}$: the following diagram must lax commute:
\begin{diagram}
I & \rTo^{\llbracket \p_\mathsf{at-}(p) \rrbracket(M,w)} & \llbracket \phi(\overrightarrow{s}) , \Gamma \rrbracket(M,w) \\
\dTo^\id & \sqsupseteq & \dTo_{i\llbracket \phi(\overrightarrow{s}) , \Gamma \rrbracket(f)} \\
I & \rTo_{\llbracket \p_\mathsf{at-}(p) \rrbracket(L,v)} & \llbracket \phi(\overrightarrow{s}) , \Gamma \rrbracket(L,v) \\
\end{diagram}

\noindent To see this, note that if $(L,v)$ and $(M,w)$ agree on $\phi(\overrightarrow{x})$ then the diagram lax commutes by lax naturality of $\epsilon$ or $\llbracket p \rrbracket$. If they disagree, then we must have $(L,v) \models \overline{\phi}(\overrightarrow{x})$ and $(M,w) \models \phi(\overrightarrow{x})$.
We need to show that $\llbracket \p_\mathsf{at-}(p) \rrbracket(L,v) \sqsupseteq i \llbracket \phi(\overrightarrow{x}) , \Gamma \rrbracket(f) \circ \llbracket \p_\mathsf{at-}(p) \rrbracket(M,w)$.
But $\llbracket \p_\mathsf{at-}(p) \rrbracket(M,w) = p\llbracket \phi(\overrightarrow{x}) , \Gamma \rrbracket(f) \circ \llbracket \p_\mathsf{at-}(p) \rrbracket(L,v)$ as both sides map into the terminal object, so $\llbracket \p_\mathsf{at-}(p) \rrbracket(L,v) \sqsupseteq i\llbracket \phi(\overrightarrow{x}) , \Gamma \rrbracket(f) \circ p\llbracket \phi(\overrightarrow{x}) , \Gamma \rrbracket(f) \circ \llbracket \p_\mathsf{at-}(p) \rrbracket(L,v) = i \llbracket \phi(\overrightarrow{x}) , \Gamma \rrbracket(f) \circ \llbracket \p_\mathsf{at-}(p) \rrbracket(M,w)$.

\begin{figure*}[ht]
\caption{Semantics of Rules involving Atoms and Quantifiers}
\begin{small}
\centering
\label{WS1-sem}

\vspace{1ex}
\hrule
\hrule
\begin{tabular}{cc}
\\[0.2ex]

\AxiomC{$\sigma : \llbracket \Theta , \overline{\phi}(\overrightarrow{s}) \vdash \bot, \Gamma \rrbracket$}
\LeftLabel{$\p_\mathsf{at-}$}
\UnaryInfC{$[\sigma , \epsilon]_{\overline{\phi}(\overrightarrow{s}) , \phi(\overrightarrow{s})} : \llbracket \Theta \vdash \phi(\overrightarrow{s}) , \Gamma \rrbracket$}
\DisplayProof

\\[4ex]

\AxiomC{$\sigma : \llbracket \Theta , \overline{\phi}(\overrightarrow{s}) \vdash \top , \Gamma \rrbracket$}
\LeftLabel{$\p_\mathsf{at+}$}
\UnaryInfC{$\sigma : \llbracket \Theta , \overline{\phi}(\overrightarrow{s}) \vdash \overline{\phi}(\overrightarrow{x}) , \Gamma \rrbracket$}
\DisplayProof

\\[4ex]

\AxiomC{$\sigma : \llbracket (X ; \Theta \vdash \Gamma)[ \frac{z}{x} , \frac{z}{y} ] \rrbracket$}
\AxiomC{$\tau : \llbracket X ; \Theta , x \neq y \vdash \Gamma \rrbracket$}
\LeftLabel{$\p_\match^{x,y,z}$}
\BinaryInfC{$[\sigma H_{x,y,z}, \tau]_{x = y , x \neq y} : \llbracket X ; \Theta \vdash \Gamma \rrbracket$}
\DisplayProof

\\[4ex]

\AxiomC{}
\LeftLabel{$\p_{\neq}$}
\UnaryInfC{$\emptyset : \llbracket \Theta , x \neq x \vdash \Gamma \rrbracket$}
\DisplayProof

\\[4ex]

\AxiomC{$\sigma : \llbracket X \uplus \{ x \} ; \Theta \vdash N , \Gamma \rrbracket$}
\LeftLabel{$\p_\forall$}
\RightLabel{\small{$x \not \in FV(\Theta,\Gamma)$}}
\UnaryInfC{$ \dist_{\Gamma}^{-1} \circ \hat{\sigma} : \llbracket X ; \Theta \vdash \forall x . N , \Gamma \rrbracket$}
\DisplayProof

\\[4ex]

\AxiomC{$\sigma : \llbracket X ; \Theta \vdash P[s/x] , \Gamma \rrbracket$}
\LeftLabel{$\p_\exists^s$}
\RightLabel{$FV(s) \subseteq X$}
\UnaryInfC{$\sigma \circ \pi_s \circ \dist_\Gamma : \llbracket X ; \Theta \vdash \exists x . P, \Gamma \rrbracket$}
\DisplayProof

\\[4ex]




\AxiomC{$\sigma : \llbracket X ; \Theta \vdash M , \Gamma , \forall x . N , \Delta \rrbracket$}
\LeftLabel{$\p_\forall^\mathsf{T}$}
\RightLabel{$FV(s) \subseteq X$}
\UnaryInfC{$\llbracket \Delta \rrbracket^-(\id \oslash \pi_{s}) \circ \sigma : \llbracket X ; \Theta \vdash M , \Gamma , N[s/x] , \Delta \rrbracket$}
\DisplayProof

\\[4ex]

\AxiomC{$\sigma : \llbracket X ; \Theta \vdash Q , \Gamma , \forall x . N , \Delta \rrbracket$}
\LeftLabel{$\p_\forall^\mathsf{T}$}
\RightLabel{$FV(s) \subseteq X$}
\UnaryInfC{$\sigma \circ \llbracket \Delta \rrbracket^+(\pi_{s} \multimap \id) : \llbracket X ; \Theta \vdash Q , \Gamma , N[s/x] , \Delta \rrbracket$}
\DisplayProof

\\[4ex]

\AxiomC{$\sigma : \llbracket X ; \Theta \vdash M , \Gamma , P[s/x] , \Delta \rrbracket$}
\LeftLabel{$\p_\exists^\mathsf{T}$}
\RightLabel{$FV(s) \subseteq X$}
\UnaryInfC{$\llbracket \Delta \rrbracket^-(\pi_{s} \multimap \id) \circ \sigma : \llbracket X ; \Theta \vdash M , \Gamma , \exists x . P , \Delta \rrbracket$}
\DisplayProof

\\[4ex]

\AxiomC{$\sigma : \llbracket X ; \Theta \vdash Q , \Gamma , P[s/x] , \Delta \rrbracket$}
\LeftLabel{$\p_\exists^\mathsf{T}$}
\RightLabel{$FV(s) \subseteq X$}
\UnaryInfC{$\sigma \circ \llbracket \Delta \rrbracket^+(\id \oslash \pi_{s}) : \llbracket X ; \Theta \vdash Q , \Gamma , \exists x . P , \Delta \rrbracket$}
\DisplayProof

\\[4ex]

\end{tabular}

\hrule
\hrule
\end{small}
\end{figure*}


\section{Full Completeness}

We next show a full completeness result for the \emph{function-free}
fragment of \textsf{WS1}: in this section we assume that $\mathcal{L}$
contains no function symbols. Thus, the only uses of the $\p_\exists$
rule are of the form $\p_\exists^y$ where $y$ is some variable in
scope.

We show that the core rules suffice to represent any uniform
winning strategy $\sigma$ on a type object provided $\sigma$ is
\emph{bounded} --- i.e. there is a bound on the size of plays
occurring in $\sigma$. In particular, such a strategy is the semantics
of a unique \emph{analytic} proof --- a proof using only the core
rules, with some further restrictions on the use of the matching
rule. Given a sequent $X; \Theta \vdash \Gamma$, we say $\Theta$ is
\emph{lean} if it contains $x \neq y$ for all distinct $x$ and $y$ in
$X$ and does not contain $x \neq x$. We assume an arbitrary ordering
on variables.

\begin{definition}
A proof in \textsf{WS1} is \emph{analytic} if it
uses only core rules and has the following additional restrictions:
\begin{itemize}
\item Rules other than $\p_{\neq}$ and $\p_\match^{x,y,z}$ can only
  conclude sequents with a lean $\Theta$
\item If $\p_\match^{x,y,z}$ is used to conclude $X ; \Theta \vdash
  \Gamma$ then $\Theta$ does not contain $w \neq w$ for any $w$; $(x,y)$ is
  the least pair with $x,y \in X$, $x \not \equiv y$ and $x \neq y
  \not \in \Theta$; and $z$ is the least variable in $\mathsf{Fr}(X ;
  \Theta \vdash \Gamma)$ (the least fresh variable).
\end{itemize}
\end{definition}

\begin{theorem}
  Let $X ; \Theta \vdash \Gamma$ be a sequent of \textsf{WS1} and
  $\sigma$ a bounded uniform winning strategy on $\llbracket X; \Theta
  \vdash \Gamma \rrbracket$. Then there is a unique analytic proof $p$
  of $X ; \Theta \vdash \Gamma$ with $\llbracket p \rrbracket =
  \sigma$.
  \label{fullcomp}
\end{theorem}

All strategies on the denotations of exponential-free sequents are
bounded. Consequently, in the affine fragment we can perform
reduction-free normalisation from proofs to (cut-free) core proofs, by
reification of their semantics. We thus see that all of the non-core
rules are admissible (when restricted to this fragment).


The rest of this section sketches the proof of this full completeness
result, and describes an extension to reify unbounded strategies as
\emph{infinitary} analytic proofs. We perform a semantics-guided proof
search procedure, following \cite{HO_PCF,AJM_PCF,Lau_PG,m_ag4}.

\subsection{Uniform Choice}

When constructing a proof of a given sequent out of core rules there
is a choice of which rule to use when the outermost head connective is
$\oplus$ (either $\p_\oplus^1$ or $\p_\oplus^2$) or $\exists$ (which
$s$ to use in $\p_\exists^s$). Our choice of rule will depend on the
given strategy, depending on which component Player plays in
first. However, the input to our procedure is a family of strategies,
and we need to ensure that the same component choice is made in each
strategy. We will next show that our uniformity condition ensures
this.


\begin{proposition}
If $\Theta$ is lean and $(L,v),(M,w) \in \mathcal{M}_X^\Theta$ there exists an $\mathcal{L}$-model $(L,v) \sqcup (M,w)$ with maps $f_{(L,v,M,w)} : (L,v) \rightarrow (L,v) \sqcup (M,w)$ and $g_{(L,v,M,w)} : (M,w) \rightarrow (L,v) \sqcup (M,w)$.
\label{firstkeyobs}
\end{proposition}
\begin{proof}
If $(L,v)$ is an $\mathcal{L}$-model, define $U_{(L,v)}$ to be the elements of $|L|$ not in the image of $v$. Then the carrier of $(L,v) \sqcup (M,w)$ is defined to be $X \uplus U_{(L,v)} \uplus U_{(M,w)}$. The $\mathcal{L}$-structure validates all positive atoms, and the valuation is just $\mathsf{inj}_1$. Then the map $f_{(L,v,M,w)}$ sends $v(x)$ to $\mathsf{inj}_1(x)$ and $u \in U_{(L,v)}$ to $\mathsf{inj}_2(u)$. This is an injection because $\Theta$ is lean. $g_{(L,v,M,w)}$ is defined similarly. \qed
\end{proof}


We also recall that if $f : (L,v) \rightarrow (M,w)$ then $\sigma_{(L,v)}$ is determined entirely by $f$ and $\sigma_{(M,w)}$. In particular, uniformity for positive strategies $\sigma : N \Rightarrow o$ requires that $\sigma_{(L,v)} \sqsubseteq \sigma_{(M,w)} \circ N(f)$ but since $\sigma_{(L,v)}$ is total, it is maximal in the ordering and so we must have $\sigma_{(L,v)} = \sigma_{(M,w)} \circ N(f)$.

\begin{proposition}
Let $X ; \Theta \vdash \Gamma$ be a sequent and suppose $\Theta$ is lean. Then there exists an object in $\mathcal{M}_X^\Theta$.
\label{leannonempty}
\end{proposition}
\begin{proof}
Note that $\Theta$ just contains positive atoms. We can take $(X,\id)$, with $(X,\id) \models \overline{\phi}(\overrightarrow{x})$ just if $\overline{\phi}(\overrightarrow{x}) \in \Theta$. Then each formula in $\Theta$ is satisfied: each such formula is either $\overline{\phi}(\overrightarrow{x})$, or $x \neq y$ for distinct $x,y$. \qed
\end{proof}

We now use the above lemmas to show that in any uniform winning
strategy on a sequent whose head formula is $P \oplus Q$, either all
strategies play their first move in $P$, or all strategies play their
first move in $Q$.

\begin{proposition}
Let $M_1 , M_2 : \mathcal{M}_X^\Theta \rightarrow \mathcal{G}_e$. Suppose $\Theta$ is lean, and let $\sigma : M_1 \times M_2 \Rightarrow o$ be a uniform total (resp. winning) strategy. Then $\sigma = \tau \circ \pi_1$ for some uniform total (resp. winning) strategy $\tau : M_1 \Rightarrow o$, or $\sigma = \tau \circ \pi_2$ for some uniform total (resp. winning) strategy $\tau : M_2 \Rightarrow o$. \label{choice1}
\end{proposition}
\begin{proof}
We know that each $\sigma_{(L,v)}$ is of the form $\tau_{(L,v)} \circ \pi_i$ for some $i \in \{ 1 , 2 \}$ since in the game $M_1(L,v) \times M_2(L,v) \multimap o$ we must respond to the initial Opponent-move either with a move in $M_1$ or a move in $M_2$ (the $\pi$-atomicity condition). But we need to check that $i$ is uniform across components. Suppose that $i$ is not uniform --- then we have $(L,v)$ and $(T,w)$ with $\sigma_{(L,v)} = \tau_{(L,v)} \circ \pi_1$ and $\sigma_{(T,w)} = \tau_{(T,w)} \circ \pi_2$. Now consider $(L,v) \sqcup (T,w)$ and let $k$ be such that $\sigma_{(L,v) \sqcup (T,w)} = \tau_{(L,v) \sqcup (T,w)} \circ \pi_k$. By uniformity and totality, $\sigma_{(L,v)} = \sigma_{(L,v) \sqcup (T,w)} \circ (M_1 \times M_2)(f_{(L,v,T,w)}) = \tau_{(L,v) \sqcup (T,w)} \circ \pi_k \circ (M_1 \times M_2)(f_{(L,v,T,w)}) = \tau_{(L,v) \sqcup (T,w)} \circ M_k(f_{(L,v,T,w)}) \circ \pi_k$. But since $\sigma_{(L,v)}$ is of the form $\tau_{(L,v)} \circ \pi_1$, we must have $k = 1$. But we can reason similarly using $\sigma_{(T,w)}$ and $g_{(L,v,T,w)}$ and discover that $k = 2$. This is a contradiction.

Thus there is some $i$ such that each $\sigma_{(L,v)}$ can be decomposed into $\tau_{(L,v)} \circ \pi_i$. In particular, we can take $i$ such that $\sigma_{(X,\id)} = \tau_{(X,\id)} \circ \pi_i$ where $(X,\id)$ is as defined in Proposition \ref{leannonempty}. We only need to show that $\tau$ is lax natural. We can construct a natural transformations $\iota_1 : \langle \id , \epsilon \rangle : M_1 \rightarrow M_1 \times M_2$ and $\iota_2 : \langle \epsilon , \id \rangle : M_2 \rightarrow M_1 \times M_2$. Then $\tau = \sigma \circ \iota_i$, and so is lax natural. \qed
\end{proof}


We next show that in any uniform family of winning
strategies on a sequent with head $\exists x. P$, Player chooses the
same $x$ in each strategy component. Moreover, the chosen $x$ is the value of some variable in scope.

\begin{proposition}
Let $M : \mathcal{M}_{X \uplus \{ x \}}^\Theta \rightarrow \mathcal{G}_e$. Suppose $\Theta$ is lean, and let $\sigma : \forall x . M \Rightarrow o$ be a uniform total (resp. winning) strategy. Then there exists a unique variable $y \in X$ and uniform total (resp. winning) strategy $\tau : M \mathsf{set}^{x}_y \Rightarrow o$ such that $\sigma = \tau \circ \pi_{y}$. \label{choice2}
\end{proposition}

\begin{proof}
We firstly show that given any $\mathcal{L}$-model $(L,v)$ there is some $x$ with $\sigma_{(L,v)} = \tau_{(L,v)} \circ \pi_{v(x)}$. Suppose for contradiction that $\sigma_{(L,v)} = \tau_{(L,v)} \circ \pi_u$ for some $u \in U_{(L,v)}$. Build the $\mathcal{L}$-model $L' = X \uplus \{ a , b \} \uplus U_{(L,v)}$ with valuation $\mathsf{inj}_1$ and validating all positive atoms. Let $\sigma_{(L',\mathsf{inj}_1)} = \tau_{(L',\mathsf{inj}_1)} \circ \pi_r$. Define $m_1 : (L,v) \rightarrow (L',\mathsf{inj}_1)$ sending $v(x)$ to $\mathsf{inj}_1(x)$, $u$ to $\mathsf{inj}_2(a)$ and $v \in U_{(L,v)} - \{ u \}$ to $\mathsf{inj}_3(v)$. Then $\sigma_{(L,v)} = \sigma_{(L',\mathsf{inj}_1)} \circ \pi_r \circ \forall x . M(m_1)$.
\begin{itemize}
\item If $r = \mathsf{inj}_2(b)$ then this is $\sigma_{(L',\mathsf{inj}_1)} \circ \epsilon$ which is $\epsilon$ as $\sigma_{(L',\mathsf{inj}_1)}$ must be strict (as its total and a map into $o$). This is impossible.
\item If $r = \mathsf{inj}_1(x)$ then this is $\sigma_{(L',\mathsf{inj}_1)} \circ M(m_1) \circ \pi_{v(x)}$, which is impossible by assumption.
\item Hence we must have $r = \mathsf{inj}_2(a)$.
\end{itemize}
Define $m_2 : (L,v) \rightarrow (L',\mathsf{inj}_1)$ sending $v(x)$ to $\mathsf{inj}_1(x)$, $u$ to $\mathsf{inj}_2(b)$ and $v \in U_{(L,v)} - \{ u \}$ to $\mathsf{inj}_3(v)$. We can use similar reasoning to show that $r = \mathsf{inj}_2(b)$. This is a contradiction.

Hence, given any $(L,v)$ there is some variable $x$ such that $\sigma_{(L,v)} = \tau_{(L,v)} \circ \pi_{v(x)}$. Let $y \in X$ be the unique variable such that $\sigma_{(X,\id)} = \tau_{(X,\id)} \circ \pi_y$ where $(X,\id)$ is constructed as in Proposition \ref{leannonempty}. We now show the stronger fact that $\sigma_{(L,v)} = \tau_{(L,v)} \circ \pi_{v(y)}$. Suppose that $\sigma_{(L,v)} = \tau_{(L,v)} \circ \pi_{v(x)}$ and $\sigma_{(L,v) \sqcup (X,\id)} = \tau_{(L,v) \sqcup (X,\id)} \circ \pi_{\mathsf{inj}_1(z)}$. By lax naturality, $\tau_{(L,v)} \circ \pi_{v(x)} = \sigma_{(L,v)} = \sigma_{(L,v) \sqcup (X,\id)} \circ \forall x . M(f_{(L,v,X,\id)}) = \tau_{(L,v) \sqcup (X,\id)} \circ \pi_{v(z)} \circ \forall x . M(f_{(L,v,X,\id)})$. Since $\mathsf{inj}_1(z) = f_{(L,v,X,\id)}(v(z))$, we have $\sigma_{(L,v)} = \tau_{(L,v) \sqcup (X,\id)} \circ M(f_{(L,v,X,\id)}) \circ \pi_{v(z)}$ and so we must have $x = z$. By similar reasoning using $g_{(L,v,X,\id)}$, we see that $y = z$, so $x = y$.


Hence there is a variable $y$ such that for all $(L,v)$, $\sigma_{(L,v)} = \tau_{(L,v)} \circ \pi_{v(y)}$ for some $\tau_{(L,v)} : M(L,v[x \mapsto v(y)]) \Rightarrow o$. Since $\Theta$ is lean, $y$ is the unique variable such that $\sigma_{(L,v)} = \tau_{(L,v)} \circ \pi_{v(y)}$. Note that $M(L,v[x \mapsto v(y)]) = M(\mathsf{set}^x_y(L,v))$. We can easily check that the resulting transformation $\tau : M \mathsf{set}^x_y \Rightarrow o$ is lax natural. \qed
\end{proof}




\subsection{Reification of Strategies}

We define a procedure $\reify$ which transforms a bounded uniform
winning strategy on a formula object into a proof of that formula. It
may be seen as a semantics-guided proof search procedure: given such a
strategy $\sigma$ on the interpretation of $\Gamma$, $\reify$ finds a
proof which denotes it. Reading upwards, the procedure first 
decomposes the head formula into a unit (nullary connective) using the
head introduction rules. If this unit is $\mathbf{1}$, we are done. It
cannot be $\mathbf{0}$, as there are no (total) strategies on this
game. If the unit is $\top$ or $\bot$, the procedure then consolidates
the tail of $\Gamma$ into a single formula, using the core elimination
rules. Once this is done, the head unit is removed using $\p_\bot^+$
or $\p_\top^-$, strictly decreasing the size of the sequent. These
steps are then repeated until termination. We further have to deal
with equality: whenever a free variable is introduced, we must
consider if it is equal to each of the other free variables using the
$\p_\match$ rule.

Informally, if $\Theta$ is not lean:
\begin{itemize} 
\item If $\Theta$ contains $x \neq x$ we use $\p_{\neq}$ and halt.
\item Otherwise, we consider the least two variables $x,y \in X$ that are not declared distinct by $\Theta$ and split the family into those models that identify $x$ and $y$, and those that do not. In the former case, we can substitute fresh $z$ for both $x$ and $y$. We then apply the inductive hypothesis to both halves and apply $\p_\match^{x,y,z}$ using $H^{-1}_{x,y,z}$.
\end{itemize}
If $\Theta$ is lean, then:
\begin{itemize}
\item The case $\Gamma = \mathbf{0} , \Gamma'$ is impossible: there are no total strategies on this game.
\item If $\Gamma = \mathbf{1} , \Gamma'$ then $\sigma$ must be the empty strategy, since it is the unique total strategy on this game. This is the interpretation of the proof $\p_\mathbf{1}$.
\item If $\Gamma = \top$ then $\sigma$ must similarly be the unique total strategy on this game, i.e. the interpretation of $\p_\top$.
\item If $\Gamma = \top , P , \Gamma'$ then $\sigma$ can never play in $P$ since if it did the play restricted to $\top , P$ would not be alternating. Thus $\sigma$ is a strategy on $\top , \Gamma'$. We can call $\reify$ inductively yielding a proof of $\vdash \top , \Gamma'$, and apply $\p_\top^+$ to yield a proof of $\top , P , \Gamma$.
\item If $\Gamma = \top , N , P , \Gamma'$ then $\sigma$ is a total strategy on $\top , N \lhd P , \Gamma$ up to retagging and we can proceed inductively using $\p_\top^\lhd$. If $\Gamma = \top , N , M , \Gamma'$ we can proceed similarly, using $\p_\top^\otimes$.
\item If $\Gamma = \top , N$ then $\sigma$ is a total strategy on $\downarrow N$: we can strip off the first move yielding a total strategy on $N$, apply $\reify$ inductively yielding a proof of $\vdash N$, and finally apply $\p_\top^-$ yielding a proof of $\vdash \top , N$.
\item The case $\Gamma = \bot$ is impossible: there are no total strategies on this game. Other cases where $\bot$ is the head formula proceed as with $\top$: if the tail is a single positive formula, we remove the first move and apply $\p_\bot^+$, otherwise we shorten the tail using $\p_\bot^-$, $\p_\bot^\oslash$ or $\p_\bot^\parr$.
\item If $\Gamma = A \oslash N , \Gamma'$ then $\sigma$ is also a strategy on $A , N , \Gamma$. We can call $\reify$ inductively yielding a proof of $\vdash A , N , \Gamma$ that denotes $\sigma$, and apply $\p_\oslash$. We can proceed similarly in the following case $\Gamma = A \lhd P , \Gamma'$.
\item If $\Gamma = M \& N , \Gamma'$ then we can split $\sigma$ into those plays that start with $M$ and those that start with $N$. This yields total strategies on $M , \Gamma$ and $N , \Gamma$ respectively, which we can $\reify$ inductively and apply $\p_\&$.
\item If $\Gamma = M \otimes N , \Gamma'$ then we can split $\sigma$ into those plays that start with $M$ and those that start with $N$. This yields total strategies on $M , N , \Gamma$ and $N , M , \Gamma$ respectively, which we can $\reify$ inductively and apply $\p_\otimes$.
\item If $\Gamma = P \oplus Q , \Gamma$ then $\sigma$ specifies a
  first move that must either be in $P$ or in $Q$. In the former case,
  we have a strategy on $P , \Gamma$ and can $\reify$ inductively,
  finally applying ${\p_\oplus}_1$. In the latter case, we have a
  strategy on $Q , \Gamma$ and can $\reify$ inductively and apply
  ${\p_\oplus}_2$. The case of $\Gamma = P \parr Q , \Gamma$ is
  similar. 
\item If the head formula is a positive atom $\overline{\phi}(\overrightarrow{x})$ then we must have $\overline{\phi}(\overrightarrow{x})$ in $\Theta$, as otherwise there can be no uniform winning strategies on $\llbracket \Gamma \rrbracket$ (since some games in that family have no winning strategies). Thus we can proceed inductively and apply $\p_{\mathsf{at+}}$.
\item If the head formula is a negative atom $\phi(\overrightarrow{x})$ then we can split the family $\sigma$ into those models that satisfy $\phi(\overrightarrow{x})$ and those that do not. All strategies in the latter group must be empty, as there are no moves to play. All strategies in the former group form a uniform winning strategy on $\llbracket \Theta, \phi(\overrightarrow{x}) \vdash \bot , \Gamma \rrbracket$ and we can proceed inductively using $\p_\mathsf{at-}$.
\item If $\sigma : \llbracket X ; \Theta \vdash \Gamma = \forall x . N , \Gamma' \rrbracket$ then $\dist_{\Gamma'} \circ \sigma : I \Rightarrow \forall x . \llbracket N , \Gamma' \rrbracket$. Using our adjunction, this corresponds to a map $\eta \circ U'_x(\dist_{\Gamma'} \circ \sigma) : I \Rightarrow \llbracket N , \Gamma' \rrbracket$ in $\mathcal{W}^{\mathcal{M}_{X \uplus \{ x \}}^\Theta}$. We can then $\reify$ this inductively to yield a proof of $X \uplus \{ x \} ; \Theta \vdash N , \Gamma'$ and apply $\p_\forall$.

\item If $\Gamma = \exists x . P , \Gamma'$ then $\sigma \circ \dist_{+,\Gamma'} : \forall x . \llbracket P , \Gamma' \rrbracket \Rightarrow o$. By Proposition \ref{choice2}, there is a unique $y$ and natural transformation $\tau : \llbracket P , \Gamma' \rrbracket \mathsf{set}^x_y \Rightarrow o$ such that $\sigma \circ \dist_{+,\Gamma'} = \tau \circ \pi_{y}$. Since $x$ does not occur in $\Gamma$, we have $\llbracket P , \Gamma' \rrbracket \mathsf{set}^x_y = \llbracket P[y/x] , \Gamma' \rrbracket$. This yields a lax natural transformation $\llbracket P[y/x] , \Gamma' \rrbracket \Rightarrow o$. We can then apply the inductive hypothesis use the $\p_\exists^y$ rule.
\end{itemize}

\noindent We will later show that $\reify$ is well founded by giving a
measure on sequents that decreases on each call to the inductive
hypothesis.

\subsection{Definition of Reify}

\noindent $\reify_\Gamma$ is defined inductively in Figure
\ref{WS-reify}. Following the above remarks, the following properties
hold:
\label{compaxioms}
\begin{description}
\item [1a]The unique map $i: \varnothing  \Rightarrow \mathcal{C}(I,o)$ is a bijection.  
\item [1b]The map $\mathsf{d} = [\lambda f . f \circ \pi_1, \lambda g . f \circ \pi_2] : \mathcal{C}(M,o) + \mathcal{C}(N,o) \Rightarrow \mathcal{C}(M \times N, o)$ is a bijection. (\emph{$\pi$-atomicity} \cite{A_ADFC}).
\item [2]The map $\_\multimap o : \mathcal{C}(I,M) \Rightarrow \mathcal{C}(M \multimap o, I \multimap o)$ is a bijection.
\end{description}

\begin{figure*}
\caption{Reification of Strategies as Analytic Proofs}
\begin{small}
\centering
\label{WS-reify}

\vspace{1ex}
\hrule
\hrule

\[
\begin{array}{lll}
\mbox{For non-lean $\Theta$:} \\
\reify_{X , x \neq x ; \Theta \vdash \Gamma}(\sigma) & = & \p_{\neq} \\
\reify_{X, x , y ; \Theta \vdash \Gamma}(\sigma) & = & \p_\match^{x,y,z}(\reify(\sigma|_{\mathcal{M}_X^{\Theta,x = y}} \circ H^{-1}_{x,y,z}) , \reify(\sigma|_{\mathcal{M}_X^{\Theta,x \neq y}})) \\
& & \mbox{if $(x,y) \in X \times X$ is least such that $x \not \equiv y$ and $(x \neq y) \not \in \Theta$} \\
& & \mbox{ and $z$ is the least element in $\mathsf{Fr}(X;\Theta
  \vdash \Gamma)$} \\ 
\mbox{For lean $\Theta$: }\\
\reify_{X ; \Theta \vdash \phi(\overrightarrow{x}),\Gamma}(\sigma) & = & \p_{\mathsf{at}-}(\reify(\sigma|_{\mathcal{M}_X^{\Theta,\overline{\phi}(\overrightarrow{x})}})) \\
\reify_{X ; \Theta \vdash \overline{\phi}(\overrightarrow{x}),\Gamma}(\sigma) & = & \p_{\mathsf{at}+}(\reify(\sigma)) \\
\reify_{X ; \Theta \vdash \forall x . N , \Gamma}(\sigma) & = & \p_\forall(\reify(\eta \circ U'_x(\dist_{\Gamma'} \circ \sigma))) \\
\reify_{X ; \Theta \vdash \exists x . N , \Gamma}(\sigma) & = &
\p_\exists^y(\reify(\tau)) \mbox{where $\sigma \circ \dist_{\Gamma}^{-1} =
\tau \circ \pi_{y}$} \\ 
\reify_{\mathbf{1},\Gamma}(\sigma) & = & \p_\mathbf{1} \\
\reify_{\bot, N, \Gamma}(\sigma) & = & \p_\bot^- (\reify_{\bot , \Gamma}( \llbracket \Gamma \rrbracket^-(\mathsf{abs}) \circ \sigma)) \\
\reify_{\bot, P}(\sigma) & = & \p_\bot^+ (\reify_P(\Lambda_I^{-1}(\sigma))) \\
\reify_{\bot, P , Q , \Gamma}(\sigma) & = & \p_\bot^\parr (\reify_{\bot , P \parr Q ,\Gamma}(\llbracket \Gamma \rrbracket^-((\sym \multimap \id) \circ \passoc_\multimap^{-1} \circ \sigma))) \\
\reify_{\bot , P , N , \Gamma}(\sigma) & = & \p_\bot^- (\reify_{\bot , P \oslash N,\Gamma} (\llbracket \Gamma \rrbracket^-(\lfe) \circ \sigma)) \\
\reify_{M \& N , \Gamma}(\sigma) & = & \p_\& (\reify_{M,\Gamma} (\pi_1 \circ \dist_{-,\Gamma} \circ \sigma), \reify_{N,\Gamma} (\pi_2 \circ \dist_{-,\Gamma} \circ \sigma )) \\
\reify_{M \otimes N , \Gamma}(\sigma) & = & 
\p_\otimes(\reify_{M,N,\Gamma} (\pi_1 \circ \sigma'),
\reify_{N,M,\Gamma} (\pi_2 \circ \sigma')) \\ & & \mbox{where $\sigma'
  = \dist_{-,\Gamma} \circ \llbracket \Gamma \rrbracket^-(\dec) \circ
  \sigma$
} \\
\reify_{A \oslash N , \Gamma}(\sigma) & = & \p_\oslash (\reify_{A,N,\Gamma} (\sigma)) \\
\reify_{A \lhd P , \Gamma}(\sigma) & = & \p_\lhd (\reify_{A,P,\Gamma} (\sigma)) \\

\reify_{\top}(\sigma) & = & \p_\top \\
\reify_{\top, P, \Gamma}(\sigma) & = & \p_\bot^- (\reify_{\bot , \Gamma}(\sigma \circ \llbracket \Gamma \rrbracket^+(\mathsf{abs}^{-1}))) \\
\reify_{\top , N}(\sigma) & = & \p_\top^- (\reify_N ((\_ \multimap o)^{-1}(\unit_\multimap^{-1} \circ \sigma))) \\
\reify_{\top, N , M , \Gamma}(\sigma) & = & \p_\top^\otimes (\reify_{\top , N \otimes M,\Gamma}(\sigma \circ \llbracket \Gamma \rrbracket^+ (\passoc_\multimap \circ (\sym \multimap \id)))) \\
\reify_{\top , N , P , \Gamma}(\sigma) & = & \p_\top^\lhd (\reify_{\top , N \lhd P, \Gamma}(\sigma \circ \llbracket \Gamma \rrbracket^+(\lfe^{-1}))) \\
\reify_{P \oplus Q , \Gamma}(\sigma) & = &
[ {\p_\oplus}_1 \circ \reify_{P,\Gamma} , {\p_\oplus}_2 \circ \reify_{Q,\Gamma} ] \circ \mathsf{d}^{-1}(\sigma \circ \dist_{+,\Gamma}^{-1}) \\
\reify_{P \parr Q, \Gamma}(\sigma) & = &
[ {\p_\parr}_1 \circ \reify_{P,Q,\Gamma}, {\p_\parr}_2 \circ \reify_{Q,P,\Gamma} ]\circ \mathsf{d}^{-1}(\sigma \circ \llbracket \Gamma \rrbracket^+ (\dec^{-1}) \circ \dist_{+,\Gamma}^{-1}) \\
\reify_{!N,\Gamma}(\sigma) & = & \p_!(\reify_{N,!N,\Gamma}(\llbracket \Gamma \rrbracket^-(\alpha) \circ \sigma)) \\
\reify_{?P,\Gamma}(\sigma) & = & \p_?(\reify_{P,?P,\Gamma}(\sigma \circ \llbracket \Gamma \rrbracket^+(\alpha^{-1}))) \\
\end{array}
\]

\hrule
\hrule
\end{small}
\end{figure*}

\subsection{Termination of Reify}

We next argue for termination of our procedure. Intuitively, the full completeness procedure first breaks down the head formula until it is $\bot$ or $\top$. It then uses the core elimination rules to compose the tail into (at most) a single formula. These steps do not increase the size of the strategy. Finally, the head is removed using $\p_\bot^+$ or $\p_\top^-$, strictly reducing the size of the strategy. If $\Theta$ is not lean, the number of distinct variable pairs that are not declared distinct in $\Theta$ is reduced by using $\p_\mathsf{ma}$.

Formally, we can see this as a lexicographical ordering of four measures on $\sigma$,$X$,$\Theta$,$\Gamma$:
\begin{itemize}
\item The most dominant measure is the length of the longest play in $\sigma$.
\item The second measure is the length of $\Gamma$ as a list if the head of $\Gamma$ is $\bot$ or $\top$, and $\infty$ otherwise.
\item The third measure is the size of the head formula of $\Gamma$.
\item The fourth measure is $$|\{ (x , y) \in X \times X : \\ x \not\equiv y \wedge x \neq y \notin \Theta \} |$$
\end{itemize}

If $\Theta$ is lean:
\begin{itemize}
\item If $\Gamma = \bot , P$ or $\top , N$ then the first measure decreases in the call to the inductive hypothesis.
\item Otherwise, if $\Gamma = A , \Gamma'$ with $A \in {\bot, \top}$ the first measure does not increase and the second measure decreases.
\item If $\Gamma = A , \Gamma'$ with $A \not \not \in \{ \bot , \top \}$, the first measure does not increase and either the second or third measure decreases.
\end{itemize}

If $\Theta$ is not lean and the $\p_\mathsf{ma}$ rule is applied, in the call to the inductive hypotheses the first three measures stay the same and the fourth measure decreases.

Thus, the inductive hypothesis is used with a smaller value in the compound measure on $\mathbb{N} \times \mathbb{N} \cup \{ \infty \} \times \mathbb{N} \times \mathbb{N}$ ordered lexicographically.

\subsection{Soundness and Uniqueness}

\begin{lemma}
  For all $\sigma : \llbracket \vdash \Gamma \rrbracket$ we have
  $\llbracket \reify_\Gamma (\sigma) \rrbracket = \sigma$.
  \label{refsound}
\end{lemma}

\begin{proof}
  We proceed by induction on our reification measure $\langle |\Gamma|
  , \mathsf{tl}(\Gamma) , \mathsf{hd}(\Gamma) \rangle$ using equations
  that hold in the categorical model. We perform case analysis on
  $\Gamma$. The calculation is routine, we demonstrate only a few
  cases.

\begin{itemize}
\item If $\Theta$ is not lean with $(x,y) \in X \times X$ least such that $x \not \equiv y$ and $(x \neq y) \not \in \Theta$ and $z$ is the least element in $\mathsf{Fr}(X;\Theta \vdash \Gamma)$, then $\llbracket \reify(\sigma) \rrbracket $\\$ =
\llbracket \p_\match^{x,y,z}(\reify(\sigma|_{\mathcal{M}_X^{\Theta,x = y}} \circ H^{-1}_{x,y,z} , \reify(\sigma|_{\mathcal{M}_X^{\Theta,x \neq y}}))) \rrbracket $\\$ =
[\llbracket \reify(\sigma|_{\mathcal{M}_X^{\Theta,x = y}} \circ H^{-1}_{x,y,z}) \rrbracket H_{x,y,z}, \llbracket \reify(\sigma|_{\mathcal{M}_X^{\Theta,x \neq y}}) \rrbracket ]_{x = y , x \neq y} $\\$ =
[\sigma|_{\mathcal{M}_X^{\Theta,x = y}} \circ H^{-1}_{x,y,z} \circ H_{x,y,z} , \sigma|_{\mathcal{M}_X^{\Theta,x \neq y}}]_{x = y , x \neq y} $\\$ = 
[\sigma|_{\mathcal{M}_X^{\Theta,x = y}} , \sigma|_{\mathcal{M}_X^{\Theta,x \neq y}}]_{x = y , x \neq y} =
\sigma$.
\item If $\Theta$ is lean and $\Gamma = \phi(\overrightarrow{x}) , \Gamma'$ then $$\llbracket \reify(\sigma) \rrbracket = \llbracket \p_\mathsf{at-}(\sigma|_{\mathcal{M}_X^{\Theta,\overline{\phi}(\overrightarrow{x})}}) \rrbracket = [ \sigma|_{\mathcal{M}_X^{\Theta,\overline{\phi}(\overrightarrow{x})}} , \epsilon]_{\overline{\phi}(\overrightarrow{x}),\phi(\overrightarrow{x})} = \sigma$$ as we must have $\sigma|_{\mathcal{M}_X^{\Theta,\phi(\overrightarrow{x})}} = \epsilon$ since $\llbracket \phi(\overrightarrow{x}) , \Gamma \rrbracket_A$ is the terminal object for each $A$ in $\mathcal{M}_X^{\Theta , \phi(\overrightarrow{x})}$.
\item If $\Gamma = \forall x . N , \Gamma'$ then $\llbracket \reify(\sigma) \rrbracket = \llbracket \p_\forall(\reify(\eta \circ U'_x(\dist_{\Gamma'} \circ \sigma))) \rrbracket = $\\$ \dist_{\Gamma'}^{-1} \circ \widehat{\llbracket \reify(
\eta \circ U'_x(\dist_{\Gamma'} \circ \sigma)) \rrbracket} = \dist_{\Gamma'}^{-1} \circ \widehat{
(\eta \circ U'_x(\dist_{\Gamma'} \circ \sigma))} = \dist_{\Gamma'}^{-1} \circ \dist_{\Gamma'} \circ \sigma = \sigma$ as required.
\item If $\Gamma = P_1 \parr P_2, \Delta$ then $\llbracket \reify_{P_1
    \parr P_2,\Delta}(\sigma) \rrbracket = \llbracket [ {\p_\parr}_1
  \circ \reify_{P_1,P_2,\Delta} , {\p_\parr}_2 \circ
  \reify_{P_2,P_1,\Delta} ] \circ \mathsf{d}^{-1}(\sigma \circ
  \llbracket \Delta \rrbracket^+(\dec^{-1}) \circ \dist_{+,\Delta}^{-1})
  \rrbracket$. Suppose $\mathsf{d}^{-1} (\sigma \circ \llbracket
  \Delta \rrbracket^+(\dec^{-1}) \circ \dist_{+,\Delta}^{-1}) =
  \ini(\tau)$, so $\tau \circ \pi_i = \sigma \circ \llbracket \Delta
  \rrbracket^+(\dec^{-1}) \circ \dist_{+,\Delta}^{-1}$.

  If $i = 1$ then $\llbracket [ {\p_\parr}_1 \circ
  \reify_{P_1,P_2,\Delta} , {\p_\parr}_2 \circ \reify_{P_2,P_1,\Delta} ]
  \circ \mathsf{d}^{-1}(\sigma \circ \llbracket \Delta
  \rrbracket^+(\dec) \circ \dist_{+,\Delta}^{-1}) \rrbracket =
  \llbracket {\p_\parr}_1 ( \reify_{P_1,P_2,\Delta}(\tau)) \rrbracket =
  \llbracket \reify_{P_1,P_2,\Delta}(\tau) \rrbracket \circ \llbracket
  \Delta \rrbracket^+(\wk) =$ \\$\llbracket \reify_{P_1,P_2,\Delta}(\tau)
  \rrbracket \circ \llbracket \Delta \rrbracket^+(\pi_1 \circ \dec) =
  \tau \circ \llbracket \Delta \rrbracket^+(\pi_1 \circ \dec) = \tau
  \circ \pi_1 \circ \dist_{+,\Delta} \circ \llbracket \Delta
  \rrbracket^+(\dec) = \sigma \circ \llbracket \Delta
  \rrbracket^+(\dec^{-1}) \circ \dist_{+,\Delta}^{-1} \circ
  \dist_{+,\Delta} \circ \llbracket \Delta \rrbracket^+(\dec) =
  \sigma$.

  The case for $i = 2$ is similar.  \qed
\end{itemize}
\end{proof}

\begin{lemma}
  For any analytic proof $p$ of $\vdash \Gamma$ we have $\reify_\Gamma
  (\llbracket p \rrbracket) = p$.
  \label{refunique}
\end{lemma}
\begin{proof}
  We proceed by induction on $p$.  The calculation is routine, we
  demonstrate only a few cases.

\begin{itemize}
\item If $p = \p_\bot^+(p')$ with $\Gamma = \bot , P$ then $\reify_\Gamma
  (\llbracket p \rrbracket) = \p_\bot^+ (\reify_P (\Lambda_I^{-1}
  (\llbracket p \rrbracket))) = $\\$ \p_\bot^+ (\reify_P
  (\Lambda_I^{-1}\Lambda_I\llbracket p' \rrbracket)) = \p_\bot^+(
  \reify_P(\llbracket p' \rrbracket)) = \p_\bot^+(p') = p$.
\item If $p = \p_\& (p_1, p_2)$ with $\Gamma = M \& N , \Delta$ then
  $\reify_\Gamma(\llbracket p \rrbracket) $\\$ = \p_\&
  (\reify_{M,\Delta}(\pi_1 \circ \dist_{-,\Delta} \circ \llbracket p
  \rrbracket), \reify_{N,\Delta}(\pi_2 \circ \dist_{-,\Delta} \circ
  \llbracket p \rrbracket)) $\\$ = \p_\& (\reify_{M,\Delta}(\pi_1 \circ
  \dist_{-,\Delta} \circ \dist_{-,\Delta}^{-1} \circ \langle
  \llbracket p_1 \rrbracket, \llbracket p_2 \rrbracket \rangle),
  \reify_{N,\Delta}(\pi_2 \circ \dist_{-,\Delta} \circ
  \dist_{-,\Delta}^{-1} \circ \langle \llbracket p_1 \rrbracket,
  \llbracket p_2 \rrbracket \rangle)) $\\$ = \p_\& 
  (\reify_{M,\Delta}(\llbracket p_1 \rrbracket),
  \reify_{N,\Delta}(\llbracket p_2 \rrbracket)) $\\$ = \p_\& (p_1, p_2) = p$.
\item If $p = {\p_\parr}_1 (p')$ with $\Gamma = P_1 \parr P_2 , \Delta$
  then $\reify_\Gamma(\llbracket p \rrbracket) =
  \reify_\Gamma(\llbracket p' \rrbracket \circ \llbracket \Delta
  \rrbracket^+(\wk)) = [{\p_\parr}_1 \circ \reify_{P_1,P_2,\Delta} ,
  {\p_\parr}_2 \circ \reify_{P_2,P_1,\Delta}] \circ
  \mathsf{d}^{-1}(\llbracket p' \rrbracket \circ \llbracket \Delta
  \rrbracket^+(\wk) \circ \llbracket \Delta \rrbracket^+(\dec^{-1})
  \circ \dist_{+,\Delta}^{-1}) = [{\p_\parr}_1 \circ
  \reify_{P_1,P_2,\Delta} , {\p_\parr}_2 \circ \reify_{P_2,P_1,\Delta}]
  \circ \mathsf{d}^{-1}(\llbracket p' \rrbracket \circ \llbracket
  \Delta \rrbracket^+(\pi_1) \circ \dist_{+,\Delta}^{-1}) =
  [{\p_\parr}_1 \circ \reify_{P_1,P_2,\Delta} , {\p_\parr}_2 \circ
  \reify_{P_2,P_1,\Delta}] \circ \mathsf{d}^{-1}(\llbracket p'
  \rrbracket \circ \pi_1) = [{\p_\parr}_1 \circ \reify_{P_1,P_2,\Delta}
  , {\p_\parr}_2 \circ \reify_{P_2,P_1,\Delta}] \circ
  \mathsf{d}^{-1}(\mathsf{d}(\inn_1(\llbracket p' \rrbracket))) =
  {\p_\parr}_1(\reify_{P_1,P_2,\Delta}(\llbracket p' \rrbracket)) =
  {\p_\parr}_1(p') = p$ as required. \qed
\end{itemize}
\end{proof}

\noindent This completes our proof of Theorem \ref{fullcomp}.

\subsection{Infinitary Analytic Proofs}

\label{ws!-infcore}

We have seen that any bounded winning strategy is the denotation of a
unique analytic proof of \textsf{WS1}. We cannot use this to normalise
proofs to their analytic form 
because proofs do not necessarily denote bounded strategies. We will
next show that our reification procedure can be extended to winning
strategies that may be unbounded, provided the resulting analytic proofs
are allowed to be \emph{infinitary} --- that is, proofs using the core
rules that may be infinitely deep. More precisely, we will show that
\emph{total} strategies on a type object correspond precisely to the
infinitary analytic proofs. Thus we can normalise any proof of
\textsf{WS1} to an infinitary normal form, by taking its semantics
and then constructing the corresponding infinitary analytic proof. Two
proofs of \textsf{WS1} are semantically equivalent if and only if they
have the same normal form as an infinitary analytic proof.

\subsubsection{Infinitary Proofs as a Final Coalgebra}

Let $L$ be a set. Let $\mathcal{T}_L$ denote the final coalgebra of
the functor $X \mapsto L \times X^\ast$ in \textbf{Set}. The
inhabitants of $T_L$ are $L$-labelled trees of potentially infinite
depth. We let $\alpha : \mathcal{T}_L \rightarrow L \times
\mathcal{T}_L^\ast$ describe the arrow part of this final coalgebra:
this maps a tree to its label and sequence of subtrees. Given a
natural number $n$, we define a function $N_n : \mathcal{T}_L
\rightarrow \mathcal{P}(L \times \mathcal{T}_L^\ast)$, by induction:
$N_0(T) = \emptyset$ and $N_{n+1}(T) = \{ \alpha(T) \} \uplus \bigcup
\{ N_n(T') : T' \in \pi_2(\alpha(T)) \}$. We define the set of nodes
$N(T)$ to be $\{ N_n(T) : n \in \mathbb{N} \}$. Let $\prf$ be the set
of (names of) proof rules of \textsf{WS1} and $\seqo$ the set of
sequents of \textsf{WS1}.

\begin{definition}
An \emph{infinitary analytic proof} of \textsf{WS1} is an infinitary proof using only the core rules of \textsf{WS1}. Formally, this is an element $T$ of $\mathcal{I} = \mathcal{T}_{\prf \times \seqo}$ such that for each node $((\p_x,X;\Theta\vdash \Gamma),c) \in N(T)$ we have $|c| = \mathsf{ar}(\p_x)$ and if $(\pi_2 \circ \pi_1 \circ \alpha)(c_i) = X_i;\Theta_i \vdash \Gamma_i$ then the following is a valid core rule of \textsf{WS1}:

\begin{prooftree}
\AxiomC{$X_1;\Theta_1 \vdash \Gamma_1$}
\AxiomC{$\ldots$}
\AxiomC{$X_{|c|};\Theta_{|c|} \vdash \Gamma_{|c|}$}
\LeftLabel{$\p_x$}
\TrinaryInfC{$X;\Theta \vdash \Gamma$}
\end{prooftree}

\noindent We let $\mathcal{I}_\Gamma$ denote the set of infinitary analytic proofs of $\vdash \Gamma$.
\end{definition}

Let $\{ A_{X;\Theta \vdash \Gamma} : X;\Theta \vdash \Gamma \in \seqo
\}$ be family of sets indexed by sequents. We can construct a family
of maps $A_{X;\Theta \vdash \Gamma} \rightarrow \mathcal{I}_{X;\Theta
  \vdash \Gamma}$ by giving, for each $X;\Theta \vdash \Gamma$ and $a
\in A_{X;\Theta \vdash \Gamma}$, a proof rule that concludes $X;\Theta
\vdash \Gamma$ from $X_1;\Theta_1 \vdash \Gamma_1$, \ldots ,
$X_n;\Theta_n \vdash \Gamma_n$ and for each $i$ an element $a_i \in
A_{X_i;\Theta_i \vdash \Gamma_i}$.

\subsubsection{Infinitary Proofs as a Limit of Paraproofs}

We can consider an alternative approach for presenting our infinitary analytic proofs. We consider partial proofs, that may ``give up'' in the style of \cite{Gir_LS}.

\begin{definition}
An \emph{analytic paraproof} of \textsf{WS1} is a proof made up of the core proof rules of \textsf{WS1}, together with a \emph{d\ae mon} rule that can prove any sequent:
\begin{center}{\Large $\over \Phi \vdash \Gamma$}\mbox{ $\p_\epsilon$}
\end{center}
\end{definition}

Note that each analytic proof is also an analytic paraproof. Let $\mathcal{C}_\Gamma$ represent the set of analytic paraproofs of $\vdash \Gamma$. We can introduce an ordering $\sqsubseteq$ on this set, generated from the least congruence with $\p_\epsilon$ as a bottom element. We can take the completion of $\mathcal{C}_\Gamma$ with respect to $\omega$-chains generating an algebraic cpo $\mathcal{D}_\Gamma$. The maximal elements in this domain are precisely the infinitary analytic proofs $\mathcal{I}_\Gamma$, and the compact elements are the analytic paraproofs $\mathcal{C}_\Gamma$. 

\subsubsection{Semantics of Infinitary Analytic Proofs}

We next describe semantics of infinitary analytic proofs via the semantics of analytic paraproofs.

We can interpret analytic paraproofs as partial strategies. We
interpret paraproofs of $X;\Theta \vdash \Gamma$ in
$\mathcal{G}^{\mathcal{M}_X^\Theta}$. For the rules other than
$\p_\epsilon$, we use the fact that
$\mathcal{G}^{\mathcal{M}_X^\Theta}$ is a WS!-category. We interpret
$\p_\epsilon$ as the strategy $\{ \epsilon \}$ where $\epsilon$
denotes the empty play on any game.  We can hence interpret a analytic
paraproof of $\vdash \Gamma$ as a strategy on $\llbracket \vdash
\Gamma \rrbracket$.

The category $\mathcal{G}^{\mathcal{M}_X^\Theta}$ is cpo-enriched,
with $\sigma \sqsubseteq \tau$ if for each $A$, $\sigma_A \subseteq
\tau_A$ as a set of plays. The bottom element is the uniform strategy
that is $\{ \epsilon \}$ at each component. Composition, pairing and
currying are continuous maps of hom sets; as are the operations used
in the first-order structure.

\begin{proposition}
If $p$ and $q$ are analytic paraproofs of $\vdash \Gamma$ and $p \sqsubseteq q$ then $\llbracket p \rrbracket \sqsubseteq \llbracket q \rrbracket$.
\end{proposition}
\begin{proof}
A simple induction on the proof rules for \textsf{WS1}, using the fact that composition, pairing and currying are monotonic operations. Note that $\llbracket - \rrbracket$ is also strict, as $\llbracket \p_\epsilon \rrbracket = \{ \epsilon \}$. \qed
\end{proof}


Hom sets of $\mathcal{G}^{\mathcal{M}_X^\Theta}$ are algebraic
domains: each strategy is the limit of its compact (finite)
approximants. Our monotonic map $\mathcal{C}_\Gamma \rightarrow
\llbracket X;\Theta \vdash \Gamma \rrbracket$ thus extends uniquely to
a continuous map $\mathcal{D}_\Gamma \rightarrow \llbracket X;\Theta
\vdash \Gamma \rrbracket$. By construction this agrees with the
semantics given above for analytic paraproofs in
$\mathcal{D}_\Gamma$. Given any infinitary analytic proof $p$ if
$p\downarrow$ is the set of analytic paraproofs less than $p$ then
$\llbracket p \rrbracket = \bigsqcup \llbracket p\downarrow
\rrbracket$ using the cpo structure in $\mathcal{G}^{\mathcal{M}_X^\Theta}$.

We can show that this really does capture the intended semantics of infinitary analytic proofs.

\begin{proposition}
  The equations for the semantics of analytic proofs given in Figures
  \ref{WS-sem}, \ref{WS!-sem} and \ref{WS1-sem} hold for infinitary
  analytic proofs.
\label{infsemantics}
\end{proposition}
\begin{proof}
We use the fact that the constructs used in the semantics of the core proof rules are continuous. We proceed by case analysis on the proof rule.

We just give an example. In the case of $\p_\otimes$, note that $\llbracket \p_\otimes(p,q) \rrbracket = \bigsqcup \{ \llbracket r \rrbracket : r \sqsubseteq \p_\otimes(p,q) \} = \bigsqcup \{ \llbracket \p_\otimes(p', q') \rrbracket : p' \sqsubseteq p \wedge q' \sqsubseteq q \} = \bigsqcup \{ \llbracket \Gamma \rrbracket^-(\dec^{-1}) \circ \dist^{-1}_{-,\Gamma} \circ \langle \llbracket p' \rrbracket , \llbracket q' \rrbracket \rangle : p' \sqsubseteq p \wedge q' \sqsubseteq q \} = \llbracket \Gamma \rrbracket^-(\dec^{-1}) \circ \dist^{-1}_{-,\Gamma} \circ \langle \llbracket \bigsqcup \{ p' : p' \sqsubseteq p \} \rrbracket , \llbracket \bigsqcup \{ q' : q' \sqsubseteq q \} \rrbracket \rangle = \llbracket \Gamma \rrbracket^-(\dec^{-1}) \circ \dist^{-1}_{-,\Gamma} \circ \langle \llbracket p \rrbracket , \llbracket q \rrbracket \rangle$ as required. All other cases are similar. \qed
\end{proof}


\subsubsection{Totality}

We need to show that given $p \in \mathcal{I}_\Gamma$, $\llbracket p \rrbracket$ is a total uniform strategy. Note that this is not true of arbitrary paraproofs in $\mathcal{D}_\Gamma$, nor is it true for infinite derivations in full \textsf{WS1} (for example, one could repeatedly apply the $\p_\sym$ rules forever).

To show this fact, we first introduce some auxiliary notions.

\begin{definition}
  Let $\sigma : N$ be a strategy on a negative game. We say that
  $\sigma$ is $n$-total if whenever $s \in \sigma \wedge |s| \leq n
  \wedge so \in P_N \Rightarrow \exists p . sop \in \sigma$. A uniform
  strategy is $n$-total if it is pointwise $n$-total.
\end{definition}

\noindent It is clear that a strategy is total if and only if it is
$n$-total for each $n$.

\begin{proposition}
\label{ntotal}
The following  hold:
\begin{enumerate}
\item If $\sigma$ is $n$-total and $\tau$ is an isomorphism then $\tau
  \circ \sigma$ is $n$-total. If $\sigma$ is $n$-total and $\tau$ is
  an isomorphism then $\sigma \circ \tau$ is $n$-total.
\item If $\sigma : A \rightarrow B$ and $\tau : A \rightarrow C$ are
  $n$-total then $\langle \sigma, \tau \rangle$ is also $n$-total. If
  $\sigma : A_i \multimap B$ is $n$-total then $\sigma \circ \pi_i :
  A_1 \times A_2 \multimap B$ is $n$-total.
\item If $\sigma : A \otimes B \multimap C$ is $n$-total then
  $\Lambda(\sigma)$ is $n$-total. If $\sigma: A \multimap B$ is
  $n$-total then $\sigma \multimap \id : (B \multimap o) \multimap (A
  \multimap o)$ is $(n+2)$-total.
\item If $\sigma$ and $\tau$ are $n$-total, then so is $[ \sigma ,
  \tau ]_{\mathcal{C},\mathcal{D}}$, $\sigma \circ H$. If $\sigma$ is
  $n$-total, then so is $\hat{\sigma}$.
\end{enumerate}
\end{proposition}
\begin{proof}
Simple verification. \qed
\end{proof}

\begin{proposition}
Given any infinitary analytic proof $p$ of $X ; \Theta \vdash \Gamma$, $\llbracket p \rrbracket$ is total.
\label{inftotal}
\end{proposition}
\begin{proof}
We show that $\llbracket p \rrbracket$ is $n$-total for each $n$. We proceed by induction on a compound measure.
\begin{itemize}
\item Define $\mathsf{tl}^+(A,\Gamma)$ to be the length of $\Gamma$ as a list if $A = \top$ or $\infty$ otherwise.
\item Define $\mathsf{hd}^+(A,\Gamma)$ to be $|A|$ if $A$ is positive or $\infty$ otherwise.
\item Define $\mathsf{tl}^-(A,\Gamma)$ to be the length of $\Gamma$ as a list if $A = \bot$ or $\infty$ otherwise.
\item Define $\mathsf{hd}^-(A,\Gamma)$ to be $|A|$ if $A$ is negative or $\infty$ otherwise.
\end{itemize}
We proceed by induction on $$f(n,X,\Theta,\Gamma) = \langle n ,
\mathsf{tl}^+(\Gamma), \mathsf{hd}^+(\Gamma), \mathsf{tl}^-(\Gamma),
\mathsf{hd}^-(\Gamma), \mathsf{L}(X,\Theta) \rangle.$$ We proceed by
case analysis on $p$. If $p = \p_\otimes(p_1,p_2)$ then $\llbracket
\p_\otimes(p_1,q_2) \rrbracket = \llbracket \Gamma
\rrbracket^-(\dec^{-1}) \circ \dist^{-1}_{-,\Gamma} \circ \langle
\llbracket p_1 \rrbracket , \llbracket p_2 \rrbracket \rangle$. By by
Proposition \ref{ntotal} $\llbracket \p_\otimes(p,q) \rrbracket =
\llbracket p \rrbracket$ is $n$-total.
The remaining cases work in an entirely analogous way. For $\p_\bot^+$ we must use the fact that currying is continuous and preserves $n$-totality. For termination:
\begin{itemize}
\item If $\Theta$ is not lean, in the call to the inductive hypothesis
  the first five measures do not increase, and the fifth measure
  $\mathsf{L}$ decreases.
\item In the case of $\p_\otimes$, $\p_\&$, $\p_!$ the first three measures ($n , \mathsf{tl}^+(\Gamma), \mathsf{hd}^+(\Gamma)$) stay the same and either the fourth measure $\mathsf{tl}^-(\Gamma)$ decreases, or the fourth measure stays the same and the fifth measure $\mathsf{hd}^-(\Gamma)$ decreases.
\item In the case of $\p_\bot^\parr$, $\p_\bot^\oslash$, $\p_\bot^-$ the first three measures stay the same and the fourth measure decreases.
\item In the cases of $\p_\bot^+$, $\p_\parr$, $\p_\oplus$, $\p_?$ the first measure $n$ stays the same and either the second measure $\tl^+(\Gamma)$ decreases, or the second measure stays the same and the third measure $\hd^+(\Gamma)$ decreases.
\item In the case of $\p_\top^\otimes$, $\p_\top^\lhd$, $\p_\top^+$ the first measure stays the same and the second measure decreases.
\item In the case of $\p_\top^-$, the first measure decreases. In particular, $\llbracket \p_\top^+(q) \rrbracket = \unit_\multimap \circ (\llbracket q \rrbracket \multimap \id)$. By induction $\llbracket q \rrbracket$ is $(n-2)$-total, and so $\llbracket q \rrbracket \multimap \id$ is $n$-total, and so $\llbracket p \rrbracket$ is $n$-total by Proposition \ref{ntotal}. \qed
\end{itemize}
\end{proof}

\noindent Note that there are infinitary analytic proofs that denote strategies that are total, but not winning. For example, there is an infinitary analytic proof of $\vdash \bot , ?(\top \lhd \bot)$ given by $\p_\bot^+(h)$ where $h$ is the infinitary analytic proof of $\vdash ?(\top \lhd \bot)$ given by $h = \p_?(\p_\lhd(\p_\top^\lhd(\p_\top^-(\p_\lhd(\p_\bot^+(h))))))$. But there are no winning strategies on this game.

\subsubsection{Reification of Total Strategies as Infinitary Analytic Proofs}

We next show that any total strategy $\sigma$ on the denotation of a sequent is the interpretation of a unique infinitary analytic proof $\reify(\sigma)$.

We first define $\reify$ for winning strategies. We have seen that we can construct a family of maps $A_{X;\Theta \vdash \Gamma} \rightarrow \mathcal{I}_{X;\Theta \vdash \Gamma}$ by giving, for each $X;\Theta \vdash \Gamma$ and $a \in A_{X;\Theta \vdash \Gamma}$, a proof rule that concludes $X;\Theta \vdash \Gamma$ from $X_1;\Theta_1 \vdash \Gamma_1$, \ldots , $X_n, \Theta_n \vdash \Gamma_n$ and for each $i$ an element $a_i \in A_{X_i;\Theta_i \vdash \Gamma_i}$.
\begin{diagram}
\sum_{X;\Theta \vdash \Gamma \in \seqo} A_\Gamma & \rTo^f & (\prf \times \seqo) \times (\sum_{X;\Theta \vdash \Gamma \in \seqo} A_{X;\Theta \vdash \Gamma})^\ast \\
\dTo^{\leftmoon f \rightmoon} & & \dTo_{\id \times \leftmoon f \rightmoon^\ast} \\
\mathcal{I} & \rTo_{\alpha} & (\prf \times \seqo) \times \mathcal{I}^\ast \\
\end{diagram}

Note that our reification function $\reify$ defined in Figure
\ref{WS-reify} is exactly of this shape. In this case $A_{X;\Theta
  \vdash \Gamma}$ is the set of uniform winning strategies on
$\llbracket X;\Theta \vdash \Gamma \rrbracket$. The function
specifies, for each strategy, the root-level proof rule and the
derived strategies that are given as input to $\reify$ coinductively.
In the case that $\sigma$ is bounded, we have seen that the process terminates and $\reify(\sigma)$ is a finite proof.

In fact, we note that this family of maps are still well defined if
$A_{X;\Theta \vdash \Gamma}$ is the set $\tot_{X;\Theta \vdash
  \Gamma}$ of uniform \emph{total} strategies on $\llbracket X ;
\Theta \vdash \Gamma
\rrbracket$. 
In particular, the composition of a total strategy and an isomorphism
is a total strategy; the composition of a total strategy and a
projection is a total strategy; and the completeness axioms in Section
\ref{compaxioms} hold with respect to total strategies. This procedure
provides, for each total strategy on $X;\Theta \vdash \Gamma$, a proof
rule $\p_x$ concluding $X;\Theta \vdash \Gamma$ from $X_1;\Theta_1
\vdash \Gamma_1, \ldots , X_n;\Theta_n \vdash \Gamma_n$ and
total strategies on each $\llbracket X_i;\Theta_i \vdash \Gamma_i
\rrbracket$. We write this map as $\mathsf{reif}_{X;\Theta \vdash
  \Gamma}$.
\begin{diagram}
\sum_{X;\Theta \vdash \Gamma \in \seqo} \tot_{X;\Theta \vdash \Gamma} & \rTo^{\mathsf{reif}} & (\prf \times \seqo) \times (\sum_{X; \Theta \vdash \Gamma \in \seqo} \tot_{X;\Theta \vdash \Gamma})^\ast \\
\dTo^{\reify = \leftmoon \mathsf{reif} \rightmoon} & & \dTo_{\id \times \reify^\ast} \\
\mathcal{I} & \rTo_{\alpha} & (\prf \times \seqo) \times \mathcal{I}^\ast \\
\end{diagram}

Thus we can take the anamorphism of this map yielding a map from total strategies on $\llbracket X;\Theta \vdash \Gamma \rrbracket$ to $\mathcal{I}_{X;\Theta \vdash \Gamma}$, as required.



\subsubsection{Soundness and Uniqueness}

We can show that given any winning strategy $\sigma$, $\reify(\sigma)$ is the unique infinitary analytic proof $p$ such that $\llbracket \reify(p) \rrbracket = \sigma$.

For soundness, we first introduce some auxiliary notions.

\begin{definition}
Let $\sigma$ and $\tau$ be strategies on $A$. We say that $\sigma =_n \tau$ if each play in $\sigma$ of length at most $n$ is in $\tau$, and each play in $\tau$ of length at most $n$ is in $\sigma$.
\end{definition}

\noindent It is clear that $=_n$ is an equivalence relation, and
$\sigma = \tau$ if and only if $\sigma =_n \tau$ for each $n \in
\mathbb{N}$. We can lift the relation $=_n$ to uniform total
strategies pointwise.

\begin{proposition}
\begin{enumerate}
\item If $\sigma =_n \tau$ and $\rho$ is an isomorphism then $\sigma \circ \rho =_n \tau \circ \rho$. If $\sigma =_n \tau$ and $\rho$ is an isomorphism then $\rho \circ \sigma =_n \rho \circ \tau$.
\item If $\sigma =_n \tau$ and $\rho =_n \delta$ then $\langle \sigma, \rho \rangle =_n \langle \tau , \delta \rangle$. If $\sigma =_n \tau$ then $\sigma \circ \pi_i =_n \tau \circ \pi_i$.
\item If $\sigma =_n \tau$ then $\Lambda(\sigma) =_n \Lambda(\tau)$. If $\sigma =_n \tau$ then $\sigma \multimap \id =_{n+2} \tau \multimap \id$.
 \item If $\sigma_1 =_n \sigma_2$ and $\tau_1 =_n \tau_2$ then $[ \sigma_1 , \tau_1 ]_{\mathcal{C},\mathcal{D}} = [ \sigma_2 , \tau_2 ]_{\mathcal{C},\mathcal{D}}$. If $\sigma_1 =_n \sigma_2$ then $\sigma_1 \circ H =_n \sigma_2 \circ H$. If $\sigma_1 =_n \sigma_2$ then $\hat{\sigma_1} =_n \hat{\sigma_2}$.
\end{enumerate}
\label{equaln}
\end{proposition}

\begin{proof}
Simple verification. \qed
\end{proof}

\begin{proposition}
  For every uniform total strategy $\sigma : \llbracket \vdash
  \Gamma \rrbracket$, $\llbracket \reify(\sigma) \rrbracket =
  \sigma$. \label{refsound-ws!-inf}
\end{proposition}
\begin{proof} We show that for each $n$, $\llbracket \reify(\sigma)
  \rrbracket =_n \sigma$. The structure of the induction follows that
  of Proposition \ref{inftotal}, lexicographically on $$\langle n ,
  \mathsf{tl}^+(\Gamma), \mathsf{hd}^+(\Gamma), \mathsf{tl}^-(\Gamma),
  \mathsf{hd}^-(\Gamma), \mathsf{L}(X,\Theta) \rangle.$$ In each
  particular case, the reasoning follows the proof of Proposition
  \ref{refsound} using $=_n$ in the inductive hypothesis rather than
  $=$, and propagating this to the main equation using Proposition
  \ref{equaln}. In the case of $\Gamma = \top , N$ we use the
  inductive hypothesis with a smaller $n$, using the final clause in
  Proposition \ref{equaln}. \qed
\end{proof}

\begin{proposition}
Given any infinitary analytic proof $p$, $\reify(\llbracket p \rrbracket) = p$.
\label{refunique-ws!-inf}
\end{proposition}
\begin{proof} Since $\id = \leftmoon \alpha \rightmoon$, we know that $\id$ is the unique morphism $f$ such that:
\begin{diagram}
\mathcal{I}_\Gamma & \rTo^{\alpha} & (\prf \times \seqo) \times \mathcal{I}^\ast \\
\dTo^{f} & & \dTo_{\id \times f^\ast} \\
\mathcal{I}_\Gamma & \rTo^{\alpha} & (\prf \times \seqo) \times \mathcal{I}^\ast \\
\end{diagram}

Thus to show that $\reify \circ \llbracket - \rrbracket = \id$ it is sufficient to show that $\alpha \circ \reify \circ \llbracket - \rrbracket = \id \times (\reify \circ \llbracket - \rrbracket)^\ast \circ \alpha$, i.e. that for each infinitary analytic proof $p$ we have $\alpha(\reify(\llbracket p \rrbracket)) = (\id \times (\reify \circ \llbracket - \rrbracket)^\ast)(\alpha(p))$.
\begin{itemize}
\item For binary rules $\p_x$ we must show that $$\reify(\llbracket \p_x(p_1,p_2) \rrbracket) = \p_x(\reify(\llbracket p_1 \rrbracket), \reify(\llbracket p_2 \rrbracket)).$$
\item For unary rules $\p_x$ we must show that $\reify(\llbracket \p_x(p)) \rrbracket = \p_x(\reify(\llbracket p \rrbracket))$.
\item For nullary rules $\p_x$ we must show that $\reify(\llbracket \p_x \rrbracket) = \p_x$.
\end{itemize}
For each proof rule, we have already shown this in the proof of Proposition \ref{refunique}. Proposition \ref{infsemantics} ensures that the proof applies in this setting. \qed
\end{proof}

\subsubsection{Full Completeness and Normalisation}

We have thus shown:

\begin{theorem}
  Each total strategy $\sigma$ on $\vdash \Gamma$ is the denotation of
  a unique infinitary analytic proof $\reify(\sigma)$.
\label{infrefok}
\end{theorem}

\noindent We hence have a bijection between infinitary analytic proofs
of a formula, and total strategies on the denotation of that formula,
via the semantics. Since any proof in \textsf{WS1} can be given
semantics as a winning strategy, and winning strategies are total, we
may $\reify$ the semantics of a \textsf{WS1} proof to generate its
infinitary normal form $\reify(\llbracket p \rrbracket)$.

\begin{theorem}
  For each \textsf{WS1} proof $p$, there is a unique infinitary
  analytic proof $q$ such that $\llbracket p \rrbracket = \llbracket q
  \rrbracket$.
\label{proofnorm}
\end{theorem}
\begin{proof}
  Let $q = \reify(\llbracket p \rrbracket)$. Then $\llbracket q
  \rrbracket = \llbracket \reify(\llbracket p \rrbracket) \rrbracket = \llbracket
  p \rrbracket$ by Proposition \ref{refsound-ws!-inf}. If $q'$ is an
  infinitary analytic proof with $\llbracket q' \rrbracket =
  \llbracket p \rrbracket$ then $\llbracket q' \rrbracket = \llbracket
  q \rrbracket$ and so
  $\reify(\llbracket q' \rrbracket) = \reify(\llbracket q \rrbracket)$
  and Proposition \ref{refunique-ws!-inf} ensures that $q' = q$. \qed
\end{proof}

\noindent While infinitary analytic proofs may denote strategies that
are not winning, any infinitary analytic proof generated as a result
of the above normalisation denotes a winning strategy. The above
result also ensures that proofs $p_1$ and $p_2$ in \textsf{WS1} denote
the same strategy if and only if their normal forms (as infinitary
analytic proofs) are identical.

\section{Further Directions}




In this paper, we have given some simple examples of ``stateful  
proofs''. We aim to investigate further examples in more expressive  
logics, and to specify additional properties of programs
in more powerful programming languages (such as the  
games-based language in  e.g. \cite{Long_PLGM}). Further extensions  
to our work which may be required  in order to do so include:
\begin{itemize}
\item \textsf{WS1} has been presented as a general first-order logic.  
By adding  axioms, we may specify and study programs in particular  
domains. For example, can we derive a version of  Peano Arithmetic in  
which proofs have constructive, stateful content (cf \cite{coq95})?
\item Extension with \emph{propositional variables} (and potentially,  
second-order quantification) would allow generic  ``copycat  
strategies'' to be captured. On the programming side, this would allow  
us to model languages with polymorphism.
\item We have interpreted the exponentials as greatest fixpoints.  
Adding general inductive and coinductive  types, as in $\mu L J$  
\cite{Cla_FIX} would extend \textsf{WS1} to a rich collection of  
datatypes (including finite and infinite lists, for example).
\end{itemize}

\paragraph{Acknolwedgements} The authors would like to thank
Pierre-Louis Curien, Alessio Guglielmi, Pierre Clairambault and
anonymous reviewers for earlier comments on this work. This work was
supported by the (UK) EPSRC grant EP/HO23097.

\bibliographystyle{model2-names}
\bibliography{mybib}

\end{document}